\documentclass[11pt]{article}

\usepackage{etoolbox}
\usepackage[margin=1in]{geometry}
\usepackage{graphicx}
\usepackage{xcolor}
\usepackage{mathtools}
\usepackage{amssymb}
\usepackage{amsthm}
\usepackage{amsfonts}
\usepackage{bm}
\usepackage{mathpazo}
\usepackage{newpxtext}
\usepackage[T1]{fontenc}
\usepackage[utf8]{inputenc}
\usepackage{microtype}
\usepackage{framed}
\usepackage[pagebackref]{hyperref}
\usepackage{url}
\usepackage[capitalize]{cleveref}
\usepackage[bf]{caption}
\usepackage{booktabs}
\usepackage{titling}
\usepackage{multicol}
\usepackage{array}
\usepackage{tikz}
\usetikzlibrary{arrows.meta}

\definecolor{citegreen}{HTML}{208054}
\definecolor{citeblue}{HTML}{0055cc}
\hypersetup{
colorlinks=true,
citecolor=citegreen,
linkcolor=citeblue,
linktocpage=true
}
\renewcommand*{\backref}[1]{}
\renewcommand*{\backrefalt}[4]{%
  \ifcase #1
  \or
    (p.~#2)%
  \else
    (pp.~#2)%
  \fi
}

\makeatletter
\newcommand{\unlinkedfootnotemark}{%
  \stepcounter{footnote}%
  \protected@xdef\@thefnmark{\thefootnote}%
  \H@@footnotemark
}

\DeclareRobustCommand\thanks[1]{%
  \unlinkedfootnotemark
  \protected@xdef\@thanks{\@thanks
    \protect\footnotetext[\the\c@footnote]{#1}}%
}
\makeatother

\numberwithin{equation}{section}

\theoremstyle{plain}
\newtheorem{theorem}{Theorem}[section]
\newtheorem{lemma}[theorem]{Lemma}
\newtheorem{corollary}[theorem]{Corollary}
\newtheorem{proposition}[theorem]{Proposition}

\newtheorem{remark}[theorem]{Remark}

\DeclareMathOperator{\Reop}{Re}

\theoremstyle{plain}
\newtheorem{definition}[theorem]{Definition}
\newtheorem{problem}[theorem]{Problem}

\newcommand{\bq}[0]{\mathbf{q}}
\newcommand{\bel}[0]{\mathbf{e}_\ell}
\newcommand{\cK}[0]{\mathcal{K}}

\newcommand{\cO}[0]{\mathcal{O}}
\newcommand{\cB}[0]{\mathcal{B}}

\renewcommand{\O}{\mathcal{O}}
\newcommand{\Ot}{\widetilde{\mathcal{O}}}

\DeclareMathOperator{\poly}{poly}
\DeclareMathOperator{\polylog}{polylog}
\newcommand{\R}{\mathbb{R}}
\newcommand{\C}{\mathbb{C}}
\newcommand{\N}{\mathbb{N}}

\newcommand{\I}{\mathbb{I}}
\renewcommand{\Re}{\operatorname{Re}}

\newcommand{\U}{\mathrm{U}}

\newcommand{\cG}{\mathcal{G}}
\newcommand{\cH}{\mathcal{H}}

\DeclareMathOperator{\spn}{span}
\DeclareMathOperator{\rank}{rank}
\DeclareMathOperator*{\E}{\mathbb{E}}

\DeclareMathOperator{\sgn}{sgn}

\DeclareMathOperator{\tr}{\operatorname{Tr}}

\DeclareMathOperator{\ad}{ad}
\newcommand{\op}[2]{\ket{#1}\!\bra{#2}}

\newcommand{\melem}[3]{\langle #1 | #2 | #3 \rangle}
\newcommand{\ev}[2]{\melem{#2}{#1}{#2}}
\newcommand{\ket}[1]{| #1 \rangle}
\newcommand{\bra}[1]{\langle #1 |}

\newcommand{\T}{\mathrm{T}}

\renewcommand{\l}[1]{\mathopen{}\left#1}
\renewcommand{\r}[1]{\right#1\mathclose{}}

\newcommand{\BPP}{\mathsf{BPP}}
\newcommand{\BQP}{\mathsf{BQP}}

\newcommand{\DQC}{\mathsf{DQC}}

\newcommand{\anc}{\mathsf{a}}
\newcommand{\norm}[1]{\|#1\|}
\DeclareMathOperator{\TO}{\mathcal{T}}
\newcommand{\BE}{\mathrm{BE}}

\newcommand{\HAMT}{\mathrm{HAM\text{-}T}}
\newcommand{\SWAP}{\mathrm{SWAP}}
\newcommand{\fSWAP}{\mathrm{fSWAP}}
\newcommand{\CZ}{\mathrm{CZ}}
\newcommand{\parity}{\mathsf{Par}}

\newcommand{\cut}[0]{\tau}

\newcommand{\bS}[0]{\mathbb{S}}

\newcommand{\br}[0]{\mathbf{r}}

\newcommand{\es}[0]{\mathsf{J}}

\usepackage{scalerel}

\newcommand{\dE}[0]{\scaleobj{1.1}{\Phi}}

\title{Polynomial-time classical and quantum simulation of\\ quantum impurity models}

\author{
Jiaqing Jiang\thanks{UC Berkeley, \texttt{\{jjiang5, nju\}@berkeley.edu}}
\qquad
Nathan Ju\thanksmark{1}
\qquad
Ojas Parekh\thanks{Sandia National Laboratories, \texttt{\{odparek, azhao\}@sandia.gov}}
\qquad
Chaithanya Rayudu\thanks{University of Cambridge, University of New Mexico, \texttt{sscr2@cam.ac.uk}}
\qquad
Andrew Zhao\thanksmark{2}
}

\date{October 4, 2026}

\begin{document}

\maketitle

\begin{abstract}
    Quantum impurity models are paradigmatic models of interacting quantum matter, as well as key computational primitives for modern electronic-structure methods. They describe a small subsystem of interacting fermions coupled to a large, noninteracting bath. We perform a comprehensive study of the computational complexity of simulating impurity models, delineating the boundary between classical and quantum tractability for this class of problems. Our main finding is that static properties of quantum impurity models can be calculated efficiently on a classical computer. Specifically, we give classical algorithms that (1) estimate the ground-state energy to additive precision $\delta$ in time $\mathrm{poly}(n,\delta^{-1})$, and (2) estimate the partition function at inverse temperature $\beta$ to relative precision $\delta$ in time $\mathrm{poly}(n,\beta,\delta^{-1})$, where $n$ is the system size. These results improve the previous best-known complexity for ground-state energy estimation from quasipolynomial to polynomial time, while establishing for the first time rigorous polynomial-time guarantees for simulating impurity models in thermal equilibrium. On the other hand, we find that simulating dynamical properties of impurity models is hard for classical computers but easy on a quantum computer. As a canonical example, we show that computing their nonequilibrium Green's functions captures the full power of quantum computation, even at finite temperature. Taken together, our results rule out superpolynomial quantum speedups for computing static properties, but provide an avenue for quantum advantage in simulating impurity physics out of equilibrium.
\end{abstract}


\tableofcontents

\clearpage

\section{Introduction}

Quantum impurity models describe a small, strongly interacting fermionic subsystem, called an impurity, coupled to a large noninteracting bath. Historically, Anderson introduced his impurity model to describe localized magnetic moments in metals~\cite{anderson1961localized}. Kondo subsequently used the impurity paradigm to explain the anomalous low-temperature resistance of dilute magnetic alloys~\cite{kondo1964resistance}, and Wilson later gave a nonperturbative solution to the Kondo problem using the numerical renormalization group (NRG)~\cite{wilson1975renormalization}. Impurity models have since found applications ranging from tunneling spectroscopy of single magnetic atoms on metallic surfaces~\cite{madhavan1998tunneling}, to transport through mesoscopic devices such as semiconductor quantum dots and single-molecule transistors~\cite{goldhaber1998kondo,liang2002kondo}.

On the frontier of computational physics and chemistry, the quantum impurity model also plays a central role in embedding methods such as dynamical mean-field theory (DMFT)~\cite{georges1992hubbard,georges1996dynamical,RevModPhys.78.865} and density matrix embedding theory~\cite{knizia2012density,wouters2016practical}. These methods are widely used to study strong correlation in both lattice models and \emph{ab initio} electronic structure. At a high level, embedding methods handle large interacting systems through a sequence of effective impurity problems, solved iteratively in a self-consistent manner. As such, the success of these heuristics heavily depends on the efficiency and accuracy of the underlying impurity solver. 

Despite the fact that the impurity is localized to a subsystem of small, fixed size, these models are surprisingly nontrivial to simulate. The Kondo problem is a classic example:~a perturbative calculation diverges logarithmically at low temperatures~\cite{kondo1964resistance}. Wilson developed NRG to overcome this breakdown~\cite{wilson1975renormalization}, thereby introducing a nonperturbative method that is still widely used today to solve impurity models at low energies~\cite{bulla2008numerical}. Since then, a diverse toolbox of complementary impurity solvers has emerged, including exact diagonalization, matrix product state (MPS) methods~\cite{PhysRevLett.93.246403,PhysRevB.90.115124}, and continuous-time quantum Monte Carlo (CT-QMC)~\cite{werner2006continuous,gull2011continuous}. More recently, general-purpose many-body methods from quantum chemistry, including configuration interaction~\cite{PhysRevB.86.165128,PhysRevB.100.125165} and coupled-cluster theory~\cite{PhysRevB.100.115154}, have also been adapted as impurity solvers.

While these numerical methods are empirically successful for a wide range of quantum impurity problems, their rigorous efficiency is not well understood. This leads us to the central question:~are impurity models fundamentally easy or hard to simulate? For comparison, the ground-state energy problem for certain Hubbard and electronic-structure Hamiltonians is believed to be intractable even for quantum computers~\cite{schuch2009computational,ogorman2022intractability}. However, the prospects for impurity models are more optimistic:~in seminal work by Bravyi and Gosset~\cite{bravyi2017complexity}, it was shown that the ground-state energy can be estimated to inverse-polynomial error in quasipolynomial time on a classical computer. Building on these ideas, Erakovic \emph{et al.}~\cite{erakovic2025high} later gave polynomial-time quantum algorithms under strong assumptions, such as a constant spectral gap in the bath or access to the ground state's one-body reduced density matrix. Thus, the question of whether there is an unconditional polynomial-time algorithm for this problem, classical or quantum, has remained open.

We settle this question by proving that static properties of quantum impurity models are classically easy to compute. Specifically, we give polynomial-time classical algorithms for estimating both the ground-state energy and finite-temperature properties, at any temperature, to inverse-polynomial error. In addition to improving the previous quasipolynomial-time guarantee~\cite{bravyi2017complexity}, this also provides the first rigorous analysis of thermal simulation for impurity models. Our results demonstrate that there is no fundamental obstruction to computing these static properties of impurity models efficiently on classical computers, thereby ruling out any superpolynomial quantum speedup within this regime.

This is particularly relevant to proposals for quantum-enhanced DMFT, which use a quantum computer to solve the underlying impurity problem~\cite{bauer2016hybrid}:~our results imply that a superpolynomial quantum advantage for this application, if any, must come from the dynamical quantities required of the impurity solver. Indeed, we also show as a complementary result that this classical tractability of static properties need not extend to dynamical quantities. We demonstrate this with a canonical example, the nonequilibrium Green's function~\cite{keldysh1965diagram,aoki2014nonequilibrium}, and prove that estimating it to inverse-polynomial error at finite temperature is a complete problem for universal quantum computation. Furthermore, we show that even at infinite temperature, the problem remains complete for the one-clean-qubit model~\cite{knill1998power}, a restricted model of quantum computation that is nevertheless believed to solve problems beyond the reach of classical computers~\cite{shor2008estimating}.

\subsection{Main results}\label{sec:results_summary}

We consider quantum impurity models consisting of a constant-size, arbitrarily interacting impurity coupled to a large free-fermion bath. Let $\gamma_1,\ldots,\gamma_{2n}$ be the Majorana operators on $n$ fermionic modes. The Hamiltonian takes the form
\begin{equation}
H=H_0+H_{\mathrm{imp}},\qquad
H_0=\alpha \I+\frac{i}{4}\sum_{j,k=1}^{2n}h_{jk}\gamma_j\gamma_k, \label{eq:impurity_model}
\end{equation}
where $\alpha\in\mathbb{R}$ is an energy shift and $h$ is a real antisymmetric matrix.
The impurity term $H_{\mathrm{imp}}$ is supported only on the first $m$ modes,
$\gamma_1,\ldots,\gamma_{2m}$, where $m=\cO(1)$, and is a sum of even-degree products of these operators. The term $H_0$ acts extensively but is strictly noninteracting. We make no assumptions on the geometry or any spectral properties.

Note that the canonical fermionic ladder operators are simply a linear rewriting of the Majorana operators, $a_j=(\gamma_{2j-1}+i\gamma_{2j})/2$. Under the Jordan--Wigner transformation, Majorana operators are also equivalent to Pauli strings,
$\gamma_{2j-1} = Z_1\cdots Z_{j-1}X_j$ and 
$\gamma_{2j} = Z_1\cdots Z_{j-1}Y_j$. 

\paragraph{Efficient classical simulation of static properties.} Our main result is that both ground- and thermal-state properties of any quantum impurity model can be classically computed in polynomial time. We remark that these runtimes feature exponential dependence in the impurity size $m$, which is constant in our setting. 

\begin{theorem}[Efficient classical simulation, informal (Theorems \ref{thm:classcial_ground_energy} and \ref{thm:classical-thermal})]\label{thm:intro_main_classical} For any quantum impurity model on $n$ fermion modes, for any $\delta>0$,	 
\begin{itemize}
\item \emph{Ground states:} There is a classical algorithm with runtime $\poly(n,\delta^{-1})$ to compute the ground energy of the quantum impurity model to additive precision $\delta$.

\item  \emph{Thermal states:} For any inverse temperature $\beta$, there is a classical algorithm with runtime $\poly(n,\beta,\delta^{-1})$ to compute the expectation value of any fermionic Gaussian operator of the form $\exp\l(\sum_{j,k} A_{jk} \gamma_j \gamma_k\r)$ to additive precision $\delta$, or the partition function $Z = \tr(e^{-\beta H})$ to relative precision $\delta$.
\end{itemize}
\end{theorem}

The above theorem therefore allows efficient classical simulation to arbitrary inverse-polynomial precision. Note that all local fermionic observables are subsumed under the class of fermionic Gaussian operators, and a relative-error approximation of the partition function $Z$ implies an additive-error approximation to the free energy $F_\beta = -(1/\beta) \log Z$. For ground states, this improves the previous quasipolynomial-time classical algorithm~\cite{bravyi2017complexity} to polynomial time. It is worth noting that, although quantum impurity models can be mapped to one-dimensional spin systems, the resulting Hamiltonian may be gapless, so MPS ground-state algorithms for gapped 1D systems do not necessarily apply here~\cite{landau2015polynomial}.

For thermal states, the 1D structure also allows for an MPO representation~\cite{kuwahara2021improved,molnar2014approximating}, but the bond dimension, and hence the computational cost, can grow subexponentially with $\beta$. CT-QMC can also scale exponentially in $\beta$ due to the sign problem~\cite{gull2011continuous,shinaoka2015negative}. In contrast, our result achieves a  polynomial dependence on $\beta$, rigorously extending the regime of efficient simulation to low temperatures.

The main technical ingredient behind both algorithms is a compression lemma for impurity models. After an appropriate change of basis (a Bogoliubov transformation), we show that the ground state, or the relevant part of the thermal state, can be compressed to an exponentially smaller subspace of the full $2^n$-dimensional Hilbert space.
Our approach draws on ideas from numerical methods, notably the logarithmic separation of bath-energy scales and 1D chain transformations pioneered by Wilson for NRG~\cite{wilson1975renormalization,PhysRevB.21.1003,bulla2008numerical}.

\begin{lemma}[Compression lemma, informal (Corollaries \ref{cor:ground_compress} and \ref{cor:compression_thermal_extension})]\label{lem:compress_intro}
Consider a quantum impurity model on $n$ fermion modes.
Let $\omega$ denote the spectral gap of the bath term. For any $\delta>0$,
\begin{itemize}
    \item \emph{Ground states:} There exists a ground state whose weight is at least $1-\delta$ on a subspace of dimension $\poly(n,\delta^{-1},\omega^{-1}).$
    \item \emph{Thermal states:} Provided $\beta\omega\ge 9\log n$, the thermal state at inverse temperature $\beta$ has weight at least $1-\delta$ on a subspace of dimension $\poly(n,\beta,\delta^{-1})$.
\end{itemize}
Notably, the dimensions of these subspaces are independent of both the impurity interaction strength $\|H_{\mathrm{imp}}\|$ and the impurity--bath coupling strength.
\end{lemma}

The ground- and thermal-state compression lemmas share the same proof strategy developed in Section~\ref{sec:compress}, with slight differences in the details given in Sections~\ref{sec:compress_ground} and~\ref{sec:compress_thermal}.
The basis for each compressed subspace can be computed in time proportional to its dimension. For ground-energy estimation, the parameter $\omega$ can be artificially adjusted to $\Omega(\delta/n)$ while only changing the ground energy by $\cO(\delta)$. The compression lemma then immediately gives a classical algorithm with desired complexity by diagonalizing the Hamiltonian within the compressed subspace. As a corollary, when $H$ has an inverse-polynomial spectral gap above its ground space, the same techniques also imply efficient preparation of a ground state on a quantum computer~\cite{lin2020nearoptimal}.

For thermal states, we only apply the compression lemma to part of the Hamiltonian, leading to the following structural result.

\begin{lemma}[Thermal structural lemma, informal (Section \ref{sec:soft_hard}, Eq.~\eqref{eq:partial_compress})]
Consider any quantum impurity model on $n$ fermion modes. At inverse temperature $\beta$, the thermal state  problem can be efficiently reduced, using the compression lemma, to a Hamiltonian of the form
\begin{equation}
H_{\mathrm{ref}}+V,\qquad \beta\|V\|=\cO(\log n),
\end{equation}
where $H_{\mathrm{ref}}$ is a sum of a compressed interacting Hamiltonian and a free-fermion Hamiltonian, both supported on disjoint sets of modes;~as such, it can be handled efficiently with standard techniques.

Applying CT-QMC to this decomposition gives the classical algorithm of Theorem~\ref{thm:intro_main_classical}, while applying quantum belief propagation~\cite{hastings2007quantum} gives a quantum algorithm for preparing the thermal state within trace-distance $\delta > 0$ with gate complexity $\poly(n, \beta, \delta^{-1})$.
\end{lemma}

Here, $H_{\mathrm{ref}}$ is the easy part of the Hamiltonian, while the product of the inverse temperature and perturbation norm, $\beta\|V\|$, controls the simulation cost. Since we show that $\|V\|$ scales as $\beta^{-1}$, this cancels the exponential $\beta$ dependence of both CT-QMC and quantum belief propagation (QBP), resulting in a cost of $\poly(n, \beta, \delta^{-1})$ in both cases.  In contrast, most existing quantum Gibbs samplers for preparing thermal states feature runtime bounds with exponential or worse dependence on $\beta$, even in 1D~\cite{bergamaschi2025fast} or in perturbative regimes~\cite{tong2025fast,smid2025polynomial}.
For completeness, we provide a rigorous analysis of CT-QMC and QBP in Appendix \ref{sec:CT-QMC} and Appendix \ref{sec:QBP}, respectively.

\paragraph{Hardness results for out-of-equilibrium simulation.} Our polynomial-time classical algorithms imply that calculating static properties of impurity models do not admit any superpolynomial quantum speedup. As a complementary result, we show that this conclusion does not extend to dynamical properties. To this end, we consider time-dependent impurity Hamiltonians, i.e., of the form \cref{eq:impurity_model} but where the coefficients are now allowed to be efficiently computable functions of time. It was previously shown that such Hamiltonians are universal for quantum computation, in the sense that they can encode any quantum circuit within a logical subspace of the fermions~\cite{bravyi2017complexity,brod2014computational}. We extend on this idea to show that this universality persists even in the thermal regime, and over the full fermionic Hilbert space.

Specifically, we consider a canonical object in nonequilibrium physics, the time-dependent Green's function~\cite{aoki2014nonequilibrium}. We prove that computing it to inverse-polynomial error is complete for one of two central quantum complexity classes, depending on the temperature:~$\BQP$ (bounded-error quantum polynomial time) and $\DQC_1$ (deterministic quantum computation with one clean qubit). The former describes the full power of quantum computation~\cite{bernstein1993quantum}, while the latter is a restricted model wherein the quantum computer is only allowed to initialize one pure qubit, with all other qubits maximally mixed~\cite{knill1998power}. The analogy between $\DQC_1$ and impurity models is particularly compelling, as they both feature a small, active subsystem coupled to a large, trivial bath.

\begin{theorem}[Complexity of nonequilibrium Green's functions, informal (Theorems \ref{thm:DQC1-complete} and \ref{thm:BQP-complete})]
    For any $j, k$, consider the dynamical two-point correlation function
    \begin{equation}
        G_\beta(t, t') = \langle \gamma_j(t) \gamma_k(t') \rangle_\beta, \quad \text{where } \gamma_p(s) = U(s)^\dagger \, \gamma_p \, U(s), 
    \end{equation}  
    $U(s)$ is the time evolution under a impurity Hamiltonian $H(s)$ with bounded, time-dependent coefficients, and the thermal expectation is with respect to the Gibbs state of $H(0)$ at inverse temperature $\beta$. For $t, t' = \poly(n)$, estimating $G_\beta(t, t')$ to $1/{\poly(n)}$ additive precision is:
    \begin{itemize}
        \item $\DQC_1$-complete for $\beta = 0$, and
        \item $\BQP$-complete for $\Omega(1) \leq \beta \leq \poly(n)$.
    \end{itemize}
\end{theorem}

There are two main technical ingredients behind this theorem. First, the previous universality result for time-dependent impurity models~\cite{bravyi2017complexity} relied on an encoding of $q$ logical qubits into $n = \poly(q)$ physical modes~\cite{brod2014computational}. However, the Green's function, at $\beta = 0$ for instance, would be $\tr(W)/2^{\poly(q)}$ for some unitary $W$. In contrast, the canonical $\DQC_1$-complete problem is estimating $\tr(C)/2^q$ for any poly-size quantum circuit $C$ on $q$ qubits~\cite{shor2008estimating}. This encoding therefore introduces an exponentially small attenuation factor, making it insufficient for proving the desired theorem. To remedy this, we propose a new universal encoding that only uses $n = q + 1$ fermion modes, and therefore only incurs a constant-factor attenuation. The second ingredient concerns proving $\BQP$-hardness even at high temperature, as $\BQP$ is typically understood with respect to pure-state computations. However, we can apply algorithmic cooling~\cite{schulman1999molecular}:~an efficient reversible circuit that concentrates $\cO(q)$ thermal qubits into $q$ nearly pure qubits with exponentially small error. This turns out to be more than sufficient to encode $\BQP$-complete problems into the finite-temperature Green's function.

\subsection{Technical overview}

We now provide a high-level overview of our classical algorithms, which is the most technically demanding component of this work. Our key technique is a compression lemma showing that the ground state and the relevant part of the thermal state are approximately supported on a small subspace of the full $2^n$-dimensional Hilbert space, as described in Lemma \ref{lem:compress_intro}. For simplicity of the present exposition, we focus attention on the ground state $\ket{\psi}$ here. 

To study this support, recall the annihilation operators
\begin{equation}
a_j\coloneqq \frac{\gamma_{2j-1}+i\gamma_{2j}}{2}
=Z_1\cdots Z_{j-1}\op{0}{1}_j,
\end{equation}
which flips the $j$th fermionic mode from $1$ to $0$ while tracking the parity of all preceding modes. 
The corresponding occupation number is
$n_j\coloneqq a_j^\dagger a_j=\op{1}{1}_j$ whose statistics control the support of $\ket{\psi}$.
For example, if $\bra{\psi}\sum_j n_j\ket{\psi}$ is small, Markov's inequality implies that $\ket{\psi}$ is approximately supported on a small subspace~\cite{bravyi2017complexity}.

\paragraph{The Krylov basis.}
A crucial point is that the support of the ground state is dependent on the choice of single-particle basis, and our compression lemma depends critically on this choice. Indeed, for the canonical basis used in previous work~\cite{bravyi2017complexity}, obtained by diagonalizing $H_0$, one can construct an example in which approximating the ground state to constant accuracy requires at least quasipolynomially large support in this basis.\footnote{In this example, $H_{\mathrm{imp}}$ is chosen to be quadratic, so the full Hamiltonian $H=H_0+H_{imp}$ remains free-fermionic and is therefore not hard to solve. Nevertheless, the example shows that although the ground state of $H_0$ has exactly $1$ support in the canonical basis, a constant-size impurity can significantly change the size of the support of the new ground state in the same basis.}
On the other hand, if the one-particle reduced density matrix of the ground state is known, one can instead construct a basis in which the ground state has small support~\cite{bravyi2017complexity,erakovic2025high}. Our first key observation is that a suitable basis can be constructed \emph{without any prior knowledge of the ground state}. We use a bandwise variant of the Krylov basis, whose standard version is commonly used in numerical methods for  impurity models~\cite{wilson1975renormalization} and explored in previous work~\cite{bravyi2017complexity}. Both the Krylov basis and its bandwise variant can be computed efficiently by block-tridiagonalizing an $n\times n$ matrix.

The key structural consequence of the Krylov basis is that it transforms every impurity model into a coarse-grained one-dimensional chain of fermions:
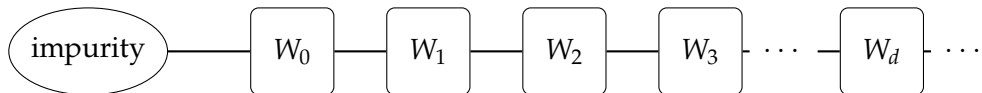
\vspace{0.2em}
\begin{center}
\begin{tikzpicture}[
    box/.style={
        draw,
        rounded corners=3pt,
        minimum height=1.15cm,
        minimum width=1.10cm,
        inner sep=3pt
    },
    line/.style={thick}
]

\draw (0,0) ellipse [x radius=1.05cm, y radius=0.575cm];
\node at (0,0) {impurity};

\node[box] (W0) at (2.7,0) {$W_0$};
\node[box] (W1) at (4.5,0) {$W_1$};
\node[box] (W2) at (6.3,0) {$W_2$};
\node[box] (W3) at (8.1,0) {$W_3$};
\node at (9.2,0) {$\cdots$};
\node[box] (Wd) at (10.5,0) {$W_d$};
\node at (11.6,0) {$\cdots$};

\draw[line] (1.05,0) -- (W0.west);
\draw[line] (W0.east) -- (W1.west);
\draw[line] (W1.east) -- (W2.west);
\draw[line] (W2.east) -- (W3.west);
\draw[line] (W3.east) -- (8.75,0);
\draw[line] (9.65,0) -- (Wd.west);
\draw[line] (Wd.east) -- (11.15,0);

\end{tikzpicture}
\end{center}
\vspace{0.2em}
In this basis, the annihilation operators are grouped into blocks $W_d$:~the first block contains all the impurity modes, each subsequent block contains at most $2m$ bath modes, and the transformed Hamiltonian couples only neighboring blocks. Although a one-dimensional structure alone does not imply efficient classical simulation~\cite{aharonov2009power,gottesman2009quantum}, having the impurity confined to the first block of the chain suggests that the occupation probability $\langle\psi|n_i|\psi\rangle$ may decay exponentially with the distance from the impurity.
This provides the basic intuition for that the ground state has small support in this basis, although the compression lemma requires a much finer analysis of the occupation statistics.

We next illustrate the key proof strategy and the main technical obstacle for the compression lemma through a simple one-dimensional example.

\paragraph{Explicit example and proof strategy.}

To illustrate the proof strategy, we consider a simple impurity model that already has the desired one-dimensional structure. For clarity, in this example we relabel the impurity modes $a_1,a_2$ as $b_1,b_2$, and the bath modes $a_3,\ldots,a_n$ as $a_1,\ldots,a_N$, where $N=n-2$. Consider the Hamiltonian
\begin{equation}
H
=
U \left(b_1^\dagger b_1 b_2^\dagger b_2 + b_1^\dagger b_2 +b_2^\dagger b_1\right)
+g\left(
a_1^\dagger b_1
+
b_1^\dagger a_1
\right)
+2\omega\sum_{j=1}^{N}a_j^\dagger a_j
-\frac{\omega}{2}\sum_{j=1}^{N-1}
\left(
a_j^\dagger a_{j+1}
+
a_{j+1}^\dagger a_j
\right)\!\,
\label{eq:overview-example}
\end{equation}
Here $U,g\in \R$ are the impurity interaction strength and  the impurity--bath coupling strength, respectively, and $\omega>0$ is a lower bound on the single-particle energies of the bath. 

Let $\ket{\psi}$ be a ground state. Define $n_j=a_j^\dagger a_j$ for the bath modes $j=1,\ldots,N$. We first illustrate our proof strategy by proving the weaker statement that the occupation probability
$\bra{\psi}n_j\ket{\psi}$ decays exponentially with the distance of the $j$th bath mode from the impurity. Instead of estimating
$\bra{\psi}\sum_j n_j\ket{\psi}$, we consider the one-mode \emph{weighted occupation statistic}
\begin{equation}
F_1
\coloneq
\sum_{j=1}^{N}
e^{j-1}
\bra{\psi}n_j\ket{\psi}.
\end{equation}
Here the subscript $1$ indicates that $F_1$ involves one-mode occupations.

The key tool for estimating $F_1$ is the \emph{commutator identity} with respect to $[a_j,H]$.
The algebraic relations for fermionic ladder operators imply that
$[a_j,H]$ is at most linear in the bath modes. For this particular example,
we have
\begin{equation}
[a_j,H]
=
2\omega a_j
-\frac{\omega}{2}
\left(
a_{j-1}+a_{j+1}
\right)
+
g\delta_{j1}b_1,
\label{eq:overview-commutator-identity}
\end{equation}
where $a_0=a_{N+1}\coloneqq 0$ and  $\delta_{j1}=1$ if $j=1$ and $0$ otherwise. If we let $\alpha_{\mathrm{gs}}$ be the ground energy, then 
$\alpha_{\mathrm{gs}}\I-H \preceq 0$ and hence
\begin{equation}
\bra{\psi}
a_j^\dagger[a_j,H]
\ket{\psi}
=
\alpha_{gs}
\bra{\psi}a_j^\dagger a_j\ket{\psi}
-
\bra{\psi}a_j^\dagger H a_j\ket{\psi}
\leq0.
\label{eq:overview-commutator-bound}
\end{equation}
Inserting the commutator identity into the left-hand side of the above equation leads to a recursion for $F_1$ with a lower-degree contribution, namely a constant $|g|$ that controls the boundary term $\|g\delta_{j1} b_1\|$. Intuitively, this is because the commutator identity in Eq.~\eqref{eq:overview-commutator-identity} never increases the degree in the bath modes:~$[a_j,H]$ remains linear in the bath modes, while the term $g\delta_{j1}b_1$ has strictly lower degree.

To make this intuition precise, let $E$ be the coefficient matrix of the quadratic bath term, which is indexed by the bath modes ${a_j}$ and whose entry $E_{kj}$ is the coefficient of $a_j^\dagger a_k$ in Eq.~\eqref{eq:overview-example}.

We also define
$D\coloneqq\sum_{j=1}^N(j-1)\op{j}{j}$,
which records the distance of each bath mode from the impurity.
Define the weighted correlation matrix $M$ and the corresponding weighted bath matrix $B$ by
\[
M_{jk}
\coloneqq
e^{\frac12(j-1)+\frac12(k-1)}
\bra{\psi}a_j^\dagger a_k\ket{\psi},
\qquad
B
\coloneqq
\Reop\l(e^{-D/2}Ee^{D/2}\r).
\]
where $\Reop(A)=(A+A^\dagger)/2$ is the Hermitian part of an operator $A$. Then $M\succeq0$ and $\tr(M)=F_1$. Next, we multiply Eq.~\eqref{eq:overview-commutator-bound} by $e^{j-1}$, insert the commutator identity, take the real part, and sum over $j$. Regrouping the resulting terms using the definitions of $B$ and $M$, we get
\begin{equation}
\tr(BM)
+
g\,\Reop\bra{\psi}a_1^\dagger b_1\ket{\psi}\leq 0.
\label{eq:overview-weighted-commutator}
\end{equation}
The term
$g \Reop\bra{\psi}a_1^\dagger b_1\ket{\psi}$
is bounded by
$
|g\,\bra{\psi}a_1^\dagger b_1\ket{\psi}|
\leq |g|\sqrt{F_1}$
by Cauchy--Schwarz. The key use of the one-dimensional structure is the \emph{positivity of the weighted bath}:~$B\succeq\frac{\omega}{2}\I$. Since $M\succeq 0$, we have
$\tr(BM)\geq\frac{\omega}{2}\tr(M)
=\frac{\omega}{2}F_1.
$
Combining these two bounds with Eq.~\eqref{eq:overview-weighted-commutator}, we obtain
\begin{equation}
\frac{\omega}{2}F_1
\leq
|g|\sqrt{F_1}
\qquad\implies\quad
F_1
\leq
\frac{4g^2}{\omega^2}
\quad \implies\quad \bra{\psi}n_j\ket{\psi}
\leq
\frac{4g^2}{\omega^2}e^{-(j-1)},
\label{eq:overview-F1-recursion}
\end{equation}
where the last equation comes from the definition of $F_1$.

\paragraph{Bandwise Krylov basis and enlarging the impurity.}
Although Eq.~\eqref{eq:overview-F1-recursion} shows exponential decay of the one-mode occupation, a first-moment bound of this form is not sufficient for the compression result in general. The recursion strategy, however, extends to \emph{higher-order} weighted occupation statistics $F_q$, obtained by replacing the one-mode occupation $n_j$ in $F_1$ by the $q$-mode occupation $n_{i_1}\cdots n_{i_q}$. Similar commutator identity and arguments then give a recursion relating $F_q$ to the lower-degree statistic $F_{q-1}$. Controlling this recursion for all $q$ allows us to go beyond the first-moment estimate and obtain a stronger concentration bound.

A second important issue is that the bound in Eq.~\eqref{eq:overview-F1-recursion} depends on $g/\omega$. In general, $\omega$ can be small (for example $\omega \sim 1/n)$ while the impurity--bath coupling can be arbitrary, so $g/\omega$ need not be small. In that case, a bound like Eq.~\eqref{eq:overview-F1-recursion} is not sufficient. To handle this, we use ideas similar to the logarithmic separation of bath-energy scales in NRG~\cite{wilson1975renormalization,PhysRevB.21.1003,bulla2008numerical}. We first perform some preprocessing to ensure $\omega>0$ and set $\omega_\ell=2^\ell \omega$. We then partition the spectrum of the bath matrix $E$ into intervals $[\omega_\ell,2\omega_\ell)$ for $\ell\geq 0$. On each interval, we restrict $E$ to the corresponding spectral subspace and apply the one-dimensional Krylov-basis transformation to the restricted matrix. We refer to the resulting basis (over the entire single-particle space) as the \emph{bandwise Krylov basis}.

In this way, the single-particle energies in band $\ell$ lie between $\omega_\ell$ and $2\omega_\ell$. At the same time, we enlarge the impurity by including the bath modes in the first Krylov block. The coupling between this enlarged impurity and the remaining bath then comes from the bath matrix $E$ itself, rather than from the original impurity--bath coupling $g$. Since this coupling lies within band $\ell$, its strength is bounded by the energy scale $\cO(\omega_\ell)$ of that band. Thus, the problematic $g/\omega$ coupling is replaced by an effective coupling which is bounded by a constant within each band. 

After this reorganization, the recursion can be applied band by band without requiring the original impurity--bath coupling to be small relative to the bath's spectral gap. We ultimately arrive at our compression lemma that states that the ground state can be compressed into a subspace whose dimension is independent of the impurity interaction strength $U$, the impurity--bath coupling strength $g$, and scales polynomially with $\omega^{-1}$.
For ground-energy estimation to precision $\delta$, we can without loss of generality take $\omega \geq \Omega(\delta/n)$, and hence the resulting compressed subspace has polynomial dimension. More details can be found in Figure \ref{fig:1D}, Corollary~\ref{cor:ground_compress} and Theorem~\ref{thm:classcial_ground_energy}. A similar compression scheme also works for the relevant part of the thermal state; see Corollary~\ref{cor:compression_thermal_extension} and Theorem~\ref{thm:ctqmc}.

\subsection{Outlook}

In this work, we have established polynomial-time classical algorithms for static properties of quantum impurity models, settling an important open question about their computational complexity~\cite{bravyi2017complexity}, while also providing evidence for possible quantum advantage in certain dynamical problems. The key technical component of our classical algorithms, the compression framework, shows that for impurity models, the ground states and the thermally gapped part of the Gibbs states can be well approximated in low-dimensional subspaces. In particular, our results establish that the ground-state problem for quantum impurity models is fundamentally classically tractable, placing the longstanding empirical success of impurity solvers~\cite{wilson1975renormalization,bulla2008numerical,gull2011continuous} on a rigorous footing. They also imply that any superpolynomial quantum advantage in solving impurity models, for example as an application to embedding theories~\cite{bauer2016hybrid,rubin2016hybrid}, cannot arise from the static problem, but must rather exploit the hardness of simulating dynamics.

An interesting question is whether this compression framework can be extended further, for example excited states, or otherwise adapted to certain dynamical problems~\cite{anders2005real,aoki2014nonequilibrium}. Such extensions could broaden the range of impurity properties that admit provably efficient classical simulation. In a separate vein, our current algorithms feature exponential dependence on the impurity size $m$, which we have treated as a fixed constant in this work. Improvements to this dependence could open the possibility of nontrivial algorithms for systems with an extensive number of interacting degrees of freedom, such as the Fermi--Hubbard model~\cite{qin2022hubbard}. Finally, while our classical algorithms rule out a superpolynomial quantum advantage for static properties of impurity models, the possibility for a polynomial speedup is left open. One possible route is to identify a substantially smaller subspace that has \emph{constant} overlap with the ground state;~combined with quantum phase estimation, such a property, if it exists, could enable a nontrivial quantum speedup.

\section{Preliminaries}

In this section, we introduce the notation used throughout the paper.

\subsection{Background on fermions}\label{sec:fermi_operator}

We review the fermionic operators, Fock basis, and single-particle basis rotations used throughout the paper. For further background on fermionic quantum information, see~\cite{nielsen2005fermionic,bravyi2005lagrangian}.

\paragraph{Second quantization.}
A system of $n$ fermionic modes is described by $n$  operators $a_1,\ldots,a_n$ satisfying the \textit{canonical anticommutation relations} (CAR)
\begin{align}
    \{a_j,a_k^\dagger\}
    \coloneqq a_j a_k^\dagger+a_k^\dagger a_j
    =\delta_{jk}\I,\qquad
    \{a_j,a_k\} =0,
    \qquad j,k\in[n].
    \label{eq:annihilation-car}
\end{align}
Using the Jordan--Wigner transformation, these operators act on an $n$-qubit space and are represented by
\begin{align}
	a_j = Z_1 \otimes \cdots \otimes Z_{j-1}\otimes \begin{bmatrix}
		0 & 1\\
		0&0
\end{bmatrix},\qquad 1\leq j \leq n, \label{eq:annihilation}
\end{align}
where $Z_k$ represents the Pauli-$Z$ operator on qubit $k$. The operators
$a_j^\dagger$ and $a_j$ are called the \textit{creation and annihilation operators}.

Alternatively, the same fermionic system can be described in terms of $2n$ Majorana operators $\gamma_1,\ldots,\gamma_{2n}$, defined by
\begin{align}
    \gamma_{2j-1} =a_j+a_j^\dagger,\qquad
    \gamma_{2j} = -i (a_j-a_j^\dagger),\quad\text{where } \{\gamma_j,\gamma_k\}=2\delta_{jk} \I. \label{eq:CAR}
\end{align}
Using the same Jordan--Wigner transformation, these Majorana operators are represented by 
\begin{align}
	\gamma_{2k-1} = Z_1 \otimes \cdots \otimes Z_{k-1} \otimes X_k, \qquad
	\gamma_{2k} = Z_1 \otimes \cdots \otimes Z_{k-1} \otimes Y_k,
\end{align}
where $X_k,Y_k$ are Pauli-$X$, Pauli-$Y$ operators on qubit $k$.

We say that an operator $O$ has \textit{even parity} if it is a sum of even-degree monomials in the creation and
annihilation operators (equivalently, in the Majorana operators), such as $a_1^\dagger a_2$ or $\gamma_1\gamma_2+\gamma_1\gamma_3\gamma_5\gamma_6$. We define \textit{odd parity} similarly. Note that this is distinct from the fermionic parity of states.

\paragraph{Fock basis.}
The annihilation operators $a_1,\ldots,a_n$ define a $2^n$-dimensional Hilbert space called the \textit{fermionic Fock space}. Its standard basis, called the \textit{Fock basis}, is labeled by occupation numbers.
We define
the \textit{vacuum state} $\lvert \mathrm{vac}\rangle$ as the common zero-eigenstate of the operators $\{a_j^\dagger a_j\}_{j=1}^n$. It satisfies
\begin{equation}
    a_j\lvert \mathrm{vac}\rangle=0 \qquad \forall j\in[n].
\end{equation}
For any $x\in\{0,1\}^n$, with respect to the ordering $a_1,\ldots,a_n$ of the annihilation operators, the corresponding \textit{Fock basis state} is defined by
\begin{align}
\lvert x_1,\ldots,x_n\rangle
\coloneqq
(a_1^\dagger)^{x_1}\cdots(a_n^\dagger)^{x_n}
\lvert \mathrm{vac}\rangle.
\end{align}
Here $x_j=1$ means that the $j$th fermionic mode is occupied, while $x_j=0$ means that it is unoccupied. The states
    $\{\lvert x\rangle:x\in\{0,1\}^n\}
$
form an orthonormal basis of the Fock space.

Under the Jordan--Wigner mapping, these Fock basis states correspond directly to the computational basis states $\lvert x_1\cdots x_n\rangle$ of $n$ qubits. In particular, $\lvert \mathrm{vac}\rangle$ corresponds to the all-zero state $\lvert 0^n \rangle$.

\paragraph{Single-particle basis rotation.}

It is convenient to use an $n$-dimensional vector to represent a linear
combination of annihilation operators. For any $\mu\in\C^n$, denote
\begin{align}
    a(\mu)\coloneqq\sum_{j=1}^n \mu_j a_j.
\end{align}
Although $\mu$ is an $n$-dimensional vector, $a(\mu)$ acts on the $2^n$-dimensional Fock space. 
The space $\mathbb{C}^n$ that $\mu$ lives in is the \emph{single-particle space} of the fermions, and any orthonormal basis $\{f_k\}_{k=1}^n$ of $\mathbb{C}^n$ is called a \emph{single-particle basis}.
In particular, the standard basis vector $e_j$ represents the annihilation operator $a_j$, i.e.,
    $a(e_j)=a_j.$
With a slight abuse of terminology, for any subspace $W\subseteq\mathbb{C}^n$, we refer to operators $a(\mu)$ with $\mu\in W$ as fermionic modes in $W$.

A unitary change of basis on the single-particle space induces a corresponding change of fermionic modes. Let $V$ be a unitary on $\C^n$, whose columns $\{Ve_j\}_j$ form a single-particle basis, and define
\begin{equation}
b_j\coloneqq a(Ve_j)=\sum_{k=1}^n V_{kj}a_k,\qquad j\in[n].
\end{equation}
The operators $b_1,\ldots,b_n$ also satisfy the CAR in Eq.~(\ref{eq:annihilation-car}), and hence define a new set of fermionic annihilation operators with a corresponding Fock basis.

These two bases are related by a unitary $\mathcal{U}(V)$ on the full $2^n$-dimensional Fock space called a \textit{single-particle basis rotation}, where  $\mathcal{U}(V)a_j\mathcal{U}(V)^\dagger=b_j$. The unitary $\mathcal{U}(V)$ is fully characterized by the $n\times n$ matrix $V$, and it suffices to work with the new fermionic modes $b_1,\ldots,b_n$. When a qubit representation is needed for algorithmic implementation, one can apply a fermion-to-qubit mapping with respect to these new annihilation operators.

\paragraph{Quadratic Hamiltonians.}
Recall that $e_j\in\mathbb{C}^n$ denotes the $j$th standard basis vector.
For any $n\times n$ Hermitian matrix $P$, we define the corresponding quadratic Hamiltonian by
\begin{align}
\dE(P)\coloneqq\sum_{j,k}\langle e_j,P e_k\rangle a_k^\dagger a_j.
\end{align}
Let $V$ be a unitary and denote its column vectors as $\{f_j\coloneqq Ve_j\}_j$, which form the corresponding single-particle basis. One can check that 
\begin{align}
    \dE(P)\coloneqq\sum_{j,k}\langle e_j,V^\dagger P V e_k\rangle a(f_k)^\dagger a(f_j).\label{eq:V_change_basis}
\end{align}
In particular, when choosing $V$ that diagonalizes $P$, that is    
$P=V\operatorname{diag}(\lambda_1,\ldots,\lambda_n)V^\dagger,$ 
and defining new annihilation operators
$b_j\coloneqq a(f_j)$, we have $\dE(P)=\sum_{j=1}^n \lambda_j b_j^\dagger b_j. $
Applying the fermion-to-qubit mapping with respect to the modes $\{b_j\}_{j=1}^n$, we can therefore represent
$\dE(P)=\sum_{j=1}^n \lambda_j \op{1}{1}_j$ and can calculate the full spectrum of $\dE(P)$. Such Hamiltonians are also said to be \textit{free-fermionic}.

The following identities are direct consequences of the canonical anticommutation relations and will be used  in later proofs. For any $\mu,\nu\in\C^n$, one can check that
\begin{align}
    \{a(\mu)^\dagger, a(\nu)\}=\langle \mu,\nu\rangle \I,\qquad \{a(\mu), a(\nu)\}=0,\qquad  [\dE(P),a(\mu)]=-a(P\mu).
\end{align}

\subsection{Preprocessing of the impurity Hamiltonian}\label{sec:hyb_decom}

Consider the quantum impurity model defined in Eq.~(\ref{eq:impurity_model}). 
We denote its ground state as $\rho^{(\infty)}$ and
its thermal state at inverse temperature $\beta$ as 
    $\rho^{(\beta)}\coloneqq e^{-\beta H}/{\tr(e^{-\beta H})}.$

Note that adding a multiple of the identity to the Hamiltonian in Eq.~(\ref{eq:impurity_model}) does not change its ground or thermal states, so we will fix the value of $\alpha\in\R$ later. We assume $\|H\|=\poly(n)$. 

Let $h_{\mathrm{bath}}$ denote the principal submatrix of $h$ corresponding to the bath Majorana operators $\gamma_{2m+1},\ldots,\gamma_{2n}$, and define $\es$ as an upper bound for the bath single-particle energies,  
\begin{align}
\|h_{\mathrm{bath}}\|\leq \es=\poly(n).
\end{align}
Note that $\es$ is independent of the impurity--bath coupling strength and the impurity interaction strength $\|H_{imp}\|$.

\paragraph{Preprocessing and the hybridization form.} 
Before turning to the proofs, we perform several preprocessing steps that put the general Hamiltonian into a convenient form.
We regard $\gamma_1,\ldots,\gamma_{2m}$ as the impurity Majorana operators and the remaining Majorana operators as the bath Majorana operators. It is convenient to decompose $H$ into an impurity-only part, a bath-only part, and a hybridization part that couples the impurity and the bath.\footnote{In fact, this is how impurity models are typically presented in the physics literature.} In Appendix~\ref{app:cano_bath}, we show that, by applying a canonical transformation to the bath, one can efficiently construct annihilation operators $a_1,\ldots,a_n$ and write $H$ in the hybridization form:

\begin{equation}
    H
    = K_0+ 
    \dE(E_0)
    +
    \sum_{j=1}^{2m}
    \left(
        a^\dagger(\mu_j)O_j
        +
        O_j^\dagger a(\mu_j)
    \right).\label{eq:impurity_model_hyb}
\end{equation}
Here $a_1,\ldots,a_m$ are called the impurity modes, and the remaining modes are called the bath
modes. The term $K_0$ only acts on impurity modes, $\dE(E_0)$ only acts on the bath modes, and $E_0$ is diagonal. The term $a^\dagger(\mu_j)O_j$ and its Hermitian conjugate describe the hybridization term, where
$a(\mu_j)$ only acts on the bath modes and $O_j$ only acts on the impurity modes and has odd parity. Here the vectors $\{\mu_j\}_j$ need not
be orthogonal.

Define the corresponding \textit{impurity single-particle space} and \textit{bath single-particle space} to be 
\begin{align}
W_{\mathrm{imp}}\coloneqq\operatorname{span}\{e_1,\ldots,e_m\}\subseteq \C^n,\qquad W_{\mathrm{bath}}\coloneqq W_{\mathrm{imp}}^\perp.
\end{align}

In Appendix~\ref{app:cano_bath}, we explain that, when $E_0$ is viewed as a matrix on $W_{\mathrm{bath}}$, without loss of generality for both ground-state and thermal-state simulation, we may assume that $E_0$ is strictly positive; that is, there exists $\omega>0$ such that
\begin{align}
0<\omega\I\leq E_0\leq\es\I. \label{eq:lower_single_particle}
\end{align}

\section{Bandwise Krylov representation}\label{sec:band_krylov_rep}

In the hybridization form Eq.~\eqref{eq:impurity_model_hyb}, we express $H$ in terms of the annihilation operators $\{a_j\}_j$; we write the corresponding single-particle basis as the standard basis $\{e_j\}_j$, i.e., $a_j=a(e_j)$. 

The first key step of our proof is to choose a different basis $\{f_j\}_j$ for the bath single-particle space, which we call the \textit{bandwise Krylov basis}, and re-express the Hamiltonian in terms of the corresponding annihilation operators $\{a(f_j)\}_j\cup \{a(e_j)\}_{j=1}^m$, where the second set corresponds to the basis for the impurity single-particle space. This particular choice of bandwise Krylov basis transforms any impurity model into a one-dimensional structure. 

In this section, we explain how the bandwise Krylov basis $\{f_j\}_j$ is chosen. In Section \ref{sec:1D_from_block_tri}, we explain how a one-dimensional structure can be obtained via block-tridiagonalization.  To handle small single-particle energies, we then use the ideas of bandwise Krylov space in Section \ref{sec:band_krylov} and the enlarged impurity  in Section \ref{sec:enlarged_impurity}. To guide the reader, we illustrate and summarize the constructions in these sections in Figure~\ref{fig:1D}.

\begin{figure}[htp]
 \centering
    \scalebox{0.9}{\begin{tikzpicture}[x=0.95mm,y=0.95mm,inner sep=0pt]
  \tikzset{
    shell/.style={
      draw=black!90,
      fill=white,
      line width=.27mm,
      rounded corners=1.9mm
    },
    main/.style={
      text=black,
      font=\sffamily\normalsize
    },
    secondary/.style={
      text=black,
      font=\sffamily\small
    },
    band/.style={
      text=black,
      font=\sffamily\normalsize
    }
  }

  \node[main,anchor=west] at (-3.0,78.6)
    {\textbf{(a)} Krylov basis};

  \draw[black!90,line width=.25mm]
    (52.4,78.6)--(64.6,78.6)
    (77.5,78.6)--(90.5,78.6)
    (103.5,78.6)--(112.6,78.6)
    (125.5,78.6)--(135.0,78.6);

\draw[shell] (42.75,78.6) ellipse [x radius=9.65, y radius=5.3];
  
  \draw[shell] (64.6,73.3) rectangle (77.5,83.9);
  \draw[shell] (90.5,73.3) rectangle (103.5,83.9);
  \draw[shell] (112.6,73.3) rectangle (125.5,83.9);

  \node[main] at (42.75,78.6) {Impurity};
  \node[main] at (71.05,78.6) {$W_0$};
  \node[main] at (97.0,78.6) {$W_1$};
  \node[main] at (119.05,78.6) {$W_2$};

  \node[secondary] at (58.5,96.2)
    {hybridization term};
  \node[secondary] at (58.5,92.1)
    {$O_j^\dagger a(\mu_j)\,\cdots$};

  \draw[->,black!65,line width=.25mm]
    (58.5,88.6)--(58.5,79.9);

  \node[secondary] at (58.5,74.2)
    {$\|\mu_j\|$};

  \node[main,anchor=west] at (-3.0,62.0)
    {\textbf{(b)} Bandwise};
  \node[main,anchor=west] at (3.5,57.9)
    {Krylov basis};

  \draw[
    black!90,
    line width=.25mm,
    dash pattern=on 2.1mm off 1.2mm
  ]
    (31.8,17.1) rectangle (79.1,68.2);

  \draw[
    black!48,
    line width=.25mm,
    dash pattern=on .30mm off .45mm,
    line cap=round
  ]
    (90.0,17.1) rectangle (136.3,68.2);

  \draw[black!90,line width=.25mm]
    (52.2,49.5)--(64.6,59.2)
    (52.2,47.8)--(64.6,45.4)
    (52.2,45.1)--(64.6,25.7);

  \foreach \yy in {59.2,45.4,25.7}{
    \draw[black!90,line width=.25mm]
      (77.5,\yy)--(90.9,\yy)
      (103.8,\yy)--(112.7,\yy)
      (125.6,\yy)--(134.6,\yy);
  }

  \draw[shell] (42.65,47.85) ellipse [x radius=9.55, y radius=5.25];
  \node[main] at (42.65,47.85) {Impurity};

  \foreach \yy/\lab in {
    59.2/{W_{0,0}},
    45.4/{W_{1,0}},
    25.7/{W_{L,0}}
  }{
    \draw[shell]
      (64.6,\yy-5.25) rectangle (77.5,\yy+5.25);
    \node[main] at (71.05,\yy) {$\lab$};
  }

  \node[main] at (71.05,35.5) {$\cdots$};

  \foreach \yy/\first/\second in {
    59.2/{W_{0,1}}/{W_{0,2}},
    45.4/{W_{1,1}}/{W_{1,2}},
    25.7/{W_{L,1}}/{W_{L,2}}
  }{
    \draw[shell]
      (90.9,\yy-5.25) rectangle (103.8,\yy+5.25);
    \draw[shell]
      (112.7,\yy-5.25) rectangle (125.6,\yy+5.25);

    \node[main] at (97.35,\yy) {$\first$};
    \node[main] at (119.15,\yy) {$\second$};
  }

  \node[main] at (97.35,35.5) {$\cdots$};
  \node[main] at (119.15,35.5) {$\cdots$};

  \node[secondary] at (50.2,34.5)
    {$\|\Pi_L\mu_j\|$};
  \node[secondary] at (84.4,28.6)
    {$2\omega_L$};

  \node[band,anchor=west] at (140.0,59.2)
    { Band 0: $[\omega_0,2\omega_0)$};
  \node[band,anchor=west] at (140.0,45.4)
    {Band 1: $[\omega_1,2\omega_1)$};
  \node[band,anchor=west] at (140.0,25.7)
    {Band L: $[\omega_L,2\omega_L)$};

  \node[main] at (52.8,10.2)
    {Enlarged impurity};
  \node[main] at (124.7,10.2)
    {Residual bath};

  \draw[
    black!65,
    line width=.22mm,
    dash pattern=on 1.2mm off .8mm
  ]
    (71.2,0.4) rectangle (109.8,10.6);

  \node[
    text=black,
    font=\sffamily\footnotesize
  ] at (90.5,7.6)
    {new hybridization term};

  \node[secondary] at (89.5,3.5)
    {$G_{Lj}^{\dagger}a(u_{Lj})\,\cdots$};

  \draw[->,black!65,line width=.25mm]
    (84.5,11.3)--(84.5,23.4);

\end{tikzpicture}}
    \captionsetup{width=1\textwidth, font=small}
\caption{Illustration of the one-dimensional structure induced by a Krylov basis for the bath single-particle space. Each block represents a group of annihilation operators, and an edge indicates that the Hamiltonian contains coupling terms between the corresponding blocks.\\
\textbf{(a) Krylov basis.} We choose a bath single-particle basis $\{g_j\}_j$ consistent with the Krylov shells $W_0,W_1,\ldots$. When $H$ is expressed in terms of the annihilation operators $\{a(g_j)\}_j$, it has a one-dimensional block structure. Each block has  at most $2m$ modes. The impurity couples only to the first block through the  hybridization terms 
$\{O_j^\dagger a(\mu_j)\}_j$, 
and each subsequent block couples only to its neighboring blocks.\\
\textbf{(b) Bandwise Krylov basis.} We first decompose the bath single-particle space into logarithmically spaced energy bands $[\omega_\ell,2\omega_\ell)$ and construct a Krylov basis separately within each band. The complexity of our classical simulation depends on
$\kappa\coloneqq\sum_\ell \xi_\ell\omega_\ell^{-2}$,
where $\xi_\ell$ measures the squared hybridization strength in band $\ell$. For the original impurity--bath partition, the hybridization vectors $\Pi_\ell\mu_j$ need not be small relative to $\omega_\ell$. We therefore enlarge the impurity by including the first Krylov block from each band. The  coupling between the enlarged impurity and the residual bath is an off-diagonal block of the bath matrix restricted to band $\ell$ and has norm at most $2\omega_\ell$. Consequently, $\xi_\ell/\omega_\ell^2=O(m)$ for each band, and hence $\kappa$ is controlled.\\
\textbf{Note on renaming the operators.} After choosing the bandwise Krylov basis, we relabel the enlarged-impurity modes as $\{b_j\}_j$ and the residual-bath modes as $\{a_j\}_j$. 
}\label{fig:1D} 
 \end{figure}

\subsection{One-dimensional structure via block-tridiagonalization}\label{sec:1D_from_block_tri}
We first explain
that, to obtain a one-dimensional structure, it suffices to find a unitary that block-tridiagonalizes the quadratic bath term.
Recall the hybridization form of the Hamiltonian in Eq.~\eqref{eq:impurity_model_hyb}. For the moment, we replace $E_0$ by a placeholder Hermitian matrix $E_*$ acting on a subspace $W_*$ of the bath single-particle space,  project each $\mu_{j}$ onto $W_*$ and denote the resulting vector by $\mu_{j*}$. Denote the resulting Hamiltonian by $H_*$,

\begin{equation}
    H_*
    = K_0+ 
    \dE(E_*)
    +
    \sum_{j=1}^{2m}
    \left(
        a^\dagger(\mu_{j*})O_j
        +
        O_j^\dagger a(\mu_{j*})
    \right).\label{eq:hyb_placeholder}
\end{equation}

Denote $W_{*0}\coloneqq \spn\{\mu_{1*},\ldots,\mu_{2m*}\}$ and $s_0\coloneqq\dim(W_{*0})\leq 2m$.
To transform the Hamiltonian $H_*$ into a one-dimensional structure, 
it suffices to find an orthonormal basis of the bath single-particle subspace $W_*$, whose basis vectors form the columns of an isometry $V$,
such that  (i) the first $s_0$ column vectors of 
$V$ span a space containing the hybridization vectors $\{\mu_{1*},\ldots,\mu_{2m*}\}$,  and  (ii) the resulting matrix is block-tridiagonal,
\begin{align}
    V^\dagger E_* V  = 
\begin{pmatrix}
P_{0} & Q_{0}^\dagger & 0 & 0 \\
Q_{0} & P_{1} & Q_{1}^\dagger & 0\\
0& Q_{1} & P_{2} & \ddots\\
0& 0 & \ddots & \ddots 
\end{pmatrix},\label{eq:block-tridiag}
\end{align}
where $P_0\in \C^{s_0\times s_0}$  and all subsequent blocks  $P_k$ are of size at most $2m
\times 2m$. The norms  $\|P_j\|, \|Q_j\|$ are bounded by $\|E_*\|$ since the corresponding blocks are submatrices of $V^\dagger E_* V$. 
To see why Eq.~(\ref{eq:block-tridiag}) leads to a one-dimensional structure, suppose we have found such a $V$ and denote its column vectors by $\{g_j\}_j$. Define a new set of annihilation operators as $\{a(g_j)\}_j$. Then, in this new basis, $H_*$ is re-expressed by
\begin{align}
    H_*= K_0  + \sum_{j,k} \langle e_j,\,(V^\dagger E_* V)\, e_k\rangle \,\, a(g_k)^\dagger a(g_j) + \sum_{j=1}^{2m} \left( a^\dagger( \mu_{j*}) O_j +  O_j^\dagger a( \mu_{j*}) \right)
\end{align}
where $K_0$ and $\{O_j\}_j$ are unchanged since we only change the bath single-particle basis.
The Hamiltonian $H_*$ has a one-dimensional structure as in Figure \ref{fig:1D} (a) due to the block-tridiagonal form of $V^\dagger E_*V$ in Eq.~(\ref{eq:block-tridiag}):
More specifically,
let $I_k$ denote the set of column indices of $V$ corresponding to the
diagonal block $P_k$ in Eq.~\eqref{eq:block-tridiag}.
We group the annihilation operators
$\{a(g_j):j\in I_k\}$ as the $k$th \textit{bath-mode group}, which
corresponds to the $(k+1)$st square block in Figure~\ref{fig:1D}(a). The block-tridiagonal form of $V^\dagger E_*V$ in Eq.~(\ref{eq:block-tridiag}) implies that the bath blocks only couple to their neighbors. Besides, since  each hybridization vector $\mu_{j*}$ lies in the span of $g_1,\ldots,g_{s_0}$,   the
hybridization term $a^\dagger( \mu_{j*}) O_j +  O_j^\dagger a( \mu_{j*})$  couples the impurity only to the first bath block. 
 
Moreover, by assumption the blocks $\{P_k\}_{k\geq 0}$ are each of size at most $2m\times 2m$, so the size of each bath group is $|I_k|\leq 2m$.

\paragraph{Block-tridiagonalization via orthogonalizing a Krylov sequence.}
To construct the isometry $V$  that block-tridiagonalizes $E_*$ and whose first few columns span a space containing the hybridization vectors, one could use the  standard block Lanczos tridiagonalization, or equivalently, orthogonalize a Krylov sequence.

More specifically, let $W_{*0}$ denote the space spanned by the hybridization vectors, and define the depth-$d$ Krylov space $K_{*d}$ by successively applying $E_*$ to $W_{*0}$:
\begin{align}
    W_{*0}
    \coloneqq
    \operatorname{span}\{\mu_{1*},\ldots,\mu_{2m*}\},\qquad
    K_{*d}
    \coloneqq
    \operatorname{span}\{W_{*0},E_*W_{*0},\ldots,E_*^dW_{*0}\}.
\end{align}

The column vectors of $V$ are obtained by successively orthogonalizing the Krylov spaces. More precisely, set $K_{*,-1}\coloneqq\{0\}$ and define the depth-$d$ Krylov shell spaces
\begin{align}
    W_{*d}\coloneqq K_{*d}\ominus K_{*,d-1},\qquad d\geq 1; \qquad  W_{*,-1}\coloneqq\{0\}.
\end{align}

Choose an orthonormal basis for each $W_{*d}$, and let $\{g_j\}_j$ denote the union of these bases. Without loss of generality, we may assume that $\{g_j\}_j$ spans the whole space $W_*$.\footnote{If $\{g_j\}_j$ does not span the full space $W_*$, let
$M$ be the orthogonal complement of $\operatorname{span}\{g_j\}_j$
in $W_*$. Since $M$ is invariant under $E_*$ and is orthogonal to all
hybridization vectors, there is no coupling in the Hamiltonian between
the modes in $M$ and the remaining modes. Moreover, since $H_*$ has
even parity, under the corresponding
Fock-space factorization, the Hamiltonian therefore decomposes as
$H_*=H_{\mathrm{active}}\otimes \I
+\I\otimes \dE(E_*|_M).$
Hence the modes in $M$ form a decoupled free-fermion sector that can be
treated separately, and we assume $M=\{0\}$ for simplicity.}
Let $V$ denote the isometry whose columns are the basis vectors constructed above,
\begin{align}
    V=[\ldots\,\, ,g_j, \,\,\ldots]
\end{align}

Then  we have
\begin{lemma}[One-dimensional structure]\label{lem:krylov_1D}
We have  $\dim(W_{*d})\leq 2m$ for all $d\geq 0$. Besides, 
\begin{align}
    E_*W_{*d}
    \subseteq
    W_{*,d-1}\oplus W_{*d}\oplus W_{*,d+1}. 
\end{align}
Thus, when the columns of $V$ are ordered according to $W_{*0},W_{*1},\ldots$, the matrix $V^\dagger E_*V$ has the block-tridiagonal form in Eq.~\eqref{eq:block-tridiag}.
\end{lemma}
\begin{proof} 
The dimension bound for $W_{*d}$ comes directly from its definition and the fact that the dimension of $W_{*0}$ is smaller than $2m$.  
To prove the inclusion, first 
note that by definition, $W_{*d}\subseteq K_{*d}$, and hence 
$E_*W_{*d}\subseteq K_{*,d+1}.$ 

For $d\geq 1$, let $u\in W_{*d}$ and $v\in K_{*,d-2}$. Since $E_*$ is Hermitian, we have $
\langle v,E_*u\rangle = \langle E_*v,u\rangle. $
But $E_*v\in K_{*,d-1}$, while $u\in W_{*d}$ is orthogonal to $K_{*,d-1}$. Therefore 
\begin{align} 
\langle v,E_*u\rangle=0. 
\end{align} 

Thus $E_*W_{*d}$ is contained in $K_{*,d+1}$ and orthogonal to $K_{*,d-2}$, which proves the Lemma for $d\geq 1$.   The case $d=0$ follows directly from $E_*W_{*0}\subseteq K_{*1}$.
Finally, with respect to the decomposition
\[
    W_*=\bigoplus_{d\geq 0}W_{*d},
\]
the inclusion above implies that $V^\dagger E_*V$ is block-tridiagonal, as in Eq.~\eqref{eq:block-tridiag}.
\end{proof}

\subsection{Bandwise Krylov basis}\label{sec:band_krylov}

Recall that $\omega$ denotes the lower bound for single-particle energy after preprocessing, i.e.
$$0<   \omega \I \leq E_0 \leq \es \I.$$

In Section \ref{sec:1D_from_block_tri},
for the Hamiltonian with placeholder notation in Eq.~(\ref{eq:hyb_placeholder}), 
 we explained 
how to construct a single-particle basis for $W_*$ so that $H_*$ is transformed to the desired one-dimensional structure. 
If we set $E_*=E_0$, $\mu_{j*}=\mu_j$, and $W_*=W_{\mathrm{bath}}$, we could get a single-particle basis $\{g_j\}_j$ (the \textit{Krylov basis}) that  transforms the impurity Hamiltonian $H$ into a one-dimensional structure, as illustrated in Figure~\ref{fig:1D}(a). 
However, the resulting complexity of the classical simulation is exponential in $\omega^{-1}$, and is therefore inefficient when $\omega$ is small.

We instead use a bandwise Krylov space, which treats different energy bands of $E_0$ separately.
Partition the spectrum of $E_0$ into logarithmically spaced bands,
\begin{align}
    [\omega_\ell,2\omega_\ell),\qquad
    \omega_\ell=2^\ell\omega,\qquad
    0\leq\ell\leq L\coloneqq
    \left\lceil\log_2\frac{\es}{\omega}\right\rceil.
    \label{eq:L}
\end{align}

Let $W_\ell$ be the space spanned by eigenvectors of $E_0$
with eigenvalues in $[\omega_\ell,2\omega_\ell)$, and let
$\Pi_\ell$ be the corresponding projector.
Since $\sum_\ell \Pi_\ell$ is the identity map in the bath single-particle space, one could rewrite Eq.~(\ref{eq:impurity_model_hyb}) 
as
\begin{equation}
    H
    = K_0+ 
    \sum_\ell\left( \dE( \Pi_\ell E_0)
    +
    \sum_{j=1}^{2m}
    \left(
        a^\dagger(\Pi_\ell \mu_j)O_j
        +
        O_j^\dagger a(\Pi_\ell \mu_j)
    \right)\right).\label{eq:hyb_strength_original}
\end{equation}

Applying the constructions in Section~\ref{sec:1D_from_block_tri}
to every $\ell$, with
\begin{align}
    E_*\coloneqq E_{(\ell)}\coloneqq\Pi_\ell E_0,\qquad
    \mu_{j*}\coloneqq\Pi_\ell\mu_j,\qquad
    W_*\coloneqq W_\ell,
    \label{eq:hyb_bandwise}
\end{align}
we denote the resulting Krylov shell spaces and their orthonormal bases by
\begin{align}
\text{Krylov shell spaces:}\,\,\, W_{\ell,d}, \,\,\,d\geq 0;\qquad \text{Orthonormal basis of } W_\ell=\bigoplus_{d\geq 0} W_{\ell d}\text{:}\,\,\, \{g_{\ell j}\}_j.\label{eq:Wld}
\end{align}

For each band $\ell$, the basis vectors $\{g_{\ell j}\}_j$ are ordered so that those spanning $W_{\ell,0}$ come first, followed by those spanning $W_{\ell,1}$, then $W_{\ell,2}$, and so on.  Combining the bases from all bands,  we obtain the orthonormal basis   $\{g_{\ell j}\}_{\ell,j}$ for the entire bath single-particle space.

Define $V_\ell=[g_{\ell1},\ldots,g_{\ell j},\ldots]$. Then $V_\ell^\dagger E_{(\ell)} V_\ell$ has the block-tridiagonal form given in Eq.~(\ref{eq:block-tridiag}). Thus the impurity model $H$ in terms of $\{a(g_{\ell j})\}_{\ell j}$ has a bandwise one-dimensional structure as in Figure \ref{fig:1D} (b).

We call this basis $\{f_j\}_j\coloneqq\{g_{\ell j}\}_{\ell j}$ the \textit{bandwise Krylov basis}.

\subsection{Bandwise hybridization form with an enlarged impurity} 

\paragraph{Relabel annihilation operators.}\label{sec:enlarged_impurity}

As in Figure~\ref{fig:1D}(b), we enlarge the impurity to include the first Krylov shell in each band. In other words, we call
$\{a(e_j)\}_{j=1}^m\cup\{a(g_{\ell j})\}_{\ell;\,j\leq \dim(W_{\ell 0})}$
the \textit{enlarged impurity modes}, and call the remaining operators
$\{a(g_{\ell j})\}_{\ell;\,j>\dim(W_{\ell 0})}$
the \textit{residual bath modes}. We explain the reason for this partition below. To simplify the notation, throughout the remainder of the manuscript, we relabel the enlarged-impurity modes as $\{b_k\}_{k\geq 1}$ and the residual-bath modes as $\{a_i\}_{i\geq 1}$. Accordingly, we also represent vectors in the enlarged-impurity single-particle space $W_{\mathrm{imp}}\bigoplus_\ell W_{\ell,0}$ and the residual-bath single-particle space $\bigoplus_{\ell;\,d\geq 1} W_{\ell,d}$ by viewing the basis $\{e_j\}_{j=1}^m\cup\{g_{\ell j}\}_{\ell,j}$ as the standard basis. For any vector $v$ in the enlarged impurity single-particle space,  define $b(v)\coloneqq\sum_k v_k b_k$.

The reason for enlarging the impurity is to reduce the cost of our classical simulation algorithm, which will heavily rely on the interaction strength between the ``impurity'' part and the ``bath'' part. Note that we do not have much control of the interaction strength between the original impurity and the bath part:  $\|\Pi_\ell \mu_j\|$ in Eq.~(\ref{eq:hyb_strength_original}) need not be small relative to $\omega_\ell$, since $\mu_j$ could be arbitrary.

In contrast, we have good control of the interaction strength between the first and second bath blocks in Figure \ref{fig:1D}(b), which is the interaction strength between the enlarged impurity and the residual bath. Let $s_{\ell 0}$ and $s_{\ell 1}$ be the dimensions of the first two Krylov shells $W_{\ell 0}$ and $W_{\ell 1}$, respectively.  
The corresponding interaction terms in the Hamiltonian $H$ can be written as
\begin{align}
    \sum_{\ell=0}^L \sum_{j=1}^{r_\ell}
    \sigma_{\ell j}
    \left(
        a^\dagger(u_{\ell j})b(v_{\ell j})
        +
        b^\dagger(v_{\ell j})a(u_{\ell j})
    \right),
    \qquad
    |\sigma_{\ell j}|\leq 2\omega_\ell,
    \label{eq:hyb_first_second}
\end{align}
where $\{u_{\ell j}\}_j$ and $\{v_{\ell j}\}_j$ are orthonormal vectors in $W_{\ell 1}$ and $W_{\ell 0}$, respectively, and $r_\ell\leq 2m$: 
note that $V_\ell^\dagger E_{(\ell)} V_\ell$ has the block-tridiagonal form in Eq.~\eqref{eq:block-tridiag}, where the off-diagonal block $Q_{\ell 0}$ connecting $W_{\ell 0}$ and $W_{\ell 1}$ satisfies
    $\|Q_{\ell 0}\|\leq \|E_{(\ell)}\|\leq 2\omega_\ell.$
Since $\dim(W_{\ell 0})\leq 2m$, the quantity $r_\ell\coloneqq\rank(Q_{\ell 0})$ satisfies the same bound.    
Taking the singular-value decomposition
\begin{align}
    Q_{\ell 0}
    =
    \sum_{j=1}^{r_\ell}
    \sigma_{\ell j}
    u_{\ell j}v_{\ell j}^\dagger, \qquad r_\ell\leq 2m,
\end{align}
where 
$\sigma_{\ell j}\leq\|Q_{\ell 0}\|\leq2\omega_\ell$, 
we obtain Eq.~\eqref{eq:hyb_first_second}.

Putting these preprocessing steps together, we arrive at the following bandwise hybridization form with enlarged impurity, which we use throughout the remainder of the manuscript.

\begin{center}
\fbox{%
\begin{minipage}{0.95\linewidth}

\paragraph{Bandwise hybridization form with enlarged impurity.}
With the relabeled operators $\{b_k\}_k\cup\{a_j\}_j$ and the decomposition into enlarged-impurity and residual-bath modes illustrated in Figure~\ref{fig:1D}(b), the corresponding terms of the Hamiltonian can be grouped into the following \textit{bandwise hybridization form with enlarged impurity}:
\begin{equation}
    \label{eq:hyb_decom}
\begin{aligned}
&H=K+\dE(E)+\sum_{\ell=0}^L \sum_{j=1}^{r_\ell}
    \left(
        a^\dagger(u_{\ell j})G_{\ell j}
        +G_{\ell j}^\dagger a(u_{\ell j})
    \right), \qquad r_\ell\leq 2m,\\
    &E=\bigoplus_{\ell=0}^L E_\ell,\qquad
    \omega_\ell\I\leq E_\ell\leq 2\omega_\ell\I,\qquad\text{$E_\ell$ acts on $W_{\ell}\ominus W_{\ell 0}$,}\\
    &G_{\ell j}\coloneqq\sigma_{\ell j}b(v_{\ell j}),
    \qquad
    \|G_{\ell j}\|\leq 2\omega_\ell. 
\end{aligned}
\end{equation}
Here $K$ and $E$ are the corresponding parts of the Hamiltonian under the decomposition in Figure~\ref{fig:1D}(b), where
$K$ corresponds to terms that act only on the enlarged-impurity modes and has even parity, while the quadratic term $\dE(E)$, with Hermitian matrix $E$, acts only on the residual bath. Moreover, as  in Figure~\ref{fig:1D}(b), $E$ can be decomposed as a direct sum over the band index $\ell$,
$E=\bigoplus_{\ell=0}^L E_\ell.$
Here $E_\ell$ should be distinguished from the previously defined $E_{(\ell)}$: $E_{(\ell)}=\Pi_\ell E_0$ is the quadratic bath matrix in band $\ell$ before enlarging the impurity, while $E_\ell$ is the block of $E_{(\ell)}$ acting within the residual-bath subspace $W_\ell\ominus W_{\ell,0}$. The vectors $\{u_{\ell j}\}_j$ are orthonormal vectors in  $W_{\ell 1}$, where $W_{\ell 1}$ is the Krylov shell at band $\ell$ and depth $1$.
The term $a^\dagger(u_{\ell j})G_{\ell j}$ and its Hermitian conjugate form the new hybridization term described in Eq.~\eqref{eq:hyb_first_second}. 

\end{minipage}%
}
\end{center}

\paragraph{Quantities that control the final complexity.}
For reference, we list below the key quantities that control the complexity of our algorithm, where $\xi_\ell$ denotes the interaction strength between the enlarged impurity and the residual bath in band $\ell$:
   \begin{align}
   \xi_\ell &\coloneqq\left\| \sum_{j=1}^{r_\ell} G^\dagger_{\ell j}G_{\ell j}\right\|\leq 8m \omega_\ell^2,\label{eq:xi}\\
\kappa &\coloneqq\sum_{\ell=0}^L\frac{\xi_\ell}{\omega_\ell^2}  \leq 8m (L+1) = 8m \left(1+\left\lceil \log_2 \frac{\es}{\omega} \right\rceil\right). \label{eq:kappa}
   \end{align}

We also define a parameter that bounds the enlarged-impurity size together with the number of bath modes at each depth, summed over all bands:
\begin{align}
\chi\coloneqq\dim \left(W_{\mathrm{imp}} \oplus \bigoplus_{\ell=0}^L W_{\ell,0}\right)+
\max_{d\geq 1}
\dim\left(
\bigoplus_{\ell=0}^L W_{\ell,d}
\right)\leq m + 4m (L+1).\label{eq:chi}
\end{align}

\subsection{Weighted positivity of the bath in the bandwise Krylov basis}

The key property of the bandwise Krylov basis is the one-dimensional structure established in Lemma~\ref{lem:krylov_1D}, which implies the following weighted-positivity property of the bath. For any matrix $K$, define $\Re(K)=(K+K^\dagger)/2$.

\begin{lemma}[Weighted positivity of the bath]
\label{lem:weighted-energy} 
Let $\{h_d\}_{d\geq 1}$ be real numbers satisfying $h_1=0$ and
$|h_{d+1}-h_d|\leq 1$. For each band $\ell$, let $E_\ell$ be the bandwise bath matrix in Eq.~(\ref{eq:hyb_decom}), and let $\Pi_{\ell,d}$ be the  projector onto the Krylov shell $W_{\ell,d}$ defined in Eq.~(\ref{eq:Wld}). Define the depth-weighted operator
\begin{align}
    D_\ell\coloneqq\sum_{d\geq1} h_d\Pi_{\ell,d}.
\end{align}
Then, for every band $\ell$ and every $0<\lambda\leq1/4$ (by default we set $\lambda=1/4$),
\begin{align}
    \Re\!\left(
        e^{-\lambda D_\ell}E_\ell e^{\lambda D_\ell}
    \right)
    \geq\frac{\omega_\ell}{2}\I.
\end{align}
\end{lemma}

\begin{proof}
By the one-dimensional structure property in Lemma \ref{lem:krylov_1D}, we can decompose $E_\ell$ as a sum of diagonal and off-diagonal blocks,
\begin{align}
    E_\ell = \sum_{d\geq 1} \Pi_{\ell d} E_\ell \Pi_{\ell d} +  \sum_{d\geq 1} \left(\Pi_{\ell d} E_\ell \Pi_{\ell, d+1} + \Pi_{\ell, d+1} E_\ell \Pi_{\ell d}\right).
\end{align}
 Define $\eta_d
\coloneqq
\cosh\!\left(\lambda(h_{d+1}-h_{d})\right)-1$. 
Using the notation $\Reop(A)=(A+A^\dagger)/2$, one can check that 
\begin{align}
\Re\!\left(
e^{-\lambda D_\ell}E_\ell e^{\lambda D_\ell}
\right)
=
E_\ell+  R_\ell,\qquad
R_\ell\coloneqq\sum_{d\geq1}
\eta_d
\left(
\Pi_{\ell,d}E_\ell\Pi_{\ell,d+1}
+
\Pi_{\ell,d+1}E_\ell\Pi_{\ell,d}
\right),\label{eq:weighed_positive_bath}
\end{align}
Since $|h_{d+1}-h_{d}|\leq1$ and $\cosh \lambda -1\leq \lambda^2$ for $0\leq \lambda\leq 1/4$, we have $0\leq \eta_d \leq \lambda^2.$

Besides, recall that $\|E_\ell\|\leq 2\omega_\ell$; thus, based on Eq.~(\ref{eq:weighed_positive_bath})  we have  
$\|R_\ell\|\leq 4 \lambda^2 \omega_\ell.$
Recall that $ E_\ell\geq \omega_\ell \I$; thus, using $0\leq \eta_d \leq \lambda^2$, we conclude that 
\[
 \Re\!\left(e^{-\lambda D_\ell}E_\ell e^{\lambda D_\ell}\right)
 \ge(1-4\lambda^2)\omega_\ell\I
 \ge\frac{\omega_\ell}{2}\I.
\]
where in the last inequality we use $0<\lambda\leq 1/4$.
\end{proof}

\section{Compression lemma and recursion for occupation statistics}
\label{sec:compress}

Let $\rho$ denote the state of interest, which will later be taken to be either a ground state or the relevant part of the thermal state of an impurity model. In this section, we develop a compression lemma showing that, although $\rho$ is in principle defined on a Hilbert space of dimension $2^n$, it can be compressed to a much smaller subspace.

Recall from Section~\ref{sec:band_krylov_rep} that $\{b_k\}_{k\geq 1}\cup\{a_i\}_{i\geq 1}$ are the annihilation operators for the enlarged impurity and the residual bath, respectively. For concreteness, we give an ordering of these $n$ annihilation operators by listing $a_1,a_2,\ldots$ first, followed by $b_1,b_2,\ldots$, and let $\{\ket{x}:x\in\{0,1\}^n\}$ denote the corresponding Fock basis. We write
$$
x\sim \rho,
$$
to denote sampling from the distribution over $x\in \{0,1\}^n$ obtained by measuring $\rho$ in this Fock basis.
Our main strategy is to control the \textit{occupation statistics} of the \textit{residual-bath modes}, namely, the statistics of the bits in $x$ corresponding to $\{a_i\}_{i\geq 1}$.

In Section~\ref{sec:occup_stat_support}, we first explain the \textit{weighted} occupation statistics and how they can be used to control the support of the state of interest. In Section~\ref{sec:occup_stat_notation}, we introduce more formal notation. Then, in Sections~\ref{sec:pull_through} and~\ref{sec:Fq_recursion}, we use the commutator identity for impurity models to derive a recursion formula for estimating the weighted occupation statistics.

\subsection{Bounding the support via weighted occupation statistics}\label{sec:occup_stat_support}
 
Let $\bS_{\mathrm{bath}}\subseteq\{1,\ldots,n\}$ denote the index set of the bath modes $\{a_i\}_i$, that is 
$$
\{a_i\}_i = \{a_i\}_{i\in \bS_{\mathrm{bath}}}. 
$$
Define the occupation number operator for the $i$th bath mode as
\begin{align}
n_i\coloneqq a_i^\dagger a_i=\op{1}{1}_i.
\end{align}
The statistics of $n_i$ can be used to control the support of $\rho$. For example, if $\tr\left(\rho\left(\sum_i n_i\right)\right)$ is small, then by Markov's inequality, $\rho$ is approximately supported on Fock basis states $\ket{x}$ with small Hamming weight $|x|$.  Note that the enlarged impurity contains only \(O(m(L+1))\) modes.

For our purpose of polynomial-time classical simulation, we carry out a more refined analysis based on \textit{weighted occupation statistics} and their \textit{exponential moments}.

\paragraph{Weighted occupation number.} As in Figure \ref{fig:1D}(b), recall that the  residual-bath single-particle space is partitioned into different bands, and each band is partitioned into Krylov shell spaces, i.e.  $W=\bigoplus_{\ell=0}^L W_{\ell}$ and $W_\ell\coloneqq\bigoplus_{d\geq 1}W_{\ell d}$.  We refer to the subscripts in $W_{\ell,d}$ as ``band''  $\ell$  and ``depth'' $d$.   For each residual-bath mode $a_i$ whose corresponding single-particle basis vector belongs to $W_{\ell,d}$, we define its band and depth as  $\ell(i)\coloneqq\ell$ and $d(i)\coloneqq d$.  

We weight the occupation statistics according to the depth of each mode. Fix nonnegative numbers $\{h_d\}_{d\geq1}$ with 
\begin{align}
    h_1=0, \quad |h_{d+1}-h_d|\leq1, \quad h_d\geq 0.
\end{align} In the later sections, we take $h_d=d-1$ for the ground-state argument and
$h_d=\min\{d-1,\cut\}$ for the thermal-state argument, where $\cut\geq 1$ is a cutoff parameter set later.

For any $x\in\{0,1\}^n$, define the weighted total occupation number $D(x)$ and the corresponding depth-weighted occupation-number operator $\mathsf{D}$ on the Fock space by
\begin{align}
D(x)\coloneqq\sum_{i\in \bS_{\mathrm{bath}}} h_{d(i)}x_i, \qquad  \mathsf{D}=\sum_{i\in \bS_{\mathrm{bath}}} h_{d(i)} n_i. \label{eq:D}
\end{align}
Recall that $\Pi_{\ell d}$ is the projector onto the single-particle subspace $W_{\ell d}$. For convenience, for each band $\ell$, we also define the corresponding depth-weighted operator on the single-particle space by
\begin{align}
D_\ell=\sum_{d\geq1}h_d\Pi_{\ell,d}.
\end{align}

\paragraph{Exponential moments and joint occupation.} 
The key quantity we estimate is the exponential moment of the weighted occupation  number.
Our goal is to estimate 
\begin{align}
	&\mathbb{E}_{x\sim \rho}\left( e^{2 \lambda D(x)} \right) = \tr\left(\rho e^{2\lambda \mathsf D }\right), \qquad\lambda\coloneqq1/4.
\end{align}
The key fact is that an upper bound on $\tr\left(\rho e^{2\lambda \mathsf D }\right)$ leads to an upper bound on the support of $\rho$:

\begin{theorem}[Exponential moment to size of support]
\label{thm:exp_moment_poly_support}
Suppose $h_d\geq \min\{d-1,\cut\}$ and set
$\cut= \max\{1,\frac18\log n\}$. Then, for any $0<\delta<1$, define a subset $\bS$
of the Fock configurations by
\begin{align}
	&\bS
\coloneqq
\left\{
x:
D(x)\leq
T
\right\}, \quad \text{for}\quad T=4\log(2\delta^{-1})
+
2\log \tr \left(\rho e^{\mathsf D/2}\right).
\end{align}
Denote the corresponding subspace by
$\mathcal H_{\bS}
\coloneqq
\operatorname{span}\{\ket{x}:x\in\bS\}$, and let $P_\bS$ denote the projection onto $\cH_\bS$. Then
\begin{align}
	\tr\left((\I-P_\bS)\rho\right)\leq \frac{\delta^2}{4},\qquad   \dim(\cH_\bS)
\leq
\exp\!\left[
\mathcal O\left(
m\log \frac{2\es}{\omega}+\log\delta^{-1}
+\log\tr\left(\rho e^{\mathsf D/2}\right)
\right)
\right].
\end{align}
Moreover,  the set $\bS$ can be enumerated in time $\cO(n|\bS|)$ by a classical algorithm.
Note that $D(x)$ involves only residual-bath occupations; thus $P_\bS$ acts
trivially on the enlarged-impurity modes.
\end{theorem}

\begin{proof}
Suppose we measure $\rho$ in the Fock basis and obtain $x\sim \rho$. By Markov's inequality, for $\lambda=1/4$,
\begin{align}
    \Pr_{x\sim\rho}\!\left(
        e^{2\lambda D(x)}>e^{2\lambda T}
    \right)\leq
    e^{-2\lambda T}
    \mathbb E_{x\sim\rho}
    \left[e^{2\lambda D(x)}\right]=
     e^{-2\lambda T}\tr\left(\rho e^{2\lambda \mathsf D }\right).
\end{align}

By the choice of $T$, we have
$e^{-2\lambda T}\tr\left(\rho e^{\mathsf{D}/2}\right) \leq \delta^2/4$.
Thus,
$\tr\left((\I-P_\bS)\rho\right)\leq \delta^2/4$.
The bound on the size of $\bS=\{x:D(x)\leq T\}$ and the algorithm to enumerate it follow from the direct combinatorial bound in Appendix~\ref{app:count_bound}, Lemma~\ref{lem:count_set}.
\end{proof}

Theorem \ref{thm:exp_moment_poly_support} implies that 
$\rho$ can be compressed, with inverse-polynomial error, to a polynomial-size subspace $\cH_{\bS}$, 
if $\omega= {1}/{\poly(n)}$ and $\tr(\rho e^{\mathsf{D}/2})=\poly(n)$.
For the task of estimating the ground-state energy, we can always assume that $\omega\geq {1}/{\poly(n)}$, as explained in Section~\ref{sec:ground_energy_classical}. For tasks related to thermal states, we will only apply Theorem \ref{thm:exp_moment_poly_support} to the thermally gapped region where $\beta \omega\geq 9\log n$. The rest of this section, together with Sections~\ref{sec:compress_ground} and~\ref{sec:compress_thermal}, is devoted to estimating $\tr(\rho e^{\mathsf{D}/2})$.

\subsection{Notation for joint occupations}\label{sec:occup_stat_notation}

Let $I=(i_1,\ldots,i_q)$ be an increasing list of bath-mode
indices, with $i_1<\cdots<i_q$. 
To simplify notation, from now on, we abbreviate $\sum_{i\in \bS_{\mathrm{bath}}}$ as $\sum_i$ and $\sum_{I\subseteq  \bS_{\mathrm{bath}}}$ as $\sum_I$. We also write
\begin{align}
 A_I\coloneqq a_{i_q}\cdots a_{i_1},\qquad A_\varnothing\coloneqq\I,
 \qquad
 h(I)\coloneqq\sum_{i\in I}h_{d(i)}, \qquad \rho_I\coloneqq e^{2\lambda h(I)}A_I\rho A_I^\dagger.
 \label{eq:list-data}
\end{align}

\paragraph{Estimate the exponential bound by joint occupation.}
In Section \ref{sec:occup_stat_support}, Theorem \ref{thm:exp_moment_poly_support}, we reduce the problem of compressing a state to bounding the exponential moment $\tr\left(\rho e^{2\lambda \mathsf D }\right)$. We further bound this quantity in terms of the joint occupation statistics $\tr(\rho \prod_{i\in I}n_{i})$. Notice that

\begin{align}
	\tr\left(\rho e^{2\lambda \mathsf D }\right)\leq\sum_{I\subseteq \bS_{\mathrm{bath}}}  e^{2\lambda h(I)}  \tr\left(\rho  \prod_{i\in I} n_i\right), \label{eq:joint_occupation}
\end{align}
where the inequality holds since 
	$e^{2\lambda \mathsf D}=\prod_{i}  \left(\I + (e^{2\lambda h_{d(i)}}-1)n_i\right ) \leq \sum_{I}  e^{2\lambda h(I)} \prod_{i\in I}n_i.$

For convenience, we group the indices in $I$ by band. For each $\ell=0,\ldots,L$, let
$q_\ell\coloneqq\bigl|{i\in I:\ell(i)=\ell}\bigr|$
denote the number of indices in $I$ belonging to band $\ell$.
Define the  (band-grouped) weighted  joint occupation  as
\begin{align}
F_{q_0,\ldots,q_L}
\coloneqq
\sum_{\substack{I:\,|\{i\in I:\ell(i)=\ell\}|=q_\ell\\
                    \text{for every }\ell}}
e^{2\lambda h(I)}
\tr\!\left(\rho\prod_{i\in I}n_i\right).
\label{eq:Fq}
\end{align}

To simplify notation, we use the notation
\begin{align}
    \bq\coloneqq(q_0,\ldots,q_L),
    \qquad
    \bq-\bel\coloneqq(q_0,\ldots,q_\ell-1,\ldots,q_L),\qquad I\sim\bq,
\end{align}
where $\bel$ is the unit vector in the $\ell$-th coordinate, and
 $I\sim\bq$
means that $I$ is an increasing list consistent with the prescribed band
occupations, namely,
$|\{i\in I:\ell(i)=\ell\}|=q_\ell$ for every $\ell$.   Using this notation, Eq.~(\ref{eq:joint_occupation}) and Eq.~(\ref{eq:Fq}) become
\begin{align}
    \tr\left(\rho e^{2\lambda \mathsf D }\right)\leq \sum_{\bq} F_\bq,\qquad\qquad F_\bq = \sum_{I\sim \bq} \tr(\rho_I).\label{eq:exp_Fq}
\end{align}

\paragraph{Substitution of modes.}
For convenience,  for $j\notin I\setminus\{i\}$, we define $I^{(i\to j)}$ to be the increasing list
obtained by replacing $i$ with $j$.

\paragraph{Remark.}
In this section we use the notation $ F_\bq, D_\ell \ldots$. In Sections~\ref{sec:ground_state_sim} and~\ref{sec:thermal_state_sim}, we distinguish the corresponding quantities for the ground state and thermal state by adding superscripts $(\infty)$ and $(\beta)$, respectively, e.g., $F_\bq^{(\infty)}$ and $F_\bq^{(\beta)}$.

\subsection{The commutator identity}\label{sec:pull_through}

We estimate $F_\bq$ by deriving a recursion that relates it to the lower-order quantities $F_{\bq-\bel}$. The key ingredient is the commutator identity in Lemma \ref{lem:commutator-identity}, an  algebraic identity for the commutator $[A_I,H]$ that exploits the structure of impurity models.

To begin, we first compute the commutator $[H,a_i]$. Recall that $\ell(i)$ denotes the band containing mode $i$. Recall that $e_i$ corresponds to mode $a_i$ and is a standard basis vector in $W_{\ell(i),d(i)}$. 

\begin{lemma}
\label{lem:elementary-commutators}
For any residual-bath mode $a_i$,
\begin{align}
 [H,a_i]
 ={}&-\sum_{j:\ell(j)=\ell(i)}
 \langle e_j,E_{\ell(i)}e_i\rangle a_j
 -\sum_{s}
 \langle u_{\ell(i)s},e_i\rangle G_{\ell(i)s}.
 \label{eq:H-ai}
\end{align}
In particular, $[H,a_i]$ does not increase the degree in the residual-bath
operators.
\end{lemma}

\begin{proof} Recall the decomposition of $H$ given in Eq.~\eqref{eq:hyb_decom}.
The operator $K$  acts only on the enlarged impurity and has even parity, so it
commutes with every residual-bath annihilation operator. Moreover, the canonical anticommutation relation gives
\begin{align}
 &[\dE(E),a_i]=-a(Ee_i),
 \label{eq:bath-commutator}\\
 &[a^\dagger(u_{\ell s})G_{\ell s}
   +G_{\ell s}^\dagger a(u_{\ell s}),a_i]
 =-\langle u_{\ell s},e_i\rangle G_{\ell s}.
\end{align}
Note that $E=\oplus_{\ell=0}^L E_\ell$ leaves vectors in each
band invariant; thus $E e_i =E_{\ell(i)} e_i \in W_\ell$. Besides, recall that $u_{\ell s}$ lies in band $\ell$. Thus we have
\begin{align}
 [H,a_i]
 &=[\dE(E),a_i]
 +\left[
 \sum_{\ell} \sum_{s} \left(
 a^\dagger(u_{\ell s})G_{\ell s}
 +G_{\ell s}^\dagger a(u_{\ell s})
 \right),a_i
 \right]\notag\\
 &=-\sum_{j:\ell(j)=\ell(i)}
 \langle e_j,E_{\ell(i)}e_i\rangle a_j
 -\sum_{s} \langle u_{\ell(i)s},e_i\rangle G_{\ell(i)s}.
\end{align}
\end{proof}

We now derive the key equation, which is called the commutator identity.

\begin{lemma}[Commutator identity]
\label{lem:commutator-identity}
For every nonempty increasing list $I=(i_1<\cdots<i_q)$,
\begin{align}
    [A_I,H]=B_I+C_I,
    \label{eq:commutator-identity-short}
\end{align}
where \(B_I\) denotes the bath-replacement contribution obtained by replacing a residual-bath annihilation operator \(a_i\) in \(A_I\) with another from the same band, whereas \(C_I\) denotes the hybridization contribution obtained by replacing \(a_i\) with an annihilation operator acting on the enlarged impurity. Explicitly,
\begin{align}
 B_I
 &\coloneqq
 \sum_{i\in I}
\sum_{\substack{j:\ell(j)=\ell(i)\\
 j\notin I\setminus\{i\}}}
 \varepsilon_I(i,j)
 \langle e_j,E_{\ell(i)}e_i\rangle
 A_{I^{(i\to j)}},
 \label{eq:BI}\\
 C_I
 &\coloneqq
 \sum_{i\in I}\eta_I(i)
 \sum_{s=1}^{r_{\ell(i)}}
 \langle u_{\ell(i)s},e_i\rangle
 G_{\ell(i)s}A_{I\setminus\{i\}}.
 \label{eq:CI}
\end{align}
Here $\varepsilon_I(i,j)$ and $\eta_I(i)\in\{\pm1\}$ are the signs arising from restoring the standard decreasing order of the annihilation operators. In particular, $\varepsilon_I(i,j)$ is obtained by replacing $a_i$ with $a_j$ in $a_{i_q}\cdots a_{i_1}$ and reordering the resulting product into decreasing-index order.
\end{lemma}

The proof of Lemma~\ref{lem:commutator-identity} follows from the fact that $[A_I,H]
=\sum_{r=1}^q
a_{i_q}\cdots
[a_{i_r},H]\cdots a_{i_1}$  and an application of Lemma~\ref{lem:elementary-commutators}. We give the detailed calculation in Appendix~\ref{app:commutator-identity}.

\paragraph{Triangular structure and recursion strategy.}
Note that the commutator identity $[A_I,H]=B_I+C_I$ starts from an $A_I$ containing $|I|$ residual-bath annihilation operators. The bath-replacement term $B_I$ contains monomials $A_{I^{(i\rightarrow j)}}$ of the same degree, $|I^{(i\rightarrow j)}|=|I|$, whereas $C_I$ contains $G_{\ell(i)s}A_{I\setminus\{i\}}$, with strictly smaller residual-bath degree $|I\setminus{i}|=|I|-1$. This triangular structure will allow us to derive a recursion for $F_\bq$ in terms of lower-order quantities.

More precisely, the recursion will be based on the 
following direct corollary of the commutator identity. For any nonempty increasing list $I$,
define the commutator contribution as
\begin{align}
    \Delta_I \coloneqq \Re\tr\!\left(
 \rho A_I^\dagger[H,A_I]
 \right).
\end{align}

\begin{corollary}\label{cor:commutator-identity}
Fix nonzero $\bq=(q_0,\ldots,q_L)$ and consider $I\sim \bq$.
Multiplying \eqref{eq:commutator-identity-short} by
$e^{2\lambda h(I)}\rho A_I^\dagger$, taking the real part of the trace,
and summing over $I\sim\bq$, we have
\begin{align}
 -\sum_{I\sim\bq}e^{2\lambda h(I)} \Delta_I  = \sum_{I\sim\bq}e^{2\lambda h(I)}
 \Re\tr(\rho A_I^\dagger B_I)
 +\sum_{I\sim\bq}e^{2\lambda h(I)}
 \Re\tr(\rho A_I^\dagger C_I)
\label{eq:key_observation_pull_through}
\end{align}
\end{corollary}

Based on the above corollary, in the next section, we lower bound the contribution of $B_I$ in terms of $F_\bq$  in Lemma~\ref{lem:lower_bath}, using the weighted positivity of the bath from Lemma~\ref{lem:weighted-energy}.
In Lemma~\ref{lem:upper_source}, we upper bound the contribution of $C_I$ in terms of $F_\bq$ and the lower-order quantities $F_{\bq-\bel}$, by two straightforward applications of the Cauchy--Schwarz inequality, with the norm bound in Eq.~(\ref{eq:xi}) controlling one of the resulting factors.
The commutator contribution $\Delta_I$ is the only term whose estimate requires specific properties of the state $\rho$; it will therefore be treated separately for the ground state and the thermal state in Sections~\ref{sec:ground_state_sim} and~\ref{sec:thermal_state_sim}, respectively.

\subsection{A recursion formula for \texorpdfstring{$F_{\bq}$}{Fq} with commutators}\label{sec:Fq_recursion}

In this subsection, we use the commutator identity to derive a recursion
for the target quantity $F_\bq$. To simplify the notation, for any fixed
increasing list $J=(j_1<\cdots<j_k)$, set
\begin{align}
 \Phi_{J,i}\coloneqq
 e^{\lambda(h(J)+h_{d(i)})}a_iA_J,
 \qquad i\notin J.
\end{align}
Thus $\Phi_{J,i}=\pm e^{\lambda h(J\cup\{i\})}A_{J\cup\{i\}}$, where the
sign comes from reordering $a_i a_{j_k}\cdots a_{j_1}$ into the standard
decreasing-index order.

For fixed $J$ and band $\ell$, let $M^{J\ell}$ denote the single-particle correlation matrix on the modes in band $\ell$ that are not contained in $J$, 
\begin{align}
 M_{ij}^{J\ell}\coloneqq\tr\!\left(
 \rho\Phi_{J,i}^\dagger\Phi_{J,j}
 \right), \qquad i,j\notin J, \,\,\ell(i)=\ell(j)=\ell.\label{eq:MJell}
\end{align}

Extend $M^{J\ell}$ to a matrix on the full band-$\ell$ single-particle space by setting the remaining entries to zero. As a correlation matrix, $M^{J\ell}$ is positive semidefinite, which can also be checked directly from the definition. Therefore, the weighted positivity bound in Lemma~\ref{lem:weighted-energy} directly gives the following:

\begin{lemma}[Single-particle bath-energy bound]\label{lem:bath-energy}
For any fixed increasing list $J$ and band $\ell$,
\begin{align}
\tr\!\left[
\Re\!\left(e^{-\lambda D_\ell}E_\ell e^{\lambda D_\ell}\right)M^{J\ell}
\right]
\geq
\frac{\omega_\ell}{2}\tr(M^{J\ell})
=
\frac{\omega_\ell}{2}
\sum_{\substack{j:\ell(j)=\ell,\,\,j\notin J}}
\tr(\rho_{J\cup\{j\}}).
\end{align}
\end{lemma}

We can now use this bound to lower bound the bath-replacement contribution from $B_I$ in terms of $F_\bq$.

\begin{lemma}[Lower bound on the bath term]
\label{lem:lower_bath}
Fix nonzero $\bq=(q_0,\ldots,q_L)$. Then
\begin{align}
\sum_{I\sim\bq}e^{2\lambda h(I)}
\Re\tr(\rho A_I^\dagger B_I)
\geq
\left(\frac12\sum_\ell q_\ell\omega_\ell\right)F_\bq.
\label{eq:lower-bath}
\end{align}
\end{lemma}

\begin{proof}
For every $I\sim\bq$ and $i\in I$, set
$J=I\setminus\{i\}$. Consider another mode $j\notin I\setminus\{i\}$  in the same band as $i$. 
Since 
 $h(I)-h\bigl(I^{(i\to j)}\bigr)=h_{d(i)}-h_{d(j)}$,
by the definition of $\Phi_{J,i}$ and the fermionic sign
$\varepsilon_I(i,j)$, we have
\begin{align}
e^{2\lambda h(I)}\varepsilon_I(i,j)
\langle e_j,E_{\ell(i)}e_i\rangle
\tr\!\left(\rho A_I^\dagger A_{I^{(i\to j)}}\right)
=
\left\langle e_j,
e^{-\lambda D_{\ell(i)}}E_{\ell(i)}e^{\lambda D_{\ell(i)}}e_i
\right\rangle
\tr\!\left(\rho\Phi_{J,i}^\dagger\Phi_{J,j}\right).
\label{eq:weight-transfer-bath}
\end{align}

Then
expanding $B_I$, using Eq.~(\ref{eq:weight-transfer-bath}) and regrouping the
terms according to $J,\ell$, we obtain \begin{align}
&\sum_{I\sim\bq}e^{2\lambda h(I)}
\Re\tr(\rho A_I^\dagger B_I)=
\sum_{\ell:q_\ell>0}
\sum_{J\sim\bq-\bel}
\tr\!\left[
\Re\!\left(
e^{-\lambda D_\ell}E_\ell e^{\lambda D_\ell}
\right)M^{J\ell}
\right],
\label{eq:bath-regroup}
\end{align}
where  $M^{J\ell}$ is the correlation matrix defined in Eq.~(\ref{eq:MJell}).
Applying Lemma~\ref{lem:bath-energy}, we get 
\begin{align}
\sum_{I\sim\bq}e^{2\lambda h(I)}
\Re\tr(\rho A_I^\dagger B_I)
\geq
\sum_{\ell:q_\ell>0}
\sum_{J\sim\bq-\bel}
\frac{\omega_\ell}{2}
\sum_{\substack{j:\ell(j)=\ell\\j\notin J}}
\tr(\rho_{J\cup\{j\}})\label{eq:lower_bath_term_1}
\end{align}
Reindexing the sum by $I=J\cup\{j\}$, we can rewrite the right-hand side of the above equation as 
$\sum_{I\sim\bq}
\tr(\rho_I)
\sum_\ell \sum_{j\in I, \ell(j)=\ell}
\omega_\ell/2.$
Then we prove the lemma by noticing that  $\sum_{j\in I, \ell(j)=\ell} 1=q_\ell$ and $\sum_{I\sim\bq} \tr(\rho_I)=F_\bq$.
\end{proof}

Then we upper bound the contribution of $C_I$ in terms of $F_\bq$ and $F_{\bq-\bel}$. This requires just two applications of the Cauchy--Schwarz inequality. The quantity $\xi_\ell$ defined in Eq.~(\ref{eq:xi}) controls one of the resulting factors.
\begin{lemma}[Upper bound on the source term]
\label{lem:upper_source}
Assume $\omega>0$, so that $\omega_\ell=2^\ell\omega>0$. Fix nonzero
$\bq=(q_0,\ldots,q_L)$. Then
\begin{align}
\left|
\sum_{I\sim\bq}e^{2\lambda h(I)}
\tr(\rho A_I^\dagger C_I)
\right|
\leq
\alpha_\bq^{1/2}
\sqrt{F_\bq}
\left(
\sum_{\ell:q_\ell>0}
\frac{\xi_\ell}{\omega_\ell}F_{\bq-\bel}
\right)^{1/2},
\qquad
\alpha_\bq\coloneqq\sum_\ell q_\ell\omega_\ell>0.
\label{eq:upper-source}\\
\end{align}
\end{lemma}

\begin{proof}
The proof consists of two applications of the Cauchy--Schwarz inequality. For
convenience, write
$\|X\|_\rho^2\coloneqq\tr(\rho X^\dagger X)$.
The first application, to the sum over $I$, gives
\begin{align}
\left|
\sum_{I\sim\bq}e^{2\lambda h(I)}
\tr(\rho A_I^\dagger C_I)
\right|^2
&\leq
\left(
\sum_{I\sim\bq}e^{2\lambda h(I)}
\|A_I\|_\rho^2
\right)
\left(
\sum_{I\sim\bq}e^{2\lambda h(I)}
\|C_I\|_\rho^2
\right)\label{eq:source-first-cauchy}
\end{align}
where the first factor on the right-hand side of Eq.~(\ref{eq:source-first-cauchy}) equals $F_\bq$.
 To bound the second factor, we derive a  bound for $\|C_I\|_\rho^2$. To simplify the notation, for a mode $i$ in band $\ell$, define
\begin{align}
	X_{\ell i}\coloneqq
\sum_s\langle u_{\ell s},e_i\rangle G_{\ell s}, \quad\text{ then}\quad C_I
=
\sum_{i\in I}\eta_I(i)
X_{\ell(i)i}A_{I\setminus\{i\}}. \label{eq:Xelli}
\end{align}

Write $|\eta_I(i)|=\sqrt{\omega_{\ell(i)}}\sqrt{\frac{|\eta_I(i)|}{\omega_{\ell(i)}}}$. A second application of
Cauchy--Schwarz gives
\begin{align}
\|C_I\|_\rho^2
\leq
\alpha_\bq
\sum_{i\in I}
\frac{1}{\omega_{\ell(i)}}
\left\|
X_{\ell(i)i}A_{I\setminus\{i\}}
\right\|_\rho^2 
\label{eq:source-second-cauchy}
\end{align}
where we use $\alpha_\bq=\sum_{i\in I}\omega_{\ell(i)}$. 
Note that since $\{u_{\ell s}\}$ are orthonormal vectors in $W_{\ell 1}$, $X_{\ell i}\neq 0$ only if $\langle u_{\ell s},e_i\rangle\neq 0$ and thus $d(i)=1$.  Since $h_1=0$, removing such a mode does not change the weight, 
$$h(I\setminus\{i\})=h(I).$$

Summing Eq.~\eqref{eq:source-second-cauchy} over $I\sim\bq$, writing
$J=I\setminus\{i\}$, and relaxing the condition $i\notin J$, we obtain

\begin{align}
\frac{1}{\alpha_\bq}
\sum_{I\sim\bq}e^{2\lambda h(I)}
\|C_I\|_\rho^2
\leq
\sum_{\ell:q_\ell>0}
\frac{1}{\omega_\ell}
\sum_{J\sim\bq-\bel}
e^{2\lambda h(J)}
\sum_{i:\ell(i)=\ell}
\left\|X_{\ell i}A_J\right\|_\rho^2.
\label{eq:source-regroup}
\end{align}

Further notice that, since $\{u_{\ell s}\}_s$ are orthonormal vectors in $W_{\ell 1}$ and $\{e_i\}_{\ell(i)=\ell}$ is the basis for $\bigoplus_{d\geq 1}W_{\ell d}$, we have 
\begin{align}
	\sum_{i:\ell(i)=\ell}X_{\ell i}^\dagger X_{\ell i}
=
\sum_sG_{\ell s}^\dagger G_{\ell s}
\leq
\xi_\ell \I, \quad \Longrightarrow\quad  e^{2\lambda h(J)}
\sum_{i:\ell(i)=\ell}
\left\|X_{\ell i}A_J\right\|_\rho^2\leq
\xi_\ell\tr(\rho_J). \label{eq:souce_Xli_bound}
\end{align}

Then, combining Eqs.~(\ref{eq:source-first-cauchy}), (\ref{eq:source-regroup}), and (\ref{eq:souce_Xli_bound}), we prove the lemma by checking
\begin{align*}
\left|
\sum_{I\sim\bq}e^{2\lambda h(I)} 
\tr(\rho A_I^\dagger C_I)
\right|^2
&\leq
F_\bq \alpha_\bq \left(\sum_{\ell:q_\ell>0} \frac{\xi_\ell}{\omega_\ell} \sum_{J\sim \bq-\bel}  \tr(\rho_J)\right)=
F_\bq \alpha_\bq \left(\sum_{\ell:q_\ell>0} \frac{\xi_\ell}{\omega_\ell} F_{\bq-\bel}\right). \qedhere
\end{align*}
\end{proof}

Based on Corollary \ref{cor:commutator-identity}, Lemma \ref{lem:lower_bath} and Lemma \ref{lem:upper_source}, we have

\begin{corollary}[Recursion formula with commutator contribution]
\label{cor:common-recursion}
Assume $\omega>0$ and fix a nonzero
$\bq=(q_0,\ldots,q_L)$. 
Then
\begin{align}
&\frac{\alpha_\bq}{2}F_\bq
+
\sum_{I\sim\bq}e^{2\lambda h(I)}\Delta_I
\leq
\sqrt{\alpha_\bq F_\bq}
\left(
\sum_{\ell:q_\ell>0}
\frac{\xi_\ell}{\omega_\ell}F_{\bq-\bel}
\right)^{1/2}.
\end{align}
where
\begin{align}
\alpha_\bq=\sum_{\ell=0}^Lq_\ell\omega_\ell, \qquad \Delta_I= \Re\tr\!\left(
 \rho A_I^\dagger[H,A_I]
 \right)
\label{eq:common-recursion}
\end{align}
\end{corollary}

The term $\Delta_I$ will be estimated separately for the ground state and the relevant part of the thermal state in Section \ref{sec:ground_state_sim} and Section \ref{sec:thermal_state_sim}, which will lead to a final bound for $F_\bq$.

\section{Efficient simulation for ground states}\label{sec:ground_state_sim}

In this section, we apply the compression method developed in Section~\ref{sec:compress} to the ground state. In particular, in Section \ref{sec:compress_ground}, we show that a ground state of the impurity model can be compressed to a subspace of dimension $\poly(n, \delta^{-1}, \es, \omega^{-1})$, as stated in Corollary~\ref{cor:ground_compress}. In Section \ref{sec:ground_energy_classical} we then use this compression to obtain an efficient classical algorithm for estimating the ground-state energy, as stated in Theorem~\ref{thm:classcial_ground_energy}.

\subsection{Compression of the  ground state}\label{sec:compress_ground}

Recall that we use the superscript $\infty$ to specialize the notation in Section \ref{sec:compress} for the ground state. 
In particular,  we use $\rho^{(\infty)}$ for the ground state and specify   
\begin{align}
  h_d^{(\infty)}=d-1,\qquad  D_\ell^{(\infty)} = \sum_{d\geq 1} (d-1) \Pi_{\ell,d}, \qquad  \mathsf{D}^{(\infty)} = \sum_{i} (d(i)-1)n_i.
\end{align}

Note that for the ground state, the commutator contribution is always nonnegative. 
Indeed, let $\alpha_{\mathrm{gs}}$ denote the ground energy of $H$. Since
$H-\alpha_{\mathrm{gs}}\I\geq 0$, for any nonempty increasing list $I$,
\begin{align}
    \Delta_I^{(\infty)} \coloneqq \Re\tr\!\left(
 \rho^{(\infty)} A_I^\dagger[H,A_I]\right) = \tr\left( H A_I \rho^{(\infty)} A_I^\dagger
 \right) - \alpha_{\mathrm{gs}} \tr\left( A_I \rho^{(\infty)} A_I^\dagger
 \right) \geq 0.\label{eq:DeltaI_ground_state}
\end{align}

\begin{theorem}[Exponential localization of residual-bath occupation in Krylov depth] 
\label{thm:main_ground_state} Let $\rho^{(\infty)}$ be a ground state of the impurity Hamiltonian in the form given in Eq.~(\ref{eq:hyb_decom}). Then we have 

\begin{equation}
\tr\left(\rho^{(\infty)}\, e^{ \frac{1}{2}\mathsf D^{(\infty)}}\right) \leq
 \exp\left(
 4\kappa
 \right), \qquad \kappa=\cO\l(m\log_2\frac{\es}{\omega}\r) \quad\text{as in Eq.~(\ref{eq:kappa}).}
 \label{eq:moment}
\end{equation}
In particular, residual-bath occupation decays exponentially with Krylov depth:
\begin{align}
\tr\!\left(\rho^{(\infty)}n_i\right)
\leq
32 \sqrt{e} m \,e^{-\frac{d(i)}{2}}.
\end{align}
\end{theorem}
\begin{proof} 
Applying Corollary~\ref{cor:common-recursion} with Eq.~(\ref{eq:DeltaI_ground_state}), we obtain a recursion for $F_\bq^{(\infty)}$, 
\begin{align}
\alpha_\bq F^{(\infty)}_\bq\leq
4 
\left(
\sum_{\ell:q_\ell>0}
\frac{\xi_\ell}{\omega_\ell}F^{(\infty)}_{\bq-\bel} 
\right). \label{eq:ground_state_recur}
\end{align}
	Using $F^{(\infty)}_{\mathbf{0}}=1$, we obtain the upper bound from the recursion
$F^{(\infty)}_\bq
\leq
\prod_{\ell=0}^L
\frac{1}{q_\ell!}
\left(\frac{4\xi_\ell}{\omega_\ell^2}\right)^{q_\ell}.$ Recall that from Eq.~(\ref{eq:xi}) we have $\xi_\ell \omega_\ell^{-2} \leq 8m$.
Setting $\bq=\mathbf{e}_{\ell(i)}$, we get the bound for  $\tr\left(\rho^{(\infty)} n_i \right)$.
 The bound in Eq.~(\ref{eq:moment})  follows from taking $\lambda=1/4$ and using
$\tr\!\left(\rho^{(\infty)}e^{2\lambda\mathsf D^{(\infty)}}\right)
\leq \sum_{\bq}F_\bq^{(\infty)}$ in Eq.~(\ref{eq:exp_Fq}). 
For completeness, we give the detailed calculations in Appendix~\ref{sec:cal_recursion}.	
\end{proof}

\begin{corollary}[Compression of the ground state]\label{cor:ground_compress}
Let $\rho^{(\infty)}$ be any ground state of the impurity Hamiltonian
in the form given in Eq.~(\ref{eq:hyb_decom}). Then for every $0<\delta<1$, there exists an efficiently enumerable subset $\bS$ of Fock configurations such that 
\begin{align}
    \cH_\bS\coloneqq\operatorname{span}\{\,|x\rangle:x\in\bS\,\},\quad
    \dim(\cH_{\bS})
    \leq\poly(n, \delta^{-1}, \es, \omega^{-1}), \quad
    \tr\left((\I-P_\bS)\rho^{(\infty)}\right)\leq \frac{\delta^2}{4}. \label{eq:compress_ground_state}
\end{align}
where $P_\bS$ is the projector onto $\cH_\bS$. Moreover, the ground energy of the projected Hamiltonian $P_\bS H P_\bS$, when restricted to $\cH_\bS$, differs from the ground energy of the original Hamiltonian $H$ by at most $\|H\|\delta$.
\end{corollary}

\begin{proof}
Eq.~(\ref{eq:compress_ground_state}) follows directly from Theorem~\ref{thm:main_ground_state}
and Theorem~\ref{thm:exp_moment_poly_support}. In particular, define $\tilde{\rho}$ as the normalized version of $P_\bS \rho^{(\infty)} P_\bS$. Then
$\|\tilde{\rho}-\rho^{(\infty)}\|_1\leq \delta$. 
The claim on the ground energy of the projected Hamiltonian follows from the variational principle: the ground energy of a principal submatrix is no less than the ground energy of the original Hermitian matrix, while $\tilde{\rho}$ lies in $\cH_\bS$ and satisfies
$\|\tilde{\rho}-\rho^{(\infty)}\|_1\leq \delta$.
\end{proof}

\subsection{Efficient classical algorithm for ground-energy estimation}\label{sec:ground_energy_classical}

For estimating the ground energy to precision $\delta$, we may without loss of generality assume that $\omega\geq \delta/(2n)$. To see this, recall that after the preprocessing step in Section \ref{sec:hyb_decom}, the bath matrix is diagonal. When restricted to the bath single-particle space, it can be written as $E_0=\operatorname{diag}(\lambda_1,\ldots,\lambda_{n-m})$ with $\lambda_j>0$. Let $\eta\coloneqq\delta/(2n)$ and define $\widehat E_0=\operatorname{diag}(\widehat\lambda_1,\ldots,\widehat\lambda_{n-m})$, where $\widehat\lambda_j\coloneqq\max\{\lambda_j,\eta\}$. Let $\widehat H$ be obtained from $H$ by replacing $\dE(E_0)$ with $\dE(\widehat E_0)$ in Section \ref{sec:hyb_decom}, Eq.~(\ref{eq:impurity_model_hyb}) and leaving all other terms unchanged. Then

$$
\|H-\widehat H\|
=\|\dE(E_0-\widehat E_0)\|
\leq \sum_j|\lambda_j-\widehat\lambda_j|
\leq n\eta=\delta/2.
$$

Hence the ground energy of $\widehat{H}$ differs by at most $\delta/2$ from the ground energy of $H$.  Thus it suffices to estimate the ground energy of $\widehat H$ to precision $\delta/2$. By construction, the bath single-particle energies satisfy $\widehat E_0\geq \eta \I$, and therefore we may assume $\omega\geq\delta/(2n)$; in particular, for $\delta=1/{\poly(n)}$, we may assume $\omega\geq1/{\poly(n)}$.

\begin{theorem}
    [Efficient classical ground energy estimation]\label{thm:classcial_ground_energy}
 Let $H=H_0+H_{\mathrm{imp}}$ be a quantum impurity model as defined in Eq.~(\ref{eq:impurity_model}). Then, for any precision parameter $\delta$, there is a classical algorithm that estimates the ground energy of $H$ to additive error $\delta$ with runtime $\poly(n, \delta^{-1})$.
\end{theorem}

\begin{proof}
By the discussion above Theorem \ref{thm:classcial_ground_energy}, for ground-energy estimation, without loss of generality we may take $\omega=\delta/(2n)$ and reduce the target precision from $\delta$ to $\delta/2$. 

Using $\es=\poly(n)$, $\omega=\delta/(2n)$, $\|H\|=\poly(n)$, and $m=\cO(1)$, 
Corollary \ref{cor:ground_compress} (set the precision parameter to be $\frac{\delta}{2\|H\|}$)
then gives an efficiently
enumerable subspace
$\cH_\bS\coloneqq\operatorname{span}\{\,|x\rangle:x\in\bS\,\}$ of dimension
$\poly(n, \delta^{-1})$, such that the projected Hamiltonian $P_\bS H P_\bS$ has a ground energy  $\delta/2$-close to the ground energy of $H$. 
Let $H_\bS$ be the restriction of $H$ to $\cH_\bS$ with entries $(H_\bS)_{x,y}\coloneqq\langle x|H|y\rangle$ for $x,y\in\bS$.   Constructing and diagonalizing $H_\bS$ takes time $\poly(n, \delta^{-1})$; thus we prove the corollary.
\end{proof}

\section{Efficient simulation for thermal states}\label{sec:thermal_state_sim}

In this section, we apply the compression method developed in Section~\ref{sec:compress} to the thermal state. In particular, in Section~\ref{sec:compress_thermal}, we show that the thermally gapped part of the thermal state can be compressed to a subspace of dimension $\poly(n, \beta, \es, \delta^{-1})$, as stated in Corollary~\ref{cor:compression_thermal_extension}. Then, in Sections~\ref{sec:soft_hard} and~\ref{sec:thermal_classical}, we combine this compression with the soft--hard decomposition and the continuous-time quantum Monte Carlo method to obtain an efficient classical algorithm for estimating thermal expectation values and the partition function, as stated in Theorem~\ref{thm:classical-thermal}.

\subsection{Compression of the thermal state}\label{sec:compress_thermal}

\paragraph{Notation.}
Recall that the superscript $(\beta)$ denotes the specialization of the
notation in Section~\ref{sec:compress} to the thermal state. In particular, 
for a cutoff
$\cut\geq1$ (specified later), we define 
\begin{equation}
    h_d^{(\beta)}\coloneqq\min\{d-1,\cut\},
    \quad
    D_\ell^{(\beta)}
    \coloneqq
    \sum_{d\geq1} \min\{d-1,\cut\} \Pi_{\ell,d},\quad
    \mathsf D^{(\beta)}
    \coloneqq
    \sum_i\min\{d(i)-1,\cut\}n_i.
    \label{eq:thermal-depth-operators}
\end{equation}
In the ground-state case, the residual-bath occupation decays exponentially with Krylov depth, and we use the true depth  $h_d^{(\infty)}=d-1$. For a thermal state, however, thermal occupation is controlled by energy and need not decay with Krylov depth. We thus replace the true Krylov depth by a capped depth to obtain a meaningful bound on the exponential moment
$\tr\left[
\rho^{(\beta)}
e^{\frac{1}{2}\mathsf D^{(\beta)}}
\right]$.

We start with a lemma that estimates the commutator contribution for the thermal state. Note that since $\rho_{\beta}$ commutes with $H$, we have
\begin{align}
\Delta_I^{(\beta)}
    &=
    \Re\tr\!\left(
    \rho^{(\beta)}A_I^\dagger[H,A_I]
    \right) = \tr\!\left(
    \rho^{(\beta)}A_I^\dagger[H,A_I]
    \right).
\end{align}

\begin{lemma}[Thermal commutator estimate]
\label{lem:thermal-commutator}
For every nonempty increasing list $I$, we have
\begin{align}
    -\Delta_I^{(\beta)}
    \leq
    \frac{x_I}{\beta}\log\frac{1}{x_I},\qquad x_I
    &\coloneqq
    \tr\!\left(
    \rho^{(\beta)}A_I^\dagger A_I
    \right).
    \label{eq:thermal-commutator-bound}
\end{align}
where the right-hand side is understood to be zero when $x_I=0$.
\end{lemma}

\begin{proof}
Let
    $y_I\coloneqq\tr(\rho^{(\beta)}A_I A_I^\dagger)$. Consider the spectral decomposition  $H=\sum_a \lambda_a \op{a}{a}$. To ease notation, we abbreviate $x_I,y_I,\Delta_I^{(\beta)}$ as $x,y,\Delta$ here.

 For $x>0$, define
\begin{align}
    \mu_{ab}
    \coloneqq
    \frac{e^{-\beta \lambda_a}|\langle b|A_I|a\rangle|^2}
    {\tr(e^{-\beta H})x},
    \qquad
    (\delta \lambda)_{ab}\coloneqq\lambda_b-\lambda_a.
\end{align}
One can check that $\mu_{ab}\geq 0$ and $\sum_{ab}\mu_{ab}=1$, so $\{\mu_{ab}\}_{a,b}$ is a probability distribution. Besides, \begin{align}
    \frac{\Delta}{x}
    =
    \sum_{a,b}\mu_{ab}(\delta \lambda)_{ab},
    \qquad
    \frac{y}{x}
    =
    \sum_{a,b}\mu_{ab}e^{-\beta(\delta \lambda)_{ab}}.\label{eq:yx}
\end{align}
Applying Jensen's inequality to $y/x$ in Eq.~(\ref{eq:yx}) with respect to the distribution $\{\mu_{ab}\}_{a,b}$ , we have 
\begin{align}
    \log\frac{y}{x}
    \geq \sum_{a,b}\mu_{ab}\,\, (-\beta)\,\, (\delta\lambda)_{ab}=
    -\beta\frac{\Delta}{x},
\end{align}
Since $x>0$ and
$0\leq y\leq\|A_I\|^2\leq1$, we get $-\Delta\leq \frac{x}{\beta} \log \frac{y}{x}\leq\frac{x}{\beta} \log \frac{1}{x},$ thus proving the lemma.
\end{proof}

Then we derive the recursion formula.
For every nonzero $\bq$, define the thermal error term
\begin{align}
    \operatorname{Err}_\bq^{(\beta)}
    \coloneqq
    \frac{e^{-1}}{\beta}
    e^{-\beta\alpha_\bq/8}
    \sum_{I\sim\bq}e^{2\lambda h^{(\beta)}(I)},
    \qquad
    \alpha_\bq\coloneqq\sum_{\ell=0}^Lq_\ell\omega_\ell.
    \label{eq:thermal-error-term}
\end{align}

\begin{lemma}[Thermal recursion]
\label{lem:thermal-recursion} Assume $\omega>0$, so $\omega_\ell=2^\ell \omega>0$ for all $\ell$.
For every nonzero $\bq=(q_0,\ldots,q_L)$,
\begin{align}
    \alpha_\bq F_\bq^{(\beta)}
    \leq
    8\sum_{\ell:q_\ell>0}
    \frac{\xi_\ell}{\omega_\ell}
    F_{\bq-\bel}^{(\beta)}
    +4\operatorname{Err}_\bq^{(\beta)}.
    \label{eq:thermal-recursion}
\end{align}
\end{lemma}

\begin{proof}
Applying Corollary~\ref{cor:common-recursion} with
$\rho=\rho^{(\beta)}$ and $h_d=h_d^{(\beta)}$, we have
\begin{align}
\frac{\alpha_\bq}{2}F_\bq^{(\beta)}
&\leq
\sqrt{\alpha_\bq F_\bq^{(\beta)}}
\left(
\sum_{\ell:q_\ell>0}
\frac{\xi_\ell}{\omega_\ell}
F_{\bq-\bel}^{(\beta)}
\right)^{1/2}
-
\sum_{I\sim\bq}
e^{2\lambda h^{(\beta)}(I)}\Delta_I^{(\beta)}.
\label{eq:thermal-common-recursion}
\end{align}
By Lemma~\ref{lem:thermal-commutator},
\begin{align}
-\sum_{I\sim\bq}
e^{2\lambda h^{(\beta)}(I)}\Delta_I^{(\beta)}
\leq
\frac{1}{\beta}
\sum_{I\sim\bq}
e^{2\lambda h^{(\beta)}(I)}
x_I\log\frac{1}{x_I}, \qquad x_I
    =
    \tr\!\left(
    \rho^{(\beta)}A_I^\dagger A_I
    \right).
\label{eq:summed-thermal-commutator}
\end{align}
For $0\leq x\leq1$ and $t\geq0$, we use the elementary inequality
    $x\log\frac{1}{x}
    \leq
    tx+e^{-t-1}.$
Taking $t=\beta\alpha_\bq/8$, summing over $I\sim\bq$, and using
the definitions of $F_\bq^{(\beta)}$ and
$\operatorname{Err}_\bq^{(\beta)}$, we obtain
\begin{align}
-\sum_{I\sim\bq}
e^{2\lambda h^{(\beta)}(I)}\Delta_I^{(\beta)}
\leq
\frac{\alpha_\bq}{8}F_\bq^{(\beta)}
+\operatorname{Err}_\bq^{(\beta)}.
\label{eq:thermal-error-bound}
\end{align}
Substituting Eq.~\eqref{eq:thermal-error-bound} into
Eq.~\eqref{eq:thermal-common-recursion} gives
\begin{align}
\frac{3\alpha_\bq}{8}F_\bq^{(\beta)}
\leq
\sqrt{\alpha_\bq F_\bq^{(\beta)}}
\left(
\sum_{\ell:q_\ell>0}
\frac{\xi_\ell}{\omega_\ell}
F_{\bq-\bel}^{(\beta)}
\right)^{1/2}
+\operatorname{Err}_\bq^{(\beta)}.
\end{align}
Finally, using
    $\sqrt{xy}\leq\frac{x}{8}+2y,$ for  $x,y\geq0$,
with
    $x=\alpha_\bq F_\bq^{(\beta)}$ and    $y=
    \sum_{\ell:q_\ell>0}
    \frac{\xi_\ell}{\omega_\ell}
    F_{\bq-\bel}^{(\beta)},$ we prove the lemma.
\end{proof}

\begin{theorem}[Thermal exponential-moment bound for capped Krylov depth]
\label{thm:thermal-exponential-bound}  Let $\rho^{(\beta)}$ be the thermal state of the impurity Hamiltonian in the form given in Eq.~(\ref{eq:hyb_decom}) at inverse temperature $\beta$.  Then
\begin{align}
\tr\!\left(
\rho^{(\beta)}e^{\frac{1}{2}\mathsf D^{(\beta)}}
\right)
\leq
e^{8\kappa}
\left(
1+\frac{4\eta_\cut e^{\eta_\cut }}{e\beta\omega}
\right),
\label{eq:thermal-exponential-bound}
\end{align}
where $\kappa=\cO(m\log_2\frac{\es}{\omega})$ as defined in Eq.~(\ref{eq:kappa}), and
\begin{align}
\eta_\cut \coloneqq\sum_i
\exp\!\left(
-\frac{\beta\omega_{\ell(i)}}{8}
+\frac{1}{2}\min\{d(i)-1,\cut\}
\right)
\leq
n\exp\!\left(-\frac{\beta\omega}{8}+\frac{\cut}{2} \right).
\label{eq:thermal-eta}
\end{align}
\end{theorem}

 Note that there is a factor of $n$ in the bound for $\eta_\cut $. To ensure $\eta_\cut\leq 1 $, it suffices that $\log n- \beta\omega/8+ \cut/2<0$. The capped depth is used to get the estimate $\cut/2$ instead of $\max_{i}( d(i)-1)/2$ which could scale as $n$.
 
 The proof of Theorem \ref{thm:thermal-exponential-bound} is similar to that for the ground-state case: We first solve the recursion in Lemma \ref{lem:thermal-recursion} to obtain a bound on $F_\bq^{(\beta)}$, and then use Eq.~(\ref{eq:exp_Fq}) to derive the desired exponential bound. For completeness, we provide the detailed calculation in Appendix \ref{sec:thermal_recursion_solution}.

We now prove a compression result for the thermal state in the thermally gapped regime, i.e., $\beta\omega\geq 9\log n$. Since we will later apply this result to only part of a general impurity Hamiltonian, we also include a version allowing additional fermionic modes.

\begin{corollary}[Compression of the thermal state in the thermally gapped region]
\label{cor:compression_thermal_extension}
Let $H$ be an impurity Hamiltonian in the form given in Eq.~(\ref{eq:hyb_decom}) satisfying $\beta\omega\geq 9\log n$, and let $\rho^{(\beta)}$ be its thermal state. Then, for every $0<\delta<1$, there exists an efficiently enumerable subset  $\bS$ of Fock configurations such that  
\begin{align}
    \cH_\bS\coloneqq \spn\{\,|x\rangle:x\in\bS\,\},\quad
    \dim(\cH_{\bS})
    =
    \poly(n, \beta, \es, \delta^{-1}),
    \quad  \tr\!\left((\I-P_{\bS})\rho^{(\beta)}\right)
    \leq\delta,\label{eq:compress_thermal_1}
\end{align}
where $P_{\bS}$ is the projector onto $\cH_{\bS}$.
Moreover, the thermal state as well as the partition function of the projected Hamiltonian are close to the original Hamiltonian 
\begin{align}
\left\|
\rho^{(\beta)}
-
\frac{
P_{\bS}e^{-\beta P_{\bS}HP_{\bS}}P_{\bS}
}{
\tr\!\left(
P_{\bS}e^{-\beta P_{\bS}HP_{\bS}}P_{\bS}
\right)
}
\right\|_1
\leq\delta, \qquad \left|\tr\!\left(e^{-\beta H}\right)-\tr\!\left(
P_{\bS}e^{-\beta P_{\bS}HP_{\bS}}P_{\bS}
\right)   \right|\leq \delta \tr\!\left(e^{-\beta H}\right).
\label{eq:thermal_projected_compression}
\end{align}
The same conclusions hold when $\rho^{(\beta)}$ is replaced by the thermal state of $H+M$, where $M$ is any  Hermitian operator with even parity, possibly involving additional fermionic modes, such that $[M,a_i]=0$ for every residual-bath mode $a_i$. In this case, $P_{\bS}$ is understood to act as the identity on the additional modes.
\end{corollary}

\begin{proof}
Set
    $\cut=\max\left\{1,\frac{1}{8}\log n\right\},
    q\coloneqq\left(
    \frac{\delta}{4(1+\beta n\es)}
    \right)^4\leq \delta.$
Under the assumption $\beta\omega\geq 9\log n$, the parameter $\eta_\cut$ defined in Eq.~(\ref{eq:thermal-eta}) satisfies $\eta_\cut\leq1$. Applying Theorem~\ref{thm:thermal-exponential-bound} and Theorem~\ref{thm:exp_moment_poly_support} with error parameter $q$ gives an efficiently enumerable set $\bS$ such that Eq.~(\ref{eq:compress_thermal_1}) holds.

Moreover, since $h_1=0$, the value of $D(x)$ depends only on residual-bath occupations at depths
$d\geq 2$. Since $K$ and $G_{\ell j}$ act on the enlarged impurity and
$u_{\ell j}\in W_{\ell,1}$, we have that $P_{\bS}$ commutes with $K$ and commutes with each
coupling term
$a^\dagger(u_{\ell j})G_{\ell j}+G_{\ell j}^\dagger a(u_{\ell j})$. Thus only the quadratic residual-bath term contributes to \(P_{\bS}H(\I-P_{\bS})\), giving
$\|P_{\bS}H(\I-P_{\bS})\|\leq n\es$.
Appendix Lemma~\ref{lem:thermal_approx_proj} therefore gives
\begin{align}
\left\|
\rho^{(\beta)}
-
\frac{
P_{\bS}e^{-\beta P_{\bS}HP_{\bS}}P_{\bS}
}{
\tr\!\left(
P_{\bS}e^{-\beta P_{\bS}HP_{\bS}}P_{\bS}
\right)
}
\right\|_1
&\leq
2q+4\sqrt{\beta n\es}\,q^{1/4}
\leq\delta.
\end{align}

Moreover, with  $\theta_{P_{\bS}} \coloneqq \left\|P_{\bS}H(\I-P_{\bS})\right\| \leq n\es$,  we have $|Z-Z_{\bS}|\leq 4(1+\beta\theta_{P_{\bS}})q^{1/4}Z\leq \delta Z.$

For the extension, since $[M,a_i]=0$ for every residual-bath mode $a_i$, replacing $H$ by $H+M$ does not change the commutator recursion in the proof of Theorem~\ref{thm:thermal-exponential-bound}. Thus, the same exponential-moment bound holds with the same value of $\kappa$. Moreover, $P_{\bS}$ depends only on the residual-bath occupation numbers, so $[P_{\bS},M]=0$ and
    $P_{\bS}(H+M)(\I-P_{\bS})
    =
    P_{\bS}H(\I-P_{\bS}).$
The extended claim now follows from the same argument.
\end{proof}

\subsection{The soft--hard decomposition}\label{sec:soft_hard}

In this section, we consider the general impurity model
$H=H_0+H_{\mathrm{imp}}$ as defined in Eq.~(\ref{eq:impurity_model}), and rewrite it in a form that separates out the thermally gapped part. We set the cutoff value $\omega_*>0$ by
\begin{align}
    \beta \omega_*=9\log n.
\end{align}

We first split $H_0$ according to the cutoff value $\omega_*$. As explained in Appendix \ref{app:cano_bath}, one can construct annihilation operators $c_1,\ldots,c_n$ such that $H_0$ takes the diagonal form
\begin{align}
H_0=\sum_{j=1}^n\epsilon_j c_j^\dagger c_j,
    \qquad 0\leq\epsilon_j.
\end{align}
Here each $c_j$ is a linear combination of the original Majorana operators $\{\gamma_k\}_k$. We then define the \textit{soft part} (low-energy scale) and \textit{hard part} (high-energy scale) of $H_0$ as
\begin{align}
H_{0<}\coloneqq\sum_{j:\,0\leq\epsilon_j<\omega_*}\epsilon_j c_j^\dagger c_j,\qquad
H_{0>}\coloneqq\sum_{j:\,\omega_*\leq\epsilon_j}\epsilon_j c_j^\dagger c_j.
\end{align}
Correspondingly, we partition the single-particle space $\C^n$ according to the cutoff value $\omega_*$ and define the soft and hard single-particle spaces as
\begin{align}
    \cK_<\coloneqq\operatorname{span}\{e_j:0\leq\epsilon_j<\omega_*\},\qquad
    \cK_>\coloneqq\operatorname{span}\{e_j:\omega_*\leq\epsilon_j\}.
\end{align}

We will group $H_{0>}$ and $H_{\mathrm{imp}}$. For the remaining part $H_{0>}$,
by definition, it is decoupled from the soft part $H_{0<}$, while $H_{\mathrm{imp}}$ may still couple to the modes in $H_{0<}$. To facilitate the simulation, we further separate the modes in $H_{0<}$ that may couple to the impurity.

More specifically, recall that we write $c(v)=\sum_j v_j c_j$ for $v\in \C^n$. Define $\cB\subseteq\C^n$
to be the subspace of vectors $v$ for which the expansion of $c(v)$ in
terms of $\gamma_1,\ldots,\gamma_{2n}$ does not contain the impurity modes
$\gamma_1,\ldots,\gamma_{2m}$. We further decompose the soft single-particle space $\cK_<$ as
\begin{align}
    \mathcal S_1\coloneqq\cK_<\cap\cB,\qquad
    \mathcal S_2\coloneqq\cK_<\ominus\mathcal S_1.
\end{align}
Let $s_1\coloneqq\dim(\mathcal S_1)$, $s_2\coloneqq\dim(\mathcal S_2)$, and $d=s_1+s_2$. Note that since $\dim(\cB)\geq n-2m$, we have
$\dim(\mathcal S_1)\geq \dim(\cK_<)-2m$, and hence
$\dim(\mathcal S_2)\leq 2m$.

Choose an orthonormal basis $v_1,\ldots,v_d$ of $\cK_<$ such that the first $s_1$ vectors span
$\mathcal S_1$ and the remaining $s_2$ vectors span $\mathcal S_2$. In this basis, there is a Hermitian matrix $A$ such that\footnote{More specifically, define the operator on $\cK_<$ by
$E_<\coloneqq\sum_{j:\,0\leq\epsilon_j<\omega_*}\epsilon_j e_j e_j^\dagger$.
Let $V:\C^d\to\C^n$ be the isometry whose columns are $\{v_j\}_{j=1}^d$. Then
$A\coloneqq V^\dagger E_<V$. Since $0\leq E_<\leq\omega_*\I$, we also have
$0\leq A\leq\omega_*\I$.}

\begin{align}
    H_{0<}
    =\sum_{j,k=1}^d A_{jk}c(v_k)^\dagger c(v_j),\qquad
    A=
    \begin{bmatrix}
        A^{(11)}&A^{(12)}\\
        A^{(21)}&A^{(22)}
    \end{bmatrix},  \qquad  0\leq A \leq \omega_* \I,
\end{align}
where the block decomposition corresponds to
$\cK_<=\mathcal S_1\oplus\mathcal S_2$. Define the \textit{dark part} of $A$, corresponding to the modes decoupled from the impurity, and the residual part as
\begin{align}
    A^{(\mathrm{dark})}
    \coloneqq
    \begin{bmatrix}
        A^{(11)}&0\\
        0&0
    \end{bmatrix},\qquad A^{(\mathrm{act})}
    \coloneqq
    \begin{bmatrix}
        0&0\\
        0&A^{(22)}
    \end{bmatrix},\qquad 
    A^{(\mathrm{res})}
    \coloneqq
    \begin{bmatrix}
        0&A^{(12)}\\
        A^{(21)}& 0
    \end{bmatrix}.
\end{align}
We then decompose $H$ as
 $H=H_{\mathrm{act}}+H_{\mathrm{dark}}+V_{\mathrm{res}}$,
where
\begin{align}
    &\text{active Hamiltonian: } H_{\mathrm{act}}\coloneqq H_{0>}+H_{\mathrm{imp}} + \sum_{j,k=1}^d A^{(\mathrm{act})}_{jk}
    c(v_k)^\dagger c(v_j),\\
    &\text{free dark Hamiltonian: } H_{\mathrm{dark}}\coloneqq\sum_{j,k=1}^dA^{(\mathrm{dark})}_{jk}
    c(v_k)^\dagger c(v_j),\\
    &\text{residual term: } V_{\mathrm{res}}\coloneqq\sum_{j,k=1}^dA^{(\mathrm{res})}_{jk}
    c(v_k)^\dagger c(v_j).
\end{align}
To express $H$ in a form convenient for later use, we choose a new
single-particle basis. By the definition of $\cB$, the impurity Majorana
operators $\gamma_1,\ldots,\gamma_{2m}$ can be expressed in terms of
$c(u)$ and $c(u)^\dagger$ with $u\in\cB^\perp$. Let
$q\coloneqq\dim(\cB^\perp)\leq2m$. Since $\mathcal S_1\subseteq\cB$, we have
$$\cB^\perp\subseteq\mathcal S_1^\perp
=\cK_>\oplus\mathcal S_2.$$ We therefore choose an orthonormal basis
$f_1,\ldots,f_{n-s_1}$ of $\cK_>\oplus\mathcal S_2$ such that
$f_1,\ldots,f_q$ span $\cB^\perp$. Together,
$f_1,\ldots,f_{n-s_1},v_1,\ldots,v_{s_1}$ form an orthonormal basis
of $\C^n$.

We call  $c(f_1),\ldots,c(f_{n-s_1})$ the \textit{active modes} and $c(v_1),\ldots,c(v_{s_1})$ the \textit{dark modes}.
Then 
 $H_{\mathrm{act}}$ acts only on the active modes,
while $H_{\mathrm{dark}}$ acts only on the dark modes. Moreover,
$H_{\mathrm{imp}}$ acts only on $c(f_1),\ldots,c(f_q)$ and their
adjoints. Thus, $H_{\mathrm{act}}$ is the sum of a quadratic Hamiltonian and an
even impurity term supported on at most $2m$ fermionic modes.

\begin{lemma}[Soft--hard decomposition]\label{lem:soft-hard}
Define
    $H_{\mathrm{ref}}\coloneqq H_{\mathrm{act}}+H_{\mathrm{dark}}$.
Then
\begin{align}
    H=H_{\mathrm{ref}}+V_{\mathrm{res}}, \qquad
    V_{\mathrm{res}}
    =
    \sum_{a=1}^R v_a U_a,
    \,\,\,
    \|U_a\|\leq1,  \,\,\,
     \sum_{a=1}^R|v_a|
    \leq6m\omega_*.
\label{eq:soft_hard}
\end{align}
Moreover, since $H_{\mathrm{act}}$ and $H_{\mathrm{dark}}$ have even parity, with $H_{\mathrm{act}}$ acting trivially on $\mathcal S_1$ and $H_{\mathrm{dark}}$ acting only on $\mathcal S_1$, we have 
\begin{align}
    \rho_{\mathrm{ref}}^{(\beta)}
    =\rho_{\mathrm{act}}^{(\beta)}\otimes\rho_{\mathrm{dark}}^{(\beta)}.\label{eq:tensor_product}
\end{align}
\end{lemma}
\begin{proof}
The decomposition $H=H_{\mathrm{ref}}+V_{\mathrm{res}}$ follows directly from the definitions. 
For $V_{\mathrm{res}}$, take the singular-value decomposition $A^{(12)}=\sum_{a=1}^{R_0}\sigma_a x_a y_a^\dagger$, where $R_0\leq\dim(\mathcal S_2)\leq2m$, $x_a\in\mathcal S_1$, $y_a\in\mathcal S_2$, and $\sigma_a\leq\|A^{(12)}\|\leq\omega_*$. Then $$V_{\mathrm{res}}=\sum_{a=1}^{R_0}\sigma_a \left(c(y_a)^\dagger c(x_a)+ c(x_a)^\dagger c(y_a)\right).$$ Viewing the two terms in the parentheses separately, each corresponding $U_a$ has norm at most one with coefficient $v_a\coloneqq\sigma_a$, so $\sum_a|v_a|\leq 2R_0 \omega_* \leq 4m\omega_*\leq6m\omega_*$.

Finally, since $H_{\mathrm{act}}$ and $H_{\mathrm{dark}}$ have even parity and act on disjoint sets of modes,
the thermal state of $H_{\mathrm{ref}}=H_{\mathrm{act}}+H_{\mathrm{dark}}$ factorizes as
    $\rho_{\mathrm{ref}}^{(\beta)}
    =\rho_{\mathrm{act}}^{(\beta)}\otimes\rho_{\mathrm{dark}}^{(\beta)}.$
\end{proof}

\paragraph{Partially compress the Hamiltonian.}  To apply Corollary~\ref{cor:compression_thermal_extension}, we include
$\mathcal S_2$ in the impurity single-particle space. Since
$\dim(\mathcal S_2)\leq 2m$, the impurity size remains constant.
Thus the corresponding $\omega$ of $H_{\mathrm{act}}$ satisfies 
\begin{align}
    \beta\omega\geq \beta \omega_*\geq  9\log n, \qquad \omega\leq \es.
\end{align}
Here we use the fact that the remaining bath single-particle space is $\cK_>\cap\cB$, so its single-particle energies lie in $[\omega_*,\es]$: the lower bound follows from $\cK_>$, while the upper bound follows from its support on the original bath Majoranas.

We then
preprocess $H_{\mathrm{act}}$ using the construction in 
Section~\ref{sec:band_krylov_rep}, which will apply the bandwise Krylov construction and construct the residual bath modes and enlarged impurity modes from the active modes $c(f_1),\ldots,c(f_{n-s_1})$. 
Let
$\{|x\rangle_{\mathrm{act}}:x\in\{0,1\}^{n-s_1}\}$ denote the Fock basis
defined by these modes, with the ordering fixed in
Section~\ref{sec:compress}.

Let $\rho^{(\beta)}$ denote the thermal state of $H$ at inverse
temperature $\beta$. By
Corollary~\ref{cor:compression_thermal_extension}, for every
$0<\delta<1$, there exists an efficiently enumerable set
$\bS\subseteq\{0,1\}^{n-s_1}$ and a corresponding subspace
\begin{align}
    \cH_{\bS}
    \coloneqq\operatorname{span}\{\,|x\rangle_{\mathrm{act}}:x\in\bS\,\},\qquad \dim(\cH_{\bS})
    \leq\poly(n, \beta, \es, \delta^{-1}),
\end{align}
such that
\begin{align}
\left\|
\rho^{(\beta)}
-
\frac{
P_{\bS}e^{-\beta P_{\bS}HP_{\bS}}P_{\bS}
}{
\tr\!\left(
P_{\bS}e^{-\beta P_{\bS}HP_{\bS}}P_{\bS}
\right)
}
\right\|_1
\leq\delta.
\end{align}
Here $P_{\bS}$ denotes the projector onto $\cH_{\bS}$. Thus, it suffices to work with
$P_{\bS}HP_{\bS}$ instead of $H$.

Regard $\widetilde H\coloneqq P_{\bS}HP_{\bS}$ as a Hamiltonian on
$\cH_{\bS}$ tensored with the Hilbert space defined by the dark modes.
Since $P_{\bS}$ acts trivially on the modes occurring in $H_{\mathrm{dark}}$ and
$V_{\mathrm{res}}$, these two terms remain unchanged under this restriction.
Therefore,
\begin{align}
    \widetilde H
    =
    P_{\bS}H_{\mathrm{act}}P_{\bS}
    +H_{\mathrm{dark}}+V_{\mathrm{res}},
    \qquad
    \beta\|V_{\mathrm{res}}\|\leq54m\log n.\label{eq:partial_compress}
\end{align}
Moreover, the matrix representation of $P_{\bS}H_{\mathrm{act}}P_{\bS}$ can be
constructed in time $\poly(n, |\bS|)$.

\subsection{Efficient classical simulation via continuous-time quantum Monte Carlo}\label{sec:thermal_classical}

To complete the classical simulation algorithm, we use continuous-time quantum Monte Carlo (CT-QMC), a standard numerical simulation technique~\cite{gull2011continuous,werner2006continuous,rubtsov2005continuous}. While its rigorous runtime generally has exponential dependence on $\beta$, our compression and soft--hard decomposition allow us to apply CT-QMC with an effective perturbation satisfying $\beta W=\cO(\log n)$, leading to a runtime polynomial in $\beta$.

\paragraph{Continuous-time quantum Monte Carlo.}
  Below, we recall only the basic idea and the complexity bound of CT-QMC; the detailed algorithm and its analysis are given in Appendix~\ref{sec:CT-QMC}.

Consider a decomposition
\begin{align}
    H=H_{\mathrm{ref}}+V,
    \qquad
    V=\sum_{a=1}^R v_a U_a,
    \qquad
    \|U_a\|\leq 1,\qquad  W\coloneqq\sum_{a=1}^R |v_a|.
\end{align}
where the coefficients $v_a\in\mathbb{R}$.  
The imaginary-time interaction-picture expansion expands the operator $e^{-\beta H}$ around $H_{\mathrm{ref}}$, which gives
\begin{align}
e^{-\beta H}
={}&
\sum_{r=0}^{\infty}(-1)^r
\sum_{a_1,\ldots,a_r}
\left(\prod_{j=1}^r v_{a_j}\right)
\int_{\Omega_r(\beta)}
e^{-(\beta-\tau_1)H_{\mathrm{ref}}}
U_{a_1}
e^{-(\tau_1-\tau_2)H_{\mathrm{ref}}}
\cdots
U_{a_r}
e^{-\tau_rH_{\mathrm{ref}}}
\,d\tau_1\cdots d\tau_r,
\label{eq:ctqmc-expansion}
\end{align}
where
$\Omega_r(\beta)
\coloneqq
\left\{
(\tau_1,\ldots,\tau_r):
\beta\geq\tau_1\geq\cdots\geq\tau_r\geq0
\right\}.$
Thus, each term in the expansion is specified by an expansion order $r$, indices
$a_1,\ldots,a_r$, and imaginary times
$\tau_1,\ldots,\tau_r$. For an observable $O$, define the corresponding configuration value
\begin{align}
\mathcal C_O(r,\boldsymbol a,\boldsymbol\tau)
\coloneqq
\tr\!\left(
O
e^{-(\beta-\tau_1)H_{\mathrm{ref}}}
U_{a_1}
e^{-(\tau_1-\tau_2)H_{\mathrm{ref}}}
\cdots
U_{a_r}
e^{-\tau_rH_{\mathrm{ref}}}
\right),
\label{eq:ctqmc-configuration}
\end{align}
where
$\boldsymbol a=(a_1,\ldots,a_r)$ and
$\boldsymbol\tau=(\tau_1,\ldots,\tau_r)$.
CT-QMC truncates the expansion at a finite order and estimates the numerator and denominator of the thermal expectation value by sampling these configurations. The following theorem summarizes the complexity bound that we need.

\begin{theorem}[Continuous-time quantum Monte Carlo]
\label{thm:ctqmc}
Let $\epsilon,\delta$ be two precision parameters.
Let $H$ be the Hamiltonian in Eq.~(\ref{eq:H_CT_QMC}).
 Let $O$ be a Hermitian observable with $\|O\|\leq1$.
Suppose that the reference partition function
$\tr(e^{-\beta H_{\mathrm{ref}}})$
and the configuration values
$\mathcal C_O(r,\boldsymbol a,\boldsymbol\tau)$ and $\mathcal C_\I(r,\boldsymbol a,\boldsymbol\tau)$
can be evaluated in time $T_{\mathrm{cfg}}(K)$ for every $r\leq K$, where
$
K=\cO\!\left(\beta W+\log\epsilon^{-1}\right).
$
Then there is a randomized classical algorithm that outputs  estimates $\widehat\mu$ and $\widehat Z$ such that
\begin{align*}
\Pr\!\left[
\left|
\widehat\mu
-
\tr\!\left(
O\frac{e^{-\beta H}}{\tr(e^{-\beta H})}
\right)
\right|
\leq\epsilon
\right]
\geq1-\delta,\qquad \Pr\!\left[ \left| \widehat Z - \tr\!\left( e^{-\beta H} \right) \right| \leq \epsilon\tr\!\left( e^{-\beta H} \right) \right] \geq1-\delta,
\end{align*}
with runtime
\[e^{\cO(\beta W)}
\operatorname{poly}\!\left(
K,\epsilon^{-1},\log\delta^{-1}
\right)
T_{\mathrm{cfg}}(K).\]
\end{theorem}

Theorem~\ref{thm:ctqmc} follows by truncating the interaction-picture expansion in Eq.~\eqref{eq:ctqmc-expansion} at order
$K=\cO(\beta W+\log\epsilon^{-1})$
and applying Monte Carlo sampling to estimate the resulting numerator and partition function. We give the detailed algorithm and proof in Appendix~\ref{sec:CT-QMC}.

We are now ready to prove that thermal states can be simulated efficiently on a classical computer.

\begin{definition}[Gaussian observable]
	We say that an operator $O$ on fermionic modes is \emph{Gaussian} if it can be written as
$O=\alpha\exp\left(\frac14\sum_{j,k}A_{jk}\gamma_j\gamma_k\right)$
for some scalar $\alpha$ and antisymmetric matrix $A$. \end{definition}

For simplicity, we state the following theorem for Gaussian observables. The same algorithm also applies when $O$ can be written as a polynomial-size sum $O=\sum_{j=1}^{\poly(n)} O^{(j)}_{\mathrm{act}}\otimes O^{(j)}_{\mathrm{dark}},$ where $O^{(j)}_{\mathrm{act}}$ can be any operator that acts on the active modes and $O^{(j)}_{\mathrm{dark}}$ is Gaussian on the dark modes.

\begin{theorem}[Efficient classical simulation of thermal states]
\label{thm:classical-thermal}
Let $H=H_0+H_{\mathrm{imp}}$ be a quantum impurity model as defined in Eq.~(\ref{eq:impurity_model}). Let $O$ be a Hermitian Gaussian observable with $\|O\|\leq1$. Then, for any $\beta>0$ and $0<\epsilon<1$, there is a randomized classical algorithm that outputs estimates $\widehat\mu$ and $\widehat Z$ such that
\begin{align}
\Pr\!\left[
\left|
\widehat\mu
-
\tr\!\left(
O\rho^{(\beta)}
\right)
\right|
\leq\epsilon
\right]
\geq\frac{2}{3},\qquad \Pr\!\left[ \left| \widehat Z - \tr\!\left( e^{-\beta H} \right) \right| \leq \epsilon\tr\!\left( e^{-\beta H} \right) \right] \geq \frac{2}{3},
\end{align}
with runtime
    $\poly(n, \beta, \es, \epsilon^{-1}).$
\end{theorem}

\begin{proof}
Apply the partial compression in Section~\ref{sec:soft_hard} with error parameter $\epsilon/2$ and denote the corresponding partially compressed Hamiltonian by
\begin{align}
    \widetilde H
    =
    P_{\bS}H_{\mathrm{act}}P_{\bS}
    +H_{\mathrm{dark}}+V_{\mathrm{res}},
\end{align}
where
$\dim(\cH_{\bS})=\poly(n, \beta, \es, \epsilon^{-1})$.
It suffices to estimate the thermal expectation value of $\widetilde H$ to precision $\epsilon/2$.
We apply Theorem~\ref{thm:ctqmc} with
\begin{align}
    \widetilde H_{\mathrm{ref}}
    \coloneqq
    P_{\bS}H_{\mathrm{act}}P_{\bS}+H_{\mathrm{dark}},
    \qquad
    V=V_{\mathrm{res}}.
\end{align}
By Lemma~\ref{lem:soft-hard},
    $W\coloneqq\sum_{a=1}^R|v_a|
    \leq6m\omega_*,$
and hence
$\beta W=\cO(\log n)$.

Moreover, by the tensor-product structure in Eq.~(\ref{eq:tensor_product}),
the reference partition function factorizes into the partition function of
$P_{\bS}H_{\mathrm{act}}P_{\bS}$ and that of $H_{\mathrm{dark}}$. Since
$P_{\bS}H_{\mathrm{act}}P_{\bS}$ acts on a
$\poly(n, \beta, \es, \epsilon^{-1})$-dimensional space and $H_{\mathrm{dark}}$ is
quadratic, the reference partition function can be evaluated in polynomial
time.

For the CT-QMC configuration values
$\mathcal C_O(r,\boldsymbol a,\boldsymbol\tau)$, we expand the observable
with respect to the active Fock basis as
\begin{align}
    P_{\bS} O P_{\bS}
    =
    \sum_{x,y\in\bS}
    \op{y}{x}_{\mathrm{act}}
    \otimes O_{xy}^{\mathrm{dark}},
    \qquad
    O_{xy}^{\mathrm{dark}}
    \coloneqq
    (\langle x|_{\mathrm{act}}\otimes I)
    O
    (|y\rangle_{\mathrm{act}}\otimes I).
\end{align}
Since $|\bS|=\poly(n,\beta,\es,\epsilon^{-1})$, this contains only
polynomially many terms. Moreover, by the SVD construction in the proof of
Lemma~\ref{lem:soft-hard}, each $U_a$ can be written, up to a fermionic
parity factor that can be absorbed into the active operator, as a product
of an operator acting only on the active modes and an operator acting only
on the dark modes. Therefore, using the tensor-product structure of
$\widetilde H_{\mathrm{ref}}$ as in Lemma~\ref{lem:soft-hard},
Eq.~(\ref{eq:tensor_product}), each configuration value reduces to a
polynomial number of products of an active-sector trace and a dark-sector
correlation function.

The active-sector trace can be evaluated directly on the polynomial-dimensional
space $\cH_{\bS}$. For the dark-sector factor, although
$O_{xy}^{\mathrm{dark}}$ need not itself be Gaussian, we do not have to construct it explicitly. Indeed, for any dark-sector operator
$X_{\mathrm{dark}}$ arising in a configuration, write
\begin{align}\label{eq:tracing_out_dark_Oxy}
    \tr_{\mathrm{dark}}
    \left(
        O_{xy}^{\mathrm{dark}} X_{\mathrm{dark}}
    \right)
    =
    \tr\left[
        \left(
            \op{y}{x}_{\mathrm{act}}
            \otimes X_{\mathrm{dark}}
        \right) O
    \right].
\end{align}
The outer product
$\op{y}{x}_{\mathrm{act}}$
can be written as a product of $\cO(n)$ fermionic creation and annihilation
operators. Since $X_{\mathrm{dark}}$ consists of Gaussian imaginary-time
evolutions together with at most $r\leq K$ fermionic insertions, and since $O$
is Gaussian, the right-hand side of Eq.~\eqref{eq:tracing_out_dark_Oxy} is a correlation function of
$\cO(n+K)$ fermionic operators within a product of fermionic Gaussian operators.
By the generalized Wick theorem, such correlation functions can be computed by
Pfaffians of matrices of polynomial size~\cite{lieb1968theorem,balian1969nonunitary}, while the remaining
Gaussian trace is evaluated by the standard Pfaffian formula
for fermionic Gaussian operators~\cite{klich2014note,bravyi2005lagrangian}.
Thus each dark-sector factor is computable in polynomial time, and hence
\begin{align}
    T_{\mathrm{cfg}}(K)
    =
    \poly(K,n,\beta,\es,\epsilon^{-1}).
\end{align}

The same argument applies to $\mathcal C_\I(r,\boldsymbol a,\boldsymbol\tau)$ by taking $O=\I$, for which $P_{\bS}\I P_{\bS}$ is simply the identity operator on the compressed Hilbert space.
Thus, by Corollary~\ref{cor:compression_thermal_extension} and Theorem~\ref{thm:ctqmc} with accuracy $\epsilon/2$ and failure probability $1/3$, we obtain the claimed estimates for both the thermal expectation value and the partition function.
\end{proof}

\subsection{Efficient thermal-state preparation from quantum belief propagation}
\label{sec:thermal_quantum}

In \cref{sec:thermal_classical}, we combined the soft--hard decomposition and partial compression with CT-QMC to obtain an efficient classical algorithm for estimating thermal expectation values and the partition function. In this section, we show that the same structural results, when combined with quantum belief propagation (QBP)~\cite{hastings2007quantum}, lead to an efficient quantum algorithm for preparing the thermal state itself.

Suppose we wish to prepare the Gibbs state of a Hamiltonian $H_{\mathrm{ref}} + V$. Assuming an efficient circuit to prepare (the purification of) the Gibbs state of $H_{\mathrm{ref}}$, our algorithm achieves this with a complexity that scales exponentially only in $\beta\|V\|$ and polynomially in all other parameters. We are not aware of any prior generically applicable thermal-state preparation algorithm featuring this $e^{\O(\beta\|V\|)}$ scaling, although some are close (for example,~\cite{holmes2022quantum} achieves a scaling of $e^{\O(\beta\|V\|/\delta)}$ for trace-distance error $\delta$). Although QBP has existed for almost two decades~\cite{hastings2007quantum}, to our knowledge it had not previously been applied to the general $H_{\mathrm{ref}} + V$ scenario. Instead, most applications assumed some physical conditions such as bounded correlation lengths~\cite{kato2019quantum,brandao2019finite}.

QBP describes how a thermal state changes under a perturbation to the Hamiltonian. Consider
\begin{align}
H_s\coloneqq H_{\mathrm{ref}}+sV,\qquad s\in[0,1].
\end{align}
Define
\begin{align}
\Phi_s
&\coloneqq
\int_{-\infty}^\infty
f_\beta(t)e^{-itH_s}Ve^{itH_s}\,dt,
\qquad
\eta_s
\coloneqq
\TO\exp\left(
-\frac{\beta}{2}\int_0^s\Phi_r\,dr
\right),
\end{align}
where $f_\beta(t)$ is some probability distribution (see \cref{eq:Phi_s_defn} of \cref{sec:QBP} for the definition). The QBP operator $\eta_s$ satisfies the key identity:
\begin{align}
\eta_s e^{-\beta H_{\mathrm{ref}}} \eta_s^\dagger
=
e^{-\beta H_s}.
\end{align}
Thus, starting from the thermal state of $H_{\mathrm{ref}}$, the QBP operator transforms it to the thermal state of $H_{\mathrm{ref}} + V$. We give a quantum algorithm that block-encodes the QBP operator based on the linear combination of Hamiltonian simulation technique~\cite{an2023linear,low2025optimal}. We provide a detailed analysis of its construction and complexity in \cref{sec:QBP}, and we summarize the resulting guarantee below.

\begin{theorem}[Quantum thermal-state preparation via QBP]
\label{thm:QBP_prep_alg_main}
Let $\delta > 0$ and $\beta > 0$. Suppose that we have query access to block encodings of Hamiltonians $H_{\mathrm{ref}}$ and $V$ with normalizations $\lambda_{\mathrm{ref}}, \lambda_V$ respectively. Also assume that we have access to a circuit $U_{\mathrm{ref}}$ preparing the purification of the
thermal state of $H_{\mathrm{ref}}$ at inverse temperature $\beta$.
Then, there is a quantum algorithm that prepares a purification of the
thermal state of $H_{\mathrm{ref}}+V$ at inverse temperature $\beta$, with trace-distance error at most $\delta$, and with query complexity (to the block encodings and $U_{\mathrm{ref}}$) of
\[
e^{\O(\beta \|V\|)} \poly(\lambda_{\mathrm{ref}},  \lambda_V, \beta, \delta^{-1}).
\]
\end{theorem}

The gate complexity is also efficient; see \cref{sec:QBP} for details. The key feature of our structural result for impurity models is that $\beta \|V\|=\cO(\log n)$, while $H_{\mathrm{ref}}$ can be handled efficiently due to the compression.

\begin{corollary}[Efficient quantum thermal-state preparation]
\label{cor:quantum-thermal}
Let $H=H_0+H_{\mathrm{imp}}$ be a quantum impurity model as
defined in \cref{eq:impurity_model}. Then, for any $\beta>0$ and
$0<\delta<1$, there is a quantum algorithm that prepares a state
$\sigma$ satisfying
$\|\sigma-\rho^{(\beta)}\|_1\leq\delta$ with runtime
$\poly(n,\beta,\es,\delta^{-1})$.
\end{corollary}

\begin{proof}
By the partial compression in Section~\ref{sec:soft_hard}, it suffices to prepare the thermal state of
\begin{align}
\widetilde H
=
\widetilde H_{\mathrm{ref}}+V_{\mathrm{res}},
\qquad
\widetilde H_{\mathrm{ref}}
\coloneqq
P_{\bS}H_{\mathrm{act}}P_{\bS}+H_{\mathrm{dark}},
\end{align}
where the reference thermal state can be prepared efficiently and
$\beta\|V_{\mathrm{res}}\|\leq54m\log n=\cO(\log n)$.
The result therefore follows directly from
Theorem~\ref{thm:QBP_prep_alg_main} together with the compression error.
\end{proof}

\section{Hardness of simulating dynamical properties}

Here we show that two-point correlation functions of time-dependent impurity Hamiltonians are hard for classical computers to compute, but are easy for quantum computers. Specifically, we show that the problem is $\DQC_1$-complete if the initial state is at infinite temperature, and is $\BQP$-complete for any inverse temperature $\beta \in [\Omega(1), \poly(n)]$. This encompasses the Green's function in nonequilibrium DMFT, which is a central object of that method.

\subsection{Complexity classes}

We work with the definition of $\DQC_1$ from Brandão's thesis~\cite{brandao2008entanglement}, which automatically allows any efficient classical sideprocessor ($\BPP$) alongside the restricted one-clean-qubit machine. This version more naturally captures the power of the quantum computation allowed within this model.

\begin{definition}[$\DQC_1$]
    Let $L = (L_{\mathrm{yes}},L_{\mathrm{no}})$ be a promise problem. We say that $L \in \DQC_1$ if there is a polynomial $q(n)$, functions $a ,b : \N \to [0,1]$ satisfying $a(n) - b(n) \geq 1/{\poly(n)}$, and a family of poly-size quantum circuits $C_x$ acting on $q(n)$ qubits, generated in polynomial time, such that, defining
    \begin{equation}
        p_x \coloneqq
        \tr\l[
            \l(\op{1}{1} \otimes \I\r)
            C_x
            \l(
                \op{0}{0}
                \otimes \frac{\I}{2^{q(n)-1}}
            \r)
            C_x^\dagger
        \r],
    \end{equation}
    we have $p_x \geq a(n)$ for $x \in L_{\mathrm{yes}}$ and $p_x \leq b(n)$ for $x \in L_{\mathrm{no}}$. The one-clean-qubit computation may be repeated polynomially many times, with the outcomes processed by a probabilistic polynomial-time classical computer.
\end{definition}

We will also use the standard circuit definition of (promise) $\BQP$.

\begin{definition}[$\BQP$]
    Let $L = (L_{\mathrm{yes}},L_{\mathrm{no}})$ be a promise problem. We say that $L \in \BQP$ if there is a polynomial $q(n)$ and a family of poly-size quantum circuits $C_x$ acting on $q(n)$ qubits, generated in polynomial time, such that, defining
    \begin{equation}
        p_x \coloneqq
        \tr\l[
            \l(\op{1}{1} \otimes \I\r)
            C_x
            \op{0^{q(n)}}{0^{q(n)}}
            C_x^\dagger
        \r],
    \end{equation}
    we have $p_x \geq 2/3$ for $x \in L_{\mathrm{yes}}$ and $p_x \leq 1/3$ for $x \in L_{\mathrm{no}}$.
\end{definition}

From these definitions it is clear that $\BPP \subseteq \DQC_1 \subseteq \BQP$. It is also conjectured that both inclusions are strict, since $\DQC_1$ can solve problems believed to be hard for classical computers~\cite{knill1998power,shor2008estimating}. Thus, $\DQC_1$ serves as an intermediate class likely to be beyond classical simulation, yet not fully capturing the power of quantum computation. Note also that $\DQC_k$, where we generalize to $k$ clean qubits, remains equivalent to $\DQC_1$ for any $k = \O(\log n)$~\cite{shepherd2006computation}.

The canonical $\DQC_1$-complete problem is unitary trace estimation~\cite{shepherd2006computation,shor2008estimating}.

\begin{problem}[Unitary trace estimation]\label{prob:trace_estimation}
    Let $C$ be a poly-size quantum circuit on $n$ qubits and $\epsilon > 0$ an error parameter. The goal is to output a number $\mu \in \C$ such that $|\mu - \frac{1}{2^n} \tr(C)| \leq \epsilon$, with probability at least $2/3$.
\end{problem}

\begin{proposition}
    Trace estimation is $\DQC_1$-complete for $\epsilon = 1/{\poly(n)}$.
\end{proposition}

The canonical $\BQP$-complete problem is quantum circuit acceptance~\cite{bernstein1993quantum,kitaev2002classical}.

\begin{problem}[Quantum circuit acceptance]\label{prob:circuit_accept}
    Let $C$ be a poly-size quantum circuit on $n$ qubits and define
    \begin{equation}
        p_C \coloneqq
        \tr\l[
            \l(\op{1}{1} \otimes \I\r) \cdot
            C
            \op{0^n}{0^n}
            C^\dagger
        \r].
    \end{equation}
    We are promised that either $p_C \geq 2/3$ or $p_C \leq 1/3$, and the goal is to decide which is the case.
\end{problem}

\begin{proposition}
    Quantum circuit acceptance is $\BQP$-complete.
\end{proposition}

\subsection{Statement of results}

Now we introduce the problem for impurity models. Throughout, we tacitly assume that the time-dependent impurity Hamiltonian $H : t \mapsto H(t)$ is given in an efficiently computable representation for any $t$, accurate up to inverse-polynomial precision. We also assume that all coefficients in $H(t)$ are bounded by $1$ for all time (i.e., we do not need polynomially large interaction strengths).

\begin{problem}[Time-dependent correlation function estimation]\label{prob:corr_func}
    We are given as inputs:~a time-dependent impurity Hamiltonian $H(t)$, two Majorana operators $A, B$, an inverse temperature $\beta \geq 0$, two real times $t, t' \geq 0$, and an error parameter $\epsilon \geq 0$. Let $\sigma_\beta$ be the Gibbs state of $H(0)$ at inverse temperature $\beta$ and $U(t) = \TO \exp\l(-i \int_0^t H(s) \, ds\r)$ be the time-evolution operator. The goal is to output a number $\mu \in \C$ such that
    \begin{equation}
    \begin{split}
        | \mu - \tr(\sigma_\beta &\cdot A(t) B(t')) | \leq \epsilon,\\
        A(t) \equiv U(t)^\dagger A \, U(t), &\quad B(t') \equiv U(t')^\dagger B \, U(t')
    \end{split}
    \end{equation}
    with probability at least $2/3$.
\end{problem}

For example, solving this problem four times with $A, B \in \{\gamma_{2j-1}, \gamma_{2j}, \gamma_{2k-1}, \gamma_{2k}\}$ allows us to compute Green's functions such as
\begin{equation}
    G_{jk}(t, t') = -i\langle a_j(t) \, a_k^\dagger(t') \rangle_\beta,
\end{equation}
which are central to the study of nonequilibrium physics~\cite{aoki2014nonequilibrium}. We show that this problem is hard for classical computers, even at infinite temperature, under the assumption that $\BPP \neq \DQC_1$.

\begin{theorem}\label{thm:DQC1-complete}
    Problem \ref{prob:corr_func} is $\DQC_1$-complete for $\epsilon = 1/{\poly(n)}$, $t, t' = \poly(n)$, and $\beta = 0$.
\end{theorem}

\begin{theorem}\label{thm:BQP-complete}
    Problem \ref{prob:corr_func} is $\BQP$-complete for $\epsilon = 1/{\poly(n)}$, $t, t' = \poly(n)$, and $\beta \in [\Omega(1), \poly(n)]$.
\end{theorem}

Both \cref{thm:DQC1-complete,thm:BQP-complete} hold even for the minimal impurity size $m = 4$ and when $A$ and $B$ are single-site operators.

\subsection{Universality of time-dependent impurity Hamiltonians}

In~\cite{bravyi2017complexity} it was shown that time-dependent impurity Hamiltonians are universal for quantum computation in the traditional $\BQP$ circuit model. The precise reduction there maps the impurity Hamiltonian to a 1D chain of XY interactions, plus a triangle at the end which encodes the impurity. Universality then follows from the construction of Brod and Childs~\cite{brod2014computational}, which encodes $\Omega(\sqrt{n})$ qubits into a logical code space of the physical $n$ qubits. However, this is an issue for $\DQC_1$ because the maximally mixed state on the physical space is $\I/2^n$, whereas we want to encode a calculation with the maximally mixed state only in the code space. With the quadratic space overhead, the target signal would be exponentially suppressed within the physically measured quantity $\tr(U)/2^n$. Hence this encoding would demand exponentially small error to solve Problem~\ref{prob:trace_estimation} in the $\DQC_1$ model.

To handle this, we prove a new universality result for time-dependent impurity models that uses only a single ancilla qubit. With this, the unitary trace is suppressed only by a constant factor. Our construction encodes an arbitrary $n$-qubit quantum circuit $V$ into an $(n+1)$-mode impurity unitary $U = e^{i\phi} (V \oplus V)$, where all parts of this construction (including computing the phase $\phi$) are efficient and incur at most polynomial overhead in gate complexity. Thus $\frac{1}{2^n} \tr(V) = \frac{1}{2^{n+1}} \tr(U)$ is recovered exactly.

\begin{theorem}[Single-ancilla universality encoding]\label{thm:td_impurity_encoding}
    Let $V = g_L \cdots g_2 g_1$ be a quantum circuit on $n$ qubits, where each $g_\ell$ is a gate acting on at most two adjacent qubits in a line. There exists a 1D time-dependent impurity Hamiltonian $H(t)$ on $n + 1$ modes, with impurity size $m = 4$, such that
    \begin{equation}
        U(t) \equiv \TO \exp\l[ -i \int_0^t ds \, H(s) \r] = e^{i\phi} (V \oplus V),
    \end{equation}
    where the orthogonal decomposition above is with respect to even- and odd-parity sectors of $(\C^2)^{\otimes(n+1)}$ and $t = \O(n^2 L)$. The coefficients of $H(t)$ are piecewise constant and can be classically computed from the description of $V$ in $\poly(n, L)$ time, as well as the phase $\phi \in \R$.
\end{theorem}

\begin{proof}
    By taking the coefficients to be piecewise constant, $U(t)$ can be decomposed into gates generated by its local terms. We define the $(n+1)$-mode fermionic system on a 1D line, where we place the ancilla mode $\anc$ behind mode $1$. To distinguish the ancilla, we label its Majorana modes as $\gamma_{\anc,1}$ and $\gamma_{\anc,2}$. We designate the first two non-ancillary modes to hold the impurity (i.e., Majorana modes $\gamma_1, \gamma_2, \gamma_3, \gamma_4$). We only need the following fermionic gate set (i.e., terms in $H(t)$):
    \begin{itemize}
        \item On modes $j \geq 1$, $\displaystyle \fSWAP_{j,j+1} \coloneqq \exp\l[ -\frac{\pi}{4} \l( \gamma_{2j} \gamma_{2j+1} - \gamma_{2j} \gamma_{2j+2} + \gamma_{2j-1} \gamma_{2j} + \gamma_{2j+1} \gamma_{2j+2} \r) \r]$;
        \item On modes $\anc$ and $1$, $R_{X,1}(\theta) \coloneqq e^{-(\theta/2) \gamma_{\anc,2} \gamma_1}$, $R_{Y,1}(\theta) \coloneqq e^{-(\theta/2) \gamma_{\anc,2} \gamma_2}$, and $R_{Z,1}(\theta) \coloneqq e^{-(\theta/2) \gamma_1 \gamma_2}$;
        \item On impurity modes $1$ and $2$, $\displaystyle \CZ_{1,2} \coloneqq \exp\l[ \frac{i\pi}{4} \l( \gamma_1 \gamma_2 \gamma_3 \gamma_4 - i\gamma_1 \gamma_2 - i\gamma_3 \gamma_4 \r) \r]$.
    \end{itemize}
    This has the impurity structure as claimed. We will map this to a universal $n$-qubit gate set through two steps:~first we apply the Jordan--Wigner transformation on the physical level. Then, we identify a logical code space corresponding to the even-parity subspace. Note that all gates above commute with the parity operator so any unitary built from them must block diagonalize as such. Afterwards, we will also show that the odd-parity subspace contains the same encoding.

    Taking into account the ancilla mode, the Jordan--Wigner mapping is
    \begin{align}
        \gamma_{\anc,1} = X_\anc, &\quad \gamma_{\anc,2} = Y_\anc\\
        \gamma_{2j-1} = \l( Z_\anc \prod_{k=1}^{j-1} Z_k \r) X_j, &\quad \gamma_{2j} = \l( Z_\anc \prod_{k=1}^{j-1} Z_k \r) Y_j.
    \end{align}
    Under this mapping we see the reason for the gate names above:
    \begin{equation}\label{eq:fSWAP+CZ}
        \fSWAP = -i \begin{pmatrix}
            1 & 0 & 0 & 0\\
            0 & 0 & 1 & 0\\
            0 & 1 & 0 & 0\\
            0 & 0 & 0 & -1
        \end{pmatrix}, \quad \CZ = e^{-i\pi/4} \begin{pmatrix}
            1 & 0 & 0 & 0\\
            0 & 1 & 0 & 0\\
            0 & 0 & 1 & 0\\
            0 & 0 & 0 & -1
        \end{pmatrix}.
    \end{equation}
    The rotation gates are understood after mapping to the code space. To be explicit, this is the image of the following encoding map:~for $b \in \{0, 1\}$, define
    \begin{equation}\label{eq:parity_encoding}
    \begin{split}
        \mathcal{E}_b& : (\C^2)^{\otimes n} \to (\C^2)^{\otimes (n + 1)}\\
        \mathcal{E}_b& \ket{x} = \ket{b\oplus\parity(x)} \ket{x}
    \end{split}
    \end{equation}
    where $\parity(x) \in \{0,1\}$ is the parity of $x \in \{0,1\}^n$. The bit $b$ simply labels which parity sector we are in. Next, let us define logical single-qubit Pauli operators on an encoded $n$-qubit space:
    \begin{equation}
        \overline{X}_j \coloneqq X_\anc X_j, \quad \overline{Y}_j \coloneqq X_\anc Y_j, \quad \overline{Z}_j \coloneqq Z_j.
    \end{equation}
    It is easily checked that these obey the necessary Pauli algebra. Thus the rotation gates $R_{X,1}(\theta) = e^{-i \frac{\theta}{2} \overline{X}_1}$, etc., precisely simulate single-qubit rotations $\overline{R_{X,1}}(\theta)$ on the first logical qubit. Furthermore, the logical versions $\overline{\fSWAP}, \overline{\CZ}$ coincide with their physical gates $\fSWAP, \CZ$ (up to global phase), since their generators are of the form $\overline{X}_j \overline{X}_{j+1} = X_j X_{j+1}$, $\overline{Y}_j \overline{Y}_{j+1} = Y_j Y_{j+1}$, $\overline{Z}_j = Z_j$, and $\overline{Z}_1 \overline{Z}_2 = Z_1 Z_2$.

    It remains to show how to encode the circuit $V = g_L \cdots g_2 g_1$ into this impurity gate set. We first show that the logical gates above allow us to address arbitrary qubits along the line. We will assume that all global phases appearing from \cref{eq:fSWAP+CZ} are compensated by adding the appropriate identity terms to the time-dependent Hamiltonian. Clearly, this does not affect the impurity structure and can be efficiently determined with $\O(1)$ effort per gate.
    
    The key identity we use is
    \begin{equation}
        \overline{\SWAP} = \overline{\fSWAP} \cdot \overline{\CZ}.
    \end{equation}
    Although $\overline{\CZ}_{j,j+1}$ is only directly available for $j = 1$, it can be generated on any neighboring pair. This is because we have access to the \emph{fermionic} swap between any pair, which can bring any pair to positions $(1, 2)$. Let
    \begin{equation}
        P_j \coloneqq \prod_{(p,q)} \overline{\fSWAP}_{p,q}
    \end{equation}
    where $\{(p, q)\}$ is any sequence of nearest-neighbor transpositions that maps $(j, j+1) \leftrightarrow (1, 2)$. For any $j \in [n-1]$ this takes at most $\O(n)$ swaps. Then in position $(1, 2)$, we apply $\overline{\CZ}_{1,2}$. Finally we uncompute the fermionic swaps using $P_j^\dagger$, bringing the qubits back to their original positions. Despite using fermionic swaps, the accumulated signs cancel:~write
    \begin{equation}
        \overline{\CZ}_{p,q} = e^{i \pi \, n_p n_q}, \quad \text{where } n_p \coloneqq \frac{\I - \overline{Z}_p}{2}
    \end{equation}
    and observe that
    \begin{equation}
        \overline{\fSWAP}_{p,p+1}^\dagger n_p \overline{\fSWAP}_{p,p+1} = n_{p+1}, \quad \overline{\fSWAP}_{p,p+1}^\dagger n_{p+1} \overline{\fSWAP}_{p,p+1} = n_p.
    \end{equation}
    Thus, each conjugation of $\overline{\CZ}_{p',q'}$ by $\overline{\fSWAP}_{p,q}$ transports the CZ gate without any sign, and so
    \begin{equation}
        P_j^\dagger \overline{\CZ}_{1,2} P_j = \overline{\CZ}_{j,j+1}.
    \end{equation}
    Thus we have $\overline{\CZ}$ on any adjacent pair, so we can implement $\overline{\SWAP}$ on any adjacent pair too, with $\O(n)$ gate overhead.

    We can therefore implement ordinary swap gates to permute any logical qubit to the first (logical) position, again with $\O(n)$ overhead in the swaps. Note however that each SWAP is generated itself by $\O(n)$ fSWAPs, so a coarse upper bound for the physical gate overhead is $\O(n^2)$. Since $\overline{R_{X,1}}(\theta), \overline{R_{Y,1}}(\theta), \overline{R_{Z,1}}(\theta)$ generate arbitrary single-qubit gates there, we also have arbitrary single-qubit gates on every logical qubit. We also have arbitrary two-qubit gates on every pair of logical qubits, because a constant number of $\overline{\CZ}$ gates and single-qubit rotations generate $\U(4)$.

    Finally, we verify that the same circuit is implemented in both parity sectors. Recall for $b \in \{0, 1\}$, that the image of $\mathcal{E}_b$ is precisely the $(-1)^b$-parity sector. Moreover, $\mathcal{E}_1 = X_\anc \mathcal{E}_0$. But every gate in our physical gate set commutes with $X_\anc$, since the only gates with nontrivial generators on the ancilla are $X_\anc X_1$ and $X_\anc Y_1$. Therefore the physical unitary $U \equiv U(t)$ obeys $[U, X_\anc] = 0$. This implies that, if $U \mathcal{E}_0 = \mathcal{E}_0 V$, then
    \begin{equation}
        U \mathcal{E}_1 = U X_\anc \mathcal{E}_0 = X_\anc U \mathcal{E}_0 = X_\anc \mathcal{E}_0 V = \mathcal{E}_1 V.
    \end{equation}
    But $U \mathcal{E}_0 = \mathcal{E}_0 V$ is precisely the representation that gives \cref{eq:fSWAP+CZ}, so $U$ is equal to $V \oplus V$, up to global phase.
\end{proof}

Before proceeding to the proofs of \cref{thm:DQC1-complete,thm:BQP-complete}, we record a helpful lemma about how logical expectation values are encoded.

\begin{lemma}\label{lem:majorana_readout}
    Let $V$ be a circuit on $q$ logical qubits and let $U$ be the corresponding physical unitary on $q+1$ modes constructed above, so that
    \begin{equation}
        U \mathcal{E}_b = \mathcal{E}_b V, \quad b \in \{0,1\}.
    \end{equation}
    Then $U^\dagger \gamma_{2q} U \gamma_{2q} \mathcal{E}_b = \mathcal{E}_b V^\dagger X_q V X_q$, and in particular for any $q$-qubit density matrix $\rho$, we have
    \begin{equation}\label{eq:majorana_readout_trace}
        \tr\l(
            \frac{1}{2} \sum_{b=0}^1
            \mathcal{E}_b \rho \mathcal{E}_b^\dagger
            U^\dagger \gamma_{2q} U \gamma_{2q}
        \r)
        =
        \tr(
            \rho V^\dagger X_q V X_q
        ).
    \end{equation}
\end{lemma}

\begin{proof}
    From the Jordan--Wigner transformation, we have $\gamma_{2q} = Z_\anc Z_1 \cdots Z_{q-1} Y_q$. Applying this to the encoding of a computational basis state $\ket{x}$ gives $\gamma_{2q} \mathcal{E}_b \ket{x} = i (-1)^b \mathcal{E}_{1-b} X_q \ket{x}$. Then using $U\mathcal{E}_b = \mathcal{E}_bV$ twice gives
    \begin{equation}
    \begin{split}
        U^\dagger \gamma_{2q} U \gamma_{2q} \mathcal{E}_b
        &=
        i(-1)^b U^\dagger \gamma_{2q}
        \mathcal{E}_{1-b} V X_q \\
        &=
        U^\dagger \mathcal{E}_b X_q V X_q \\
        &=
        \mathcal{E}_b V^\dagger X_q V X_q
    \end{split}
    \end{equation}
    as claimed. Taking the trace in the two orthogonal parity sectors proves \cref{eq:majorana_readout_trace}.
\end{proof}

\subsection{$\DQC_1$-completeness at infinite temperature}

\begin{proof}[Proof of \cref{thm:DQC1-complete}]
    We first show containment. At infinite temperature,
    \begin{equation}
        \sigma_0 = \frac{\I}{2^n},
    \end{equation}
    where $n$ is the number of fermionic modes. It is a standard fact that the real-time evolution $U(t)$ can be implemented to inverse-polynomial accuracy by a poly-size quantum circuit without additional clean ancillas~\cite{poulin2011quantum}, since it is simply a poly-size sum of Pauli operators after Jordan--Wigner with bounded coefficients. Since the Majorana operators $A$ and $B$ are mapped to Pauli operators,
    \begin{equation}
        W = U(t)^\dagger A \, U(t) \cdot U(t')^\dagger B \, U(t')
    \end{equation}
    is itself unitary, and so
    \begin{equation}
        \tr(
            \sigma_0 \cdot A(t) B(t')
        )
        =
        \frac{1}{2^n}\tr(W).
    \end{equation}
    Thus the desired correlation function is precisely a normalized unitary trace and can be estimated in $\DQC_1$. Taking the circuit approximation to $U(t)$ to be $1/{\poly(n)}$ accurate only changes this by $1/{\poly(n)}$ additive error.

    Now we show hardness by reducing trace estimation. It suffices to take $t' = 0$. Let $C$ be any poly-size circuit on $n$ qubits. Introduce one additional logical qubit $c$, which we place at the \textit{end} of the logical line, and for $\theta \in \{0,\pi/2\}$ define
    \begin{equation}
        V_\theta
        \coloneqq
        \I \otimes \op{0}{0}_c
        +
        e^{i\theta}
        C \otimes \op{1}{1}_c.
    \end{equation}
    This is also a poly-size circuit, and by \cref{thm:td_impurity_encoding} it can be encoded into a time-dependent impurity evolution operator on only one more mode. Let $U_\theta$ be that unitary on $n+2$ fermionic modes and denote
    \begin{equation}
        \Gamma \coloneqq \gamma_{2(n+1)}
    \end{equation}
    for the second Majorana operator of the last non-ancillary mode. Following \cref{lem:majorana_readout} and a direct calculation, we get
    \begin{equation}\label{eq:DQC_trace_reduction}
    \begin{split}
        \tr\l(
            \frac{\I}{2^{n+2}}
            U_\theta^\dagger \Gamma U_\theta \Gamma
        \r)
        &=
        \frac{1}{2^{n+1}}
        \tr(
            V_\theta^\dagger X_c V_\theta X_c
        )\\
        &=
        \frac{1}{2^{n+1}}
        \tr(
            e^{i\theta}\op{0}{0}_c \otimes C
            +
            e^{-i\theta}\op{1}{1}_c \otimes C^\dagger
        )\\
        &=
        \Re\l(
            e^{i\theta}\frac{\tr(C)}{2^n}
        \r).
    \end{split}
    \end{equation}
    For $\theta = 0$, \cref{eq:DQC_trace_reduction} gives the real part while $\theta = \pi/2$ gives the (negative) imaginary part. Thus two calls to Problem~\ref{prob:corr_func} solve trace estimation to $\epsilon = 1/{\poly(n)}$ accuracy. Since trace estimation is $\DQC_1$-complete, Problem~\ref{prob:corr_func} with $\beta = 0$ is $\DQC_1$-hard.
\end{proof}

\subsection{$\BQP$-completeness at finite temperature}

Because of our efficient quantum algorithm for preparing thermal states (\cref{cor:quantum-thermal}), $\BQP$ containment is fairly straightforward. However, in order to exhibit hardness from an initial thermal state, we need a result of Schulman and Vazirani~\cite{schulman1999molecular} that concentrates many thermal qubits into a nearly pure state with only mild space overhead. Then we can map the expectation with such a mixed state to the form required of Problem~\ref{prob:circuit_accept}, with only inverse-polynomial error.

For $\eta \in [0,1]$, define the binary entropy function
\begin{equation}
    h_2\l(\frac{1+\eta}{2}\r)
    \coloneqq
    -\frac{1+\eta}{2}\log_2\l(\frac{1+\eta}{2}\r)
    -\frac{1-\eta}{2}\log_2\l(\frac{1-\eta}{2}\r).
\end{equation}
\cite{schulman1999molecular} proved the following.

\begin{proposition}\label{prop:SV99}
    Consider $N$ independent random bits of bias $\eta > 0$, i.e.,
    \begin{equation}
        \Pr(0) = \frac{1+\eta}{2},
        \quad
        \Pr(1) = \frac{1-\eta}{2}.
    \end{equation}
    For $N > \eta^{-2}$, there is a reversible algorithm using no additional bits and running in time $\O(N\log N)$ which, except with probability $\exp(-N^{\Theta(1)})$, extracts
    \begin{equation}
        M
        =
        N\l(
            1-h_2\l(\frac{1+\eta}{2}\r)-o(1)
        \r)
    \end{equation}
    bits, each having bias $1-\exp(-N^{\Theta(1)})$.
\end{proposition}

In other words, there is an efficient reversible circuit which converts $N$ thermal bits into $\O(N)$ nearly pure bits, exponentially close to $\ket{0}$, provided that the bias $\eta$ is sufficiently far from $0$. For our application, we need a joint trace-distance statement rather than a per-qubit bound; this is an immediate consequence of \cref{prop:SV99}.

\begin{lemma}\label{lem:SV_trace_distance}
    Fix $\eta > 0$ and let
    \begin{equation}
        \tau_{\eta}
        \coloneqq
        \frac{\I+\eta Z}{2}.
    \end{equation}
    There is a family of poly-size reversible quantum circuits $S_N$ and a designated output register $R$ of
    \begin{equation}
        M
        =
        N\l(
            1-h_2\l(\frac{1+\eta}{2}\r)-o(1)
        \r)
    \end{equation}
    qubits such that
    \begin{equation}\label{eq:SV_trace_distance}
        \frac{1}{2}
        \left\|
            \tr_{\neg R}
            \l[
                S_N \tau_{\eta}^{\otimes N} S_N^\dagger
            \r]
            -
            \op{0^M}{0^M}
        \right\|_1
        \leq
        \exp(-N^{\Theta(1)}).
    \end{equation}
    In particular, for every number $q$ of desired nearly pure qubits, there is a sufficiently large $N = \O(q)$ such that $M \geq q$, with the implicit constant depending only on $\eta$.
\end{lemma}

\begin{proof}
    The algorithm of \cref{prop:SV99} is reversible and acts classically in the computational basis, so its action on $\tau_{\eta}^{\otimes N}$ produces another classical state. Let $G$ denote the good event in \cref{prop:SV99}. There are constants $c_1,c_2>0$ such that
    \begin{equation}
        \Pr(\neg G) \leq e^{-N^{c_1}},
    \end{equation}
    while conditioned on $G$ every one of the $M$ extracted qubits has bias at least $1-e^{-N^{c_2}}$. Therefore, if $X_1,\ldots,X_M$ denote the measurement outcomes of these qubits,
    \begin{equation}
        \Pr(X_j = 1 \mid G)
        \leq
        \frac{1}{2}e^{-N^{c_2}}.
    \end{equation}
    A union bound gives
    \begin{equation}
    \begin{split}
        \Pr\l[
            (X_1,\ldots,X_M)\neq 0^M
        \r]
        &\leq
        \Pr(\neg G)
        +
        \Pr\l[
            (X_1,\ldots,X_M)\neq 0^M
            \,\middle|\,
            G
        \r] \\
        &\leq
        e^{-N^{c_1}}
        +
        \frac{M}{2}e^{-N^{c_2}} = e^{-N^{\Theta(1)}}.
    \end{split}
    \end{equation}
    Since everything is diagonal, tracing out the complement of $R$ gives a state satisfying
    \begin{equation}
        \frac{1}{2}
        \left\|
            \rho_R-\op{0^M}{0^M}
        \right\|_1
        =
        1-\bra{0^M}\rho_R\ket{0^M},
    \end{equation}
    which proves \cref{eq:SV_trace_distance}. Finally, for constant $\eta > 0$, define
    \begin{equation}
        r_{\eta}
        \coloneqq
        1-h_2\l(\frac{1+\eta}{2}\r)
        > 0.
    \end{equation}
    Hence, for all sufficiently large $N$, $M \geq r_{\eta}N/2$. Choosing $N \geq 2q/r_{\eta}$ therefore gives at least $q$ purified qubits.
\end{proof}

We now prove $\BQP$-completeness of Problem~\ref{prob:corr_func} for $\beta \in [\Omega(1), \poly(n)]$.

\begin{proof}[Proof of \cref{thm:BQP-complete}]
    We first show containment. By \cref{cor:quantum-thermal},
    the Gibbs state $\sigma_\beta$ of any impurity Hamiltonian $H(0)$ can be prepared to $1/{\poly(n)}$ trace distance in time $\poly(n, \beta)$, and just as in the proof of \cref{thm:DQC1-complete}, $U(t)$ can be implemented to $1/{\poly(n)}$ error in polynomial time~\cite{poulin2011quantum}. Thus up to $1/{\poly(n)}$ approximation errors, we can prepare $\sigma_\beta$ and perform a Hadamard test for the unitary
    \begin{equation}
        W = U(t)^\dagger A \, U(t) \cdot U(t')^\dagger B \, U(t').
    \end{equation}
    Taking all implementation errors sufficiently smaller than $\epsilon$ and estimating the real and imaginary parts by the Hadamard test proves that Problem~\ref{prob:corr_func} is in $\BQP$.

    Next we prove hardness. It suffices to fix $t' = 0$ and take $\beta = \Theta(1)$ to be any nonzero constant for this hard instance. Let $C_x$ be the $q$-qubit circuit of an arbitrary instance of Problem~\ref{prob:circuit_accept}, with acceptance probability
    \begin{equation}
        p_x
        =
        \tr\l(
            \frac{\I-Z_1}{2} \cdot
            C_x
            \op{0^q}{0^q}
            C_x^\dagger
        \r).
    \end{equation}
    Introduce an energy scale $\Delta > 0$. We will use $N$ logical thermal qubits, one additional logical ancilla qubit $c$, and one ancilla qubit $\anc$ for the fermionic parity encoding. Set the initial physical Hamiltonian to be a simple diagonal free fermion bath:
    \begin{equation}\label{eq:thermal_hardness_H0}
        H(0)
        =
        \Delta\sum_{j=1}^{N+1} n_j,
        \quad
        \text{where } n_j = \frac{\I-Z_j}{2},
    \end{equation}
    Its Gibbs state is
    \begin{equation}\label{eq:thermal_product_state}
        \sigma_\beta
        =
        \frac{\I_\anc}{2}
        \otimes
        \tau_{\eta_\beta}^{\otimes(N+1)},
        \quad \text{where }
        \eta_\beta
        =
        \tanh\l(\frac{\beta\Delta}{2}\r)
        \quad
        \text{and}
        \quad
        \tau_{\eta_\beta}
        =
        \frac{\I+\eta_\beta Z}{2}.
    \end{equation}
    Since $\beta, \Delta > 0$ are fixed constants, $\eta_\beta > 0$ is also a constant bounded away from $0$. Moreover, from the definition of the parity encoding (\cref{eq:parity_encoding}), it holds that
    \begin{equation}\label{eq:thermal_parity_decomposition}
        \sigma_\beta
        =
        \frac{1}{2}
        \sum_{b=0}^1
        \mathcal{E}_b
        \tau_{\eta_\beta}^{\otimes(N+1)}
        \mathcal{E}_b^\dagger.
    \end{equation}

    By \cref{lem:SV_trace_distance}, we may choose $N = \O(q)$ and a reversible circuit $S_N$ on the first $N$ logical qubits such that, for a designated $q$-qubit register $R$,
    \begin{equation}\label{eq:cold_register}
        \frac{1}{2}
        \left\|
            \tr_{\neg R}
            \l[
                S_N\tau_{\eta_\beta}^{\otimes N}S_N^\dagger
            \r]
            -
            \op{0^q}{0^q}
        \right\|_1
        \leq
        e^{-N^{\Theta(1)}}.
    \end{equation}
    This register $R$ contains the nearly clean qubits on which $C_x$ will act. Without loss of generality we can arrange $R = \{1, \ldots, q\}$. Define the unitary
    \begin{equation}
        M_x
        \coloneqq
        C_x^\dagger Z_1 C_x.
    \end{equation}
    and its controlled form on the ancilla $c = N+1$:
    \begin{equation}
        W_x
        \coloneqq
        \I \otimes \op{0}{0}_c
        +
        M_x \otimes \op{1}{1}_c.
    \end{equation}
    Finally, define
    \begin{equation}
        V_x
        \coloneqq
        W_x(S_N\otimes\I_c)
    \end{equation}
    which includes the reversible nearly purifying circuit $S_N$. Observe that all of these circuits have polynomial size. Since $M_x = M_x^\dagger$, we get
    \begin{equation}
        W_x^\dagger X_c W_x X_c
        =
        \I_c\otimes M_x.
    \end{equation}
    Furthermore, $S_N$ acts trivially on $c$, so we also have
    \begin{equation}\label{eq:BQP_readout_logical}
        V_x^\dagger X_c V_x X_c
        =
        S_N^\dagger M_x S_N \otimes \I_c.
    \end{equation}

    Compile $V_x$ into an impurity evolution operator $U_x$ using the universality construction of \cref{thm:td_impurity_encoding}, and put
    \begin{equation}
        \Gamma
        \coloneqq
        \gamma_{2(N+1)}.
    \end{equation}
    We take the value of $H(s)$ at $s = 0$ to be \cref{eq:thermal_hardness_H0} and use the compiled piecewise-constant Hamiltonian for $s > 0$; this value $H(0)$ at a single time does not affect the time-ordered exponential. By \cref{lem:majorana_readout,eq:thermal_parity_decomposition,eq:BQP_readout_logical}, we get
    \begin{equation}
        \tr\l(
            \sigma_\beta
            U_x^\dagger \Gamma U_x \Gamma
        \r) =
        \tr\l(
            \tau_{\eta_\beta}^{\otimes(N+1)}
            V_x^\dagger X_c V_x X_c
        \r) =
        \tr\l(
            \tau_{\eta_\beta}^{\otimes N}
            S_N^\dagger M_x S_N
        \r).
    \end{equation}
    Following \cref{eq:cold_register}, denote the $R$-register state by
    \begin{equation}
        \rho_R
        =
        \tr_{\neg R}
        \l(
            S_N\tau_{\eta_\beta}^{\otimes N}S_N^\dagger
        \r).
    \end{equation}
    Then
    \begin{equation}
        \tr\l(
            \sigma_\beta
            U_x^\dagger \Gamma U_x \Gamma
        \r)
        =
        \tr\l(
            \rho_R C_x^\dagger Z_1 C_x
        \r),
    \end{equation}
    and hence by \cref{eq:cold_register} we have
    \begin{equation}
    \begin{split}
        \l|
            \tr\l(
                \rho_R C_x^\dagger Z_1 C_x
            \r)
            -
            \ev{C_x^\dagger Z_1 C_x}{0^q}
        \r|
        &\leq
        \|
            \rho_R - \op{0^q}{0^q}
        \|_1 \\
        &=
        e^{-N^{\Theta(1)}}.
    \end{split}
    \end{equation}
    Note that we have absorbed irrelevant constants into the $\Theta(1)$. Finally, using the fact that
    \begin{equation}
        \ev{C_x^\dagger Z_1 C_x}{0^q}
        =
        1-2p_x,
    \end{equation}
    we have for $x \in \{0,1\}^n$ for sufficiently large $n$,
    \begin{align}
        x \in L_{\mathrm{yes}}
        &\implies
        \tr\l(
            \sigma_\beta
            U_x^\dagger \Gamma U_x \Gamma
        \r)
        \leq -\frac{1}{4}, \\
        x \in L_{\mathrm{no}}
        &\implies
        \tr\l(
            \sigma_\beta
            U_x^\dagger \Gamma U_x \Gamma
        \r)
        \geq \frac{1}{4}.
    \end{align}
    Thus an inverse-polynomial additive approximation to the correlation function decides every instance of the quantum circuit acceptance problem. Since $N = \O(q)$ and $q = \poly(n)$, all circuits above are poly-size on $\poly(N)$ fermionic modes. This proves $\BQP$-hardness.
\end{proof}

\section*{Acknowledgments}\addcontentsline{toc}{section}{Acknowledgments}

We thank Andrew Baczewski, Garnet Chan, David Gosset, Zoë Holmes, Alina Kononov, Joonho Lee, Jake Nelson, Shivesh Pathak, Yu Tong, and Huang Zhen for helpful and illuminating discussions.
J.J.\ is supported by the Simons Quantum Postdoctoral Fellowship, by a Simons Investigator Award in Mathematics through Grant No. 825053.
O.P.\ was supported by the U.S.\ Department of Energy, Office of Science, National Quantum Information Science Research Centers, Quantum Systems Accelerator (Award No. DE-SCL0000121).
C.R.\ acknowledges the funding support by UK Research and Innovation (UKRI) under the UK government's Horizon Europe funding guarantee EP/X032051/1, and U.S. Department of Energy, Office of Science, Accelerated Research in Quantum
Computing, Fundamental Algorithmic Research toward Quantum Utility (FAR-Qu).
A.Z.\ was supported by the National Nuclear Security Administration’s Advanced Simulation and Computing program and the U.S.\ Department of Energy Office
of Fusion Energy Sciences ``Foundations for quantum simulation of warm dense
matter'' project.

This article has been authored by an employee of National Technology \& Engineering Solutions of Sandia, LLC under Contract No.\ DE-NA0003525 with the U.S.\ Department of Energy (DOE). The employee owns all right, title and interest in and to the article and is solely responsible for its contents. The United States Government retains and the publisher, by accepting the article for publication, acknowledges that the United States Government retains a non-exclusive, paid-up, irrevocable, world-wide license to publish or reproduce the published form of this article or allow others to do so, for United States Government purposes. The DOE will provide public access to these results of federally sponsored research in accordance with the DOE Public Access Plan \url{https://www.energy.gov/downloads/doe-public-access-plan}.

\paragraph{Concurrent work.}
Near the completion of this manuscript we became aware of concurrent and independent work by Arunachalam \emph{et al.}~\cite{arunachalam2026quantum}, which also studies the computational complexity of quantum impurity models for both statics and dynamics. They also give polynomial-time classical algorithms for the ground- and thermal-state problems that we consider, although their approaches (especially for thermal states) differ conceptually from ours. The scope of our dynamical results is also distinct:~we study the complexity of thermal Green's functions out of equilibrium, whereas they consider the universality of time-independent impurity Hamiltonians.

\paragraph{AI methodology.}
The authors began studying efficient simulation of impurity models by attempting to refine the analysis of~\cite{bravyi2017complexity}, in particular by seeking a better rational approximation to the square-root function. ChatGPT 5.5 and Claude Opus 4.8 helped identify a bug in an early proof and later produced a counterexample showing that the canonical bath basis could not yield an efficient algorithm. The authors then decided to pursue a change-of-basis approach and formulated a sufficient condition via semidefinite programming (SDP) for constructing a guiding state for quantum phase estimation.  GPT 5.6 suggested Krylov-space techniques used in Wilson's NRG method as a candidate basis and, through subsequent interactions, helped establish that a Krylov basis indeed satisfies the human-formulated SDP condition. This led to our early result of efficient quantum algorithms for estimating the ground-state energy~\cite{jiang2026impuritytalk}.

The authors subsequently attempted to dequantize this quantum algorithm. After several initially unsuccessful approaches with GPT 5.6 Sol Ultra, it eventually proposed, based on the Krylov-basis idea, a candidate proof of an efficient classical algorithm for the ground-state energy. The authors refined and verified this proof and observed that the same strategy could compress thermal states under a thermal-gap condition. The authors then abstracted the proof strategy to apply uniformly to both ground and thermal states, resulting in the framework presented in Section~\ref{sec:compress}.

The results on nonequilibrium Green's function estimation, including the low-space-overhead encoding for universal time-dependent impurity Hamiltonians, were conceived of by the authors, although initially only considering the setting of $\DQC_1$. The authors later recognized that $\BQP$ containment was immediate from their quantum algorithms and so investigated whether $\BQP$-hardness was also possible, even at high temperatures. GPT 5.6 identified the result of Schulman and Vazirani~\cite{schulman1999molecular} that became the key ingredient for bridging that gap.

Lighter models of GPT 5.6 and Opus 5 were used in sharpening calculations and refining proofs overall. The authors independently checked all technical details and take full responsibility for the correctness of the results.

\bibliographystyle{alphaurl}
\bibliography{reference}

@article{bravyi2017complexity,
  author  = {Bravyi, Sergey and Gosset, David},
  title   = {Complexity of Quantum Impurity Problems},
  journal = {Communications in Mathematical Physics},
  volume  = {356},
  number  = {2},
  pages   = {451--500},
  year    = {2017},
  doi     = {10.1007/s00220-017-2976-9},
  url     = {https://doi.org/10.1007/s00220-017-2976-9}
}

@article{arunachalam2026quantum,
  title={Quantum impurity models: easy at equilibrium, universal in motion},
  author={Arunachalam, Srinivasan and Bravyi, Sergey and Chowdhury, Anirban and Dutt, Arkopal and Gheorghiu, Alexandru and Li, Zhi},
  journal={arXiv:2610.02130},
  year={2026},
  url={https://arxiv.org/abs/2610.02130}
}

@article{molnar2014approximating,
  title = {Approximating {G}ibbs states of local {H}amiltonians efficiently with projected entangled pair states},
  author = {Molnar, Andras and Schuch, Norbert and Verstraete, Frank and Cirac, J. Ignacio},
  journal = {Physical Review B},
  volume = {91},
  issue = {4},
  pages = {045138},
  numpages = {11},
  year = {2015},
  month = {Jan},
  publisher = {American Physical Society},
  doi = {10.1103/PhysRevB.91.045138},
  url = {https://link.aps.org/doi/10.1103/PhysRevB.91.045138}
}

@misc{jiang2026impuritytalk,
  author       = {Jiaqing Jiang},
  title        = {Ground Energy Estimation of Quantum Impurity Model Is in {BQP}},
  howpublished = {Talk at the Quantum Summer Cluster Final Workshop, Simons Institute for the Theory of Computing},
  year         = {2026},
  month        = {July},
  url          = {https://www.youtube.com/watch?v=krqOy27VIu0},
  note = {Video recording; joint work with Nathan Ju, Ojas Parekh, Chaithanya Rayudu, and Andrew Zhao},
}

@article{aharonov2009power,
  title={The power of quantum systems on a line},
  author={Aharonov, Dorit and Gottesman, Daniel and Irani, Sandy and Kempe, Julia},
  journal={Communications in Mathematical Physics},
  volume={287},
  number={1},
  pages={41--65},
  year={2009},
  publisher={Springer},
  doi={10.1007/s00220-008-0710-3}
}

@article{rubin2016hybrid,
  title={A hybrid classical/quantum approach for large-scale studies of quantum systems with density matrix embedding theory},
  author={Rubin, Nicholas C},
  journal={arXiv:1610.06910},
  year={2016},
  url={https://arxiv.org/abs/1610.06910}
}

@article{bauer2016hybrid,
  title={Hybrid quantum-classical approach to correlated materials},
  author={Bauer, Bela and Wecker, Dave and Millis, Andrew J and Hastings, Matthew B and Troyer, Matthias},
  journal={Physical Review X},
  volume={6},
  number={3},
  pages={031045},
  year={2016},
  publisher={APS},
  doi={10.1103/PhysRevX.6.031045}
}

@inproceedings{gottesman2009quantum,
  title={The quantum and classical complexity of translationally invariant tiling and {Hamiltonian} problems},
  author={Gottesman, Daniel and Irani, Sandy},
  booktitle={2009 50th Annual IEEE Symposium on Foundations of Computer Science},
  pages={95--104},
  year={2009},
  organization={IEEE},
  doi={10.1109/FOCS.2009.22}
}

@article{landau2015polynomial,
  title={A polynomial time algorithm for the ground state of one-dimensional gapped local {Hamiltonians}},
  author={Landau, Zeph and Vazirani, Umesh and Vidick, Thomas},
  journal={Nature Physics},
  volume={11},
  number={7},
  pages={566--569},
  year={2015},
  publisher={Nature Publishing Group UK London},
  doi={10.1038/nphys3345}
}

@article{brod2014computational,
  title={The computational power of matchgates and the {XY} interaction on arbitrary graphs},
  author={Brod, Daniel J and Childs, Andrew M},
  journal={Quantum Information and Computation},
  volume={14},
  number={11\&12},
  pages={0901--0916},
  year={2014},
  doi={10.26421/QIC14.11-12-1},
  url={https://arxiv.org/abs/1308.1463}
}

@article{tong2025fast,
  title={Fast mixing of weakly interacting fermionic systems at any temperature},
  author={Tong, Yu and Zhan, Yongtao},
  journal={PRX Quantum},
  volume={6},
  number={3},
  pages={030301},
  year={2025},
  publisher={APS},
  doi={10.1103/h1dx-ps5p}
}

@article{smid2025polynomial,
  title={Polynomial-time quantum {G}ibbs sampling for the weak and strong coupling regime of the {Fermi-Hubbard} model at any temperature},
  author={{\v{S}}m{\'\i}d, {\v{S}}t{\v{e}}p{\'a}n and Meister, Richard and Berta, Mario and Bondesan, Roberto},
  journal={Nature Communications},
  volume={16},
  number={1},
  pages={10736},
  year={2025},
  publisher={Nature Publishing Group UK London},
  doi={10.1038/s41467-025-65765-1}
}

@article{lieb1968theorem,
  title={A theorem on {P}faffians},
  author={Lieb, Elliott H},
  journal={Journal of Combinatorial Theory},
  volume={5},
  number={3},
  pages={313--319},
  year={1968},
  publisher={Elsevier},
  doi={10.1016/S0021-9800(68)80078-X}
}

@article{balian1969nonunitary,
  title={Nonunitary {B}ogoliubov transformations and extension of {W}ick’s theorem},
  author={Balian, R and Brezin, E},
  journal={Il Nuovo Cimento B (1965-1970)},
  volume={64},
  number={1},
  pages={37--55},
  year={1969},
  publisher={Springer},
  doi={10.1007/BF02710281}
}

@article{klich2014note,
  title={A note on the full counting statistics of paired fermions},
  author={Klich, Israel},
  journal={Journal of Statistical Mechanics: Theory and Experiment},
  volume={2014},
  number={11},
  pages={P11006},
  year={2014},
  doi={10.1088/1742-5468/2014/11/P11006}
}

@article{erakovic2025high,
  title = {High Ground State Overlap via Quantum Embedding Methods},
  author = {Erakovic, Mihael and Witteveen, Freek and Harley, Dylan and G\"unther, Jakob and Bensberg, Moritz and Meitei, Oinam Romesh and Cho, Minsik and Van Voorhis, Troy and Reiher, Markus and Christandl, Matthias},
  journal = {PRX Life},
  volume = {3},
  pages = {013003},
  numpages = {23},
  year = {2025},
  publisher = {American Physical Society},
  doi = {10.1103/PRXLife.3.013003},
  url = {https://link.aps.org/doi/10.1103/PRXLife.3.013003}
}

@article{holmes2022quantum,
  doi = {10.22331/q-2022-10-06-825},
  url = {https://doi.org/10.22331/q-2022-10-06-825},
  title = {Quantum algorithms from fluctuation theorems: {T}hermal-state preparation},
  author = {Holmes, Zoe and Muraleedharan, Gopikrishnan and Somma, Rolando D. and Subasi, Yigit and {\c{S}}ahino{\u{g}}lu, Burak},
  journal = {{Quantum}},
  issn = {2521-327X},
  publisher = {{Verein zur F{\"{o}}rderung des Open Access Publizierens in den Quantenwissenschaften}},
  volume = {6},
  pages = {825},
  year = {2022}
}

@article{shinaoka2015negative,
  title = {Negative sign problem in continuous-time quantum {Monte Carlo}: Optimal choice of single-particle basis for impurity problems},
  author = {Shinaoka, Hiroshi and Nomura, Yusuke and Biermann, Silke and Troyer, Matthias and Werner, Philipp},
  journal = {Physical Review B},
  volume = {92},
  number = {19},
  pages = {195126},
  year = {2015},
  publisher = {American Physical Society},
  doi = {10.1103/PhysRevB.92.195126}
}

@inproceedings{bergamaschi2025fast,
  title={Fast Mixing of Quantum Spin Chains at All Temperatures},
  author={Bergamaschi, Thiago and Chen, Chi-Fang},
  booktitle={Proceedings of the 58th Annual ACM Symposium on Theory of Computing},
  pages={835--844},
  year={2026},
  doi={10.1145/3798129.3800798}
}

@article{wilson1975renormalization,
  title={The renormalization group: Critical phenomena and the {Kondo} problem},
  author={Wilson, Kenneth G},
  journal={Reviews of Modern Physics},
  volume={47},
  number={4},
  pages={773--840},
  year={1975},
  publisher={APS},
  doi={10.1103/RevModPhys.47.773}
}

@article{anderson1961localized,
  title={Localized magnetic states in metals},
  author={Anderson, Philip Warren},
  journal={Physical Review},
  volume={124},
  number={1},
  pages={41--53},
  year={1961},
  publisher={APS},
  doi={10.1103/PhysRev.124.41}
}

@article{kondo1964resistance,
  title={Resistance minimum in dilute magnetic alloys},
  author={Kondo, Jun},
  journal={Progress of Theoretical Physics},
  volume={32},
  number={1},
  pages={37--49},
  year={1964},
  publisher={Oxford University Press},
  doi={10.1143/PTP.32.37}
}

@article{goldhaber1998kondo,
  title={{Kondo} effect in a single-electron transistor},
  author={Goldhaber-Gordon, David and Shtrikman, Hadas and Mahalu, Diana and Abusch-Magder, David and Meirav, U and Kastner, MA},
  journal={Nature},
  volume={391},
  number={6663},
  pages={156--159},
  year={1998},
  publisher={Nature Publishing Group UK London},
  doi={10.1038/34373}
}

@article{PhysRevB.21.1003,
  title = {Renormalization-group approach to the {Anderson} model of dilute magnetic alloys. {I}. Static properties for the symmetric case},
  author = {Krishna-Murthy, HR and Wilkins, JW and Wilson, KG},
  journal = {Physical Review B},
  volume = {21},
  issue = {3},
  pages = {1003--1043},
  numpages = {0},
  year = {1980},
  month = {Feb},
  publisher = {American Physical Society},
  doi = {10.1103/PhysRevB.21.1003},
  url = {https://link.aps.org/doi/10.1103/PhysRevB.21.1003}
}

@article{bulla2008numerical,
  title={Numerical renormalization group method for quantum impurity systems},
  author={Bulla, Ralf and Costi, Theo A and Pruschke, Thomas},
  journal={Reviews of Modern Physics},
  volume={80},
  number={2},
  pages={395--450},
  year={2008},
  publisher={APS},
  doi={10.1103/RevModPhys.80.395}
}

@article{nielsen2005fermionic,
  title={The Fermionic canonical commutation relations and the {Jordan-Wigner} transform},
  author={Nielsen, Michael A},
  journal={School of Physical Sciences, The University of Queensland},
  year={2005},
  url={https://futureofmatter.com/assets/fermions_and_jordan_wigner.pdf}
}

@article{werner2006continuous,
  title={Continuous-time solver for quantum impurity models},
  author={Werner, Philipp and Comanac, Armin and de' Medici, Luca and Troyer, Matthias and Millis, Andrew J},
  journal={Physical Review Letters},
  volume={97},
  number={7},
  pages={076405},
  year={2006},
  publisher={APS},
  doi={10.1103/PhysRevLett.97.076405}
}

@article{rubtsov2005continuous,
  title={Continuous-time quantum {Monte Carlo} method for fermions},
  author={Rubtsov, Alexey N and Savkin, Vladimir V and Lichtenstein, Alexander I},
  journal={Physical Review B—Condensed Matter and Materials Physics},
  volume={72},
  number={3},
  pages={035122},
  year={2005},
  publisher={APS},
  doi={10.1103/PhysRevB.72.035122}
}

@article{gull2011continuous,
  title={Continuous-time {Monte Carlo} methods for quantum impurity models},
  author={Gull, Emanuel and Millis, Andrew J and Lichtenstein, Alexander I and Rubtsov, Alexey N and Troyer, Matthias and Werner, Philipp},
  journal={Reviews of Modern Physics},
  volume={83},
  number={2},
  pages={349--404},
  year={2011},
  publisher={APS},
  doi={10.1103/RevModPhys.83.349}
}

@article{kuwahara2021improved,
  title={Improved thermal area law and quasilinear time algorithm for quantum {Gibbs} states},
  author={Kuwahara, Tomotaka and Alhambra, {\'A}lvaro M and Anshu, Anurag},
  journal={Physical Review X},
  volume={11},
  number={1},
  pages={011047},
  year={2021},
  publisher={APS},
  doi={10.1103/PhysRevX.11.011047}
}

@article{hastings2007quantum,
  title = {Quantum belief propagation: An algorithm for thermal quantum systems},
  author = {Hastings, M. B.},
  journal = {Physical Review B},
  volume = {76},
  issue = {20},
  pages = {201102(R)},
  numpages = {4},
  year = {2007},
  publisher = {American Physical Society},
  doi = {10.1103/PhysRevB.76.201102},
  url = {https://link.aps.org/doi/10.1103/PhysRevB.76.201102}
}

@article{kato2019quantum,
  title={Quantum Approximate {M}arkov Chains are Thermal},
  author={Kato, Kohtaro and Brand{\~a}o, Fernando GSL},
  journal={Communications in Mathematical Physics},
  volume={370},
  number={1},
  pages={117--149},
  year={2019},
  publisher={Springer},
  doi={10.1007/s00220-019-03485-6}
}

@article{brandao2019finite,
  title={Finite Correlation Length Implies Efficient Preparation of Quantum Thermal States},
  author={Brand{\~a}o, Fernando GSL and Kastoryano, Michael J},
  journal={Communications in Mathematical Physics},
  volume={365},
  number={1},
  pages={1--16},
  year={2019},
  publisher={Springer},
  doi={10.1007/s00220-018-3150-8}
}

@article{an2023linear,
  title = {Linear Combination of {Hamiltonian} Simulation for Nonunitary Dynamics with Optimal State Preparation Cost},
  author = {An, Dong and Liu, Jin-Peng and Lin, Lin},
  journal = {Physical Review Letters},
  volume = {131},
  issue = {15},
  pages = {150603},
  numpages = {6},
  year = {2023},
  publisher = {American Physical Society},
  doi = {10.1103/PhysRevLett.131.150603},
  url = {https://link.aps.org/doi/10.1103/PhysRevLett.131.150603}
}

@article{low2025optimal,
  title={Optimal quantum simulation of linear non-unitary dynamics},
  author={Low, Guang Hao and Somma, Rolando D},
  journal={arXiv:2508.19238},
  year={2025},
  url={https://arxiv.org/abs/2508.19238}
}

@article{haah2021quantum,
  title={Quantum algorithm for simulating real time evolution of lattice {H}amiltonians},
  author={Haah, Jeongwan and Hastings, Matthew B and Kothari, Robin and Low, Guang Hao},
  journal={SIAM Journal on Computing},
  volume={52},
  number={6},
  pages={FOCS18--250},
  year={2021},
  publisher={SIAM},
  doi={10.1137/18M1231511}
}

@article{low2018hamiltonian,
  title={{Hamiltonian} simulation in the interaction picture},
  author={Low, Guang Hao and Wiebe, Nathan},
  journal={arXiv:1805.00675},
  year={2018},
  url={https://arxiv.org/abs/1805.00675v2}
}

@article{chen2026time,
  title={Time-dependent {H}amiltonian simulation with optimal query complexity},
  author={Chen, Boyang and Gao, Minbo and Wang, Xinzhao and Zhou, Shuo},
  journal={arXiv:2608.06094},
  year={2026},
  url={https://arxiv.org/abs/2608.06094}
}

@article{chen2026gate,
  title={Gate-Efficient Implementation of the Query-Optimal Time-Dependent {H}amiltonian Simulation},
  author={Chen, Boyang and Gao, Minbo and Ji, Zhengfeng and Li, Tongyang and Wang, Xinzhao and Zhou, Shuo},
  journal={arXiv:2608.30629},
  year={2026},
  url={https://arxiv.org/abs/2608.30629}
}

@article{carlen2014remainder,
  title={Remainder terms for some quantum entropy inequalities},
  author={Carlen, Eric A and Lieb, Elliott H},
  journal={Journal of Mathematical Physics},
  volume={55},
  number={4},
  pages={042201},
  year={2014},
  publisher={AIP Publishing},
  doi={10.1063/1.4871575}
}

@article{yoder2014fixed,
  title = {Fixed-Point Quantum Search with an Optimal Number of Queries},
  author = {Yoder, Theodore J. and Low, Guang Hao and Chuang, Isaac L.},
  journal = {Physical Review Letters},
  volume = {113},
  issue = {21},
  pages = {210501},
  numpages = {5},
  year = {2014},
  month = {Nov},
  publisher = {American Physical Society},
  doi = {10.1103/PhysRevLett.113.210501},
  url = {https://link.aps.org/doi/10.1103/PhysRevLett.113.210501}
}

@inproceedings{gilyen2019quantum,
  title={Quantum singular value transformation and beyond: exponential improvements for quantum matrix arithmetics},
  author={Gily{\'e}n, Andr{\'a}s and Su, Yuan and Low, Guang Hao and Wiebe, Nathan},
  booktitle={Proceedings of the 51st Annual ACM SIGACT Symposium on Theory of Computing},
  pages={193--204},
  year={2019},
  doi={10.1145/3313276.3316366}
}

@book{gautschi2004orthogonal,
  title={Orthogonal polynomials: Computation and Approximation},
  author={Gautschi, Walter},
  year={2004},
  publisher={OUP Oxford}
}

@article{golub1969calculation,
  title={Calculation of {G}auss quadrature rules},
  author={Golub, Gene H and Welsch, John H},
  journal={Mathematics of Computation},
  volume={23},
  number={106},
  pages={221--230},
  year={1969},
  doi={10.2307/2004418}
}

@article{liang2002kondo,
  title = {{Kondo} resonance in a single-molecule transistor},
  author = {Liang, Wenjie and Shores, Matthew P. and Bockrath, Marc and Long, Jeffrey R. and Park, Hongkun},
  journal = {Nature},
  year = {2002},
  month = jun,
  volume = {417},
  number = {6890},
  pages = {725-729},
  publisher = {Springer Nature},
  doi = {10.1038/nature00790},
  pmid = {12066180}
}

@article{georges1992hubbard,
  title = {{Hubbard} model in infinite dimensions},
  author = {Georges, Antoine and Kotliar, Gabriel},
  journal = {Physical Review B},
  year = {1992},
  month = mar,
  volume = {45},
  number = {12},
  pages = {6479-6483},
  publisher = {American Physical Society (APS)},
  doi = {10.1103/physrevb.45.6479},
  pmid = {10000408}
}

@article{georges1996dynamical,
  title = {Dynamical mean-field theory of strongly correlated fermion systems and the limit of infinite dimensions},
  author = {Georges, Antoine and Kotliar, Gabriel and Krauth, Werner and Rozenberg, Marcelo J.},
  journal = {Reviews of Modern Physics},
  year = {1996},
  month = jan,
  volume = {68},
  number = {1},
  pages = {13-125},
  publisher = {American Physical Society (APS)},
  doi = {10.1103/revmodphys.68.13}
}

@article{RevModPhys.78.865,
  title = {Electronic structure calculations with dynamical mean-field theory},
  author = {Kotliar, G. and Savrasov, S. Y. and Haule, K. and Oudovenko, V. S. and Parcollet, O. and Marianetti, C. A.},
  journal = {Reviews of Modern Physics},
  volume = {78},
  issue = {3},
  pages = {865--951},
  numpages = {0},
  year = {2006},
  month = {Aug},
  publisher = {American Physical Society},
  doi = {10.1103/RevModPhys.78.865},
  url = {https://link.aps.org/doi/10.1103/RevModPhys.78.865}
}

@article{knizia2012density,
  title = {Density Matrix Embedding: A Simple Alternative to Dynamical Mean-Field Theory},
  author = {Knizia, Gerald and Chan, Garnet Kin-Lic},
  journal = {Physical Review Letters},
  year = {2012},
  month = nov,
  volume = {109},
  number = {18},
  pages = {186404},
  publisher = {American Physical Society (APS)},
  doi = {10.1103/physrevlett.109.186404},
  pmid = {23215304},
  url = {https://arxiv.org/pdf/1204.5783}
}

@article{PhysRevB.100.125165,
  title = {Dynamical mean field theory simulations with the adaptive sampling configuration interaction method},
  author = {Mejuto-Zaera, Carlos and Tubman, Norm M. and Whaley, K. Birgitta},
  journal = {Physical Review B},
  volume = {100},
  issue = {12},
  pages = {125165},
  numpages = {10},
  year = {2019},
  month = {Sep},
  publisher = {American Physical Society},
  doi = {10.1103/PhysRevB.100.125165},
  url = {https://link.aps.org/doi/10.1103/PhysRevB.100.125165}
}

@article{PhysRevB.100.115154,
  title = {Coupled-cluster impurity solvers for dynamical mean-field theory},
  author = {Zhu, Tianyu and Jim{\'e}nez-Hoyos, Carlos A. and McClain, James and Berkelbach, Timothy C. and Chan, Garnet Kin-Lic},
  journal = {Physical Review B},
  volume = {100},
  issue = {11},
  pages = {115154},
  numpages = {9},
  year = {2019},
  month = {Sep},
  publisher = {American Physical Society},
  doi = {10.1103/PhysRevB.100.115154},
  url = {https://link.aps.org/doi/10.1103/PhysRevB.100.115154}
}

@article{wouters2016practical,
  title={A practical guide to density matrix embedding theory in quantum chemistry},
  author={Wouters, Sebastian and Jim{\'e}nez-Hoyos, Carlos A and Sun, Qiming and Chan, Garnet K-L},
  journal={Journal of Chemical Theory and Computation},
  volume={12},
  number={6},
  pages={2706--2719},
  year={2016},
  publisher={ACS Publications},
  doi={10.1021/acs.jctc.6b00316}
}

@article{PhysRevLett.93.246403,
  title = {Dynamical Mean Field Theory with the Density Matrix Renormalization Group},
  author = {Garc\'{\i}a, Daniel J. and Hallberg, Karen and Rozenberg, Marcelo J.},
  journal = {Physical Review Letters},
  volume = {93},
  issue = {24},
  pages = {246403},
  numpages = {4},
  year = {2004},
  month = {Dec},
  publisher = {American Physical Society},
  doi = {10.1103/PhysRevLett.93.246403},
  url = {https://link.aps.org/doi/10.1103/PhysRevLett.93.246403}
}

@article{PhysRevB.90.115124,
  title = {{Chebyshev} matrix product state impurity solver for dynamical mean-field theory},
  author = {Wolf, F. Alexander and McCulloch, Ian P. and Parcollet, Olivier and Schollw\"ock, Ulrich},
  journal = {Physical Review B},
  volume = {90},
  issue = {11},
  pages = {115124},
  numpages = {17},
  year = {2014},
  month = {Sep},
  publisher = {American Physical Society},
  doi = {10.1103/PhysRevB.90.115124},
  url = {https://link.aps.org/doi/10.1103/PhysRevB.90.115124}
}

@article{PhysRevB.86.165128,
  title = {Truncated configuration interaction expansions as solvers for correlated quantum impurity models and dynamical mean-field theory},
  author = {Zgid, Dominika and Gull, Emanuel and Chan, Garnet Kin-Lic},
  journal = {Physical Review B},
  volume = {86},
  issue = {16},
  pages = {165128},
  numpages = {12},
  year = {2012},
  month = {Oct},
  publisher = {American Physical Society},
  doi = {10.1103/PhysRevB.86.165128},
  url = {https://link.aps.org/doi/10.1103/PhysRevB.86.165128}
}

@article{keldysh1965diagram,
  author  = {Keldysh, L. V.},
  title   = {Diagram Technique for Nonequilibrium Processes},
  journal = {Soviet Physics--Journal of Experimental and Theoretical Physics},
  volume  = {20},
  pages   = {1018--1026},
  year    = {1965},
  url     = {https://www.jetp.ras.ru/cgi-bin/dn/e_020_04_1018.pdf}
}

@article{aoki2014nonequilibrium,
  title={Nonequilibrium dynamical mean-field theory and its applications},
  author={Aoki, Hideo and Tsuji, Naoto and Eckstein, Martin and Kollar, Marcus and Oka, Takashi and Werner, Philipp},
  journal={Reviews of Modern Physics},
  volume={86},
  number={2},
  pages={779--837},
  year={2014},
  publisher={APS},
  doi={10.1103/RevModPhys.86.779}
}

@article{bravyi2005lagrangian,
  author  = {Bravyi, Sergey},
  title   = {{Lagrangian} representation for fermionic linear optics},
  journal = {Quantum Information \& Computation},
  volume  = {5},
  number  = {3},
  pages   = {216--238},
  year    = {2005},
  doi     = {10.26421/QIC5.3-3},
  url     = {https://arxiv.org/abs/quant-ph/0404180}
}

@article{schuch2009computational,
  title = {Computational complexity of interacting electrons and fundamental limitations of density functional theory},
  author = {Schuch, Norbert and Verstraete, Frank},
  journal = {Nature Physics},
  year = {2009},
  month = {Aug},
  volume = {5},
  number = {10},
  pages = {732-735},
  publisher = {Springer Nature},
  doi = {10.1038/nphys1370},
  url = {https://www.nature.com/articles/nphys1370.pdf}
}

@article{ogorman2022intractability,
  title = {Intractability of Electronic Structure in a Fixed Basis},
  author = {O'Gorman, Bryan and Irani, Sandy and Whitfield, James and Fefferman, Bill},
  journal = {PRX Quantum},
  volume = {3},
  issue = {2},
  pages = {020322},
  numpages = {39},
  year = {2022},
  month = {May},
  publisher = {American Physical Society},
  doi = {10.1103/PRXQuantum.3.020322},
  url = {https://link.aps.org/doi/10.1103/PRXQuantum.3.020322}
}

@article{anders2005real,
  title = {Real-Time Dynamics in Quantum-Impurity Systems: A Time-Dependent Numerical Renormalization-Group Approach},
  author = {Anders, Frithjof B. and Schiller, Avraham},
  journal = {Physical Review Letters},
  volume = {95},
  issue = {19},
  pages = {196801},
  numpages = {4},
  year = {2005},
  month = {Oct},
  publisher = {American Physical Society},
  doi = {10.1103/PhysRevLett.95.196801},
  url = {https://link.aps.org/doi/10.1103/PhysRevLett.95.196801}
}

@article{qin2022hubbard,
  title={The {H}ubbard model: A computational perspective},
  author={Qin, Mingpu and Sch{\"a}fer, Thomas and Andergassen, Sabine and Corboz, Philippe and Gull, Emanuel},
  journal={Annual Review of Condensed Matter Physics},
  volume={13},
  number={1},
  pages={275--302},
  year={2022},
  publisher={Annual Reviews},
  doi={10.1146/annurev-conmatphys-090921-033948}
}

@article{madhavan1998tunneling,
  title={Tunneling into a Single Magnetic Atom: Spectroscopic Evidence of the {K}ondo Resonance},
  author={Madhavan, V and Chen, W and Jamneala, T and Crommie, MF and Wingreen, NS},
  journal={Science},
  volume={280},
  number={5363},
  pages={567--569},
  year={1998},
  publisher={American Association for the Advancement of Science},
  doi={10.1126/science.280.5363.567}
}

@article{lin2020nearoptimal,
  doi = {10.22331/q-2020-12-14-372},
  url = {https://doi.org/10.22331/q-2020-12-14-372},
  title = {Near-optimal ground state preparation},
  author = {Lin, Lin and Tong, Yu},
  journal = {{Quantum}},
  issn = {2521-327X},
  publisher = {{Verein zur F{\"{o}}rderung des Open Access Publizierens in den Quantenwissenschaften}},
  volume = {4},
  pages = {372},
  month = dec,
  year = {2020}
}

@inproceedings{bernstein1993quantum,
  title={Quantum complexity theory},
  author={Bernstein, Ethan and Vazirani, Umesh},
  booktitle={Proceedings of the Twenty-Fifth Annual ACM Symposium on Theory of Computing},
  pages={11--20},
  year={1993},
  doi={10.1145/167088.167097}
}

@article{shepherd2006computation,
  title={Computation with unitaries and one pure qubit},
  author={Shepherd, Dan},
  journal={arXiv:quant-ph/0608132},
  year={2006},
  url={https://arxiv.org/abs/quant-ph/0608132}
}

@article{brandao2008entanglement,
  title={Entanglement theory and the quantum simulation of many-body physics},
  author={Brand{\~a}o, Fernando GSL},
  journal={arXiv:0810.0026},
  year={2008},
  url={https://arxiv.org/abs/0810.0026}
}

@article{knill1998power,
  title = {Power of One Bit of Quantum Information},
  author = {Knill, E. and Laflamme, R.},
  journal = {Physical Review Letters},
  volume = {81},
  issue = {25},
  pages = {5672--5675},
  numpages = {0},
  year = {1998},
  month = {Dec},
  publisher = {American Physical Society},
  doi = {10.1103/PhysRevLett.81.5672},
  url = {https://link.aps.org/doi/10.1103/PhysRevLett.81.5672}
}

@article{shor2008estimating,
author = {Shor, Peter W. and Jordan, Stephen P.},
title = {Estimating {J}ones polynomials is a complete problem for one clean qubit},
year = {2008},
publisher = {Rinton Press, Incorporated},
volume = {8},
number = {8},
journal = {Quantum Information \& Computation},
pages = {681–714},
numpages = {34},
doi={10.26421/QIC8.8-9-1},
url={https://arxiv.org/abs/0707.2831}
}

@book{kitaev2002classical,
  title={Classical and Quantum Computation},
  author={Kitaev, Alexei Yu and Shen, Alexander and Vyalyi, Mikhail N},
  volume={47},
  year={2002},
  publisher={American Mathematical Society Providence, RI}
}

@article{poulin2011quantum,
  title = {Quantum Simulation of Time-Dependent {H}amiltonians and the Convenient Illusion of {H}ilbert Space},
  author = {Poulin, David and Qarry, Angie and Somma, Rolando and Verstraete, Frank},
  journal = {Physical Review Letters},
  volume = {106},
  issue = {17},
  pages = {170501},
  numpages = {4},
  year = {2011},
  month = {Apr},
  publisher = {American Physical Society},
  doi = {10.1103/PhysRevLett.106.170501},
  url = {https://link.aps.org/doi/10.1103/PhysRevLett.106.170501}
}

@inproceedings{schulman1999molecular,
  title={Molecular scale heat engines and scalable quantum computation},
  author={Schulman, Leonard J and Vazirani, Umesh V},
  booktitle={Proceedings of the Thirty-First Annual ACM Symposium on Theory of Computing},
  pages={322--329},
  year={1999},
  doi={10.1145/301250.301332}
}
\addcontentsline{toc}{section}{References}
\addtocontents{toc}{\protect\setcounter{tocdepth}{1}}

\appendix

\section{Preprocessing the Hamiltonian in the hybridization form}\label{app:cano_bath}

\subsection{Canonical form of a quadratic Hamiltonian}
We first review the standard canonical transformation used to diagonalize quadratic fermionic Hamiltonians.

\begin{lemma}[Canonical form of a quadratic Hamiltonian]\label{lem:canonical_quadratic}
Let $r\leq n$ be an integer. For any quadratic Hamiltonian
    $Q=\frac{i}{4}\sum_{j,k=1}^{2r}h_{jk}\gamma_j\gamma_k,$
where $h\in\R^{2r\times 2r}$ is real antisymmetric, one can choose annihilation operators $a_1,\ldots,a_r$ such that
\begin{align}
    Q=-\alpha\I+\sum_{q=1}^r\lambda_q a_q^\dagger a_q,
    \qquad
    0\leq\lambda_q\leq\|h\|,
    \qquad
    \alpha=\sum_{q=1}^r\lambda_q/2.\label{eq:canonical_tranform}
\end{align}
\end{lemma}

\begin{proof}
Since $h$ is real antisymmetric, there exists an orthogonal matrix
$G\in\R^{2r\times 2r}$ such that
\begin{align}
    G^\T h G
    =
    \bigoplus_{q=1}^r
    \begin{bmatrix}
        0 & \lambda_q\\
        -\lambda_q & 0
    \end{bmatrix},
    \qquad
    0\leq\lambda_q\leq\|h\|.
\end{align}
Define
    $\widetilde{\gamma}_j
    \coloneqq
    \sum_{s=1}^{2r}G_{sj}\gamma_s,$ and
    $a_q
    \coloneqq
    \frac{\widetilde{\gamma}_{2q-1}+i\widetilde{\gamma}_{2q}}{2}.$
One can check that $\{a_q\}_q$ satisfy the canonical anticommutation relation, and 
\begin{align}
    Q= \sum_{q=1}^r \frac{i}{4}(2\lambda_q) \widetilde{\gamma}_{2q-1}\widetilde{\gamma}_{2q}= -\alpha\I+\sum_{q=1}^r\lambda_q a_q^\dagger a_q,
  \end{align}
\end{proof}

We now use Lemma~\ref{lem:canonical_quadratic} to derive Eq.~(\ref{eq:impurity_model_hyb}). We separate the quadratic Hamiltonian $H_0$ in Eq.~(\ref{eq:impurity_model}) into the impurity part $H_{0,\mathrm{imp}}$, the bath part $H_{0,\mathrm{bath}}$, and the hybridization part $H_{0,\mathrm{hyb}}$ connecting the impurity and bath operators,
\begin{align*}
    H_{0,\mathrm{imp}}
    \coloneqq
    \frac{i}{4}\sum_{j,k=1}^{2m}h_{jk}\gamma_j\gamma_k,\quad
    H_{0,\mathrm{bath}}
    \coloneqq
    \alpha\I+
    \frac{i}{4}\sum_{j,k=2m+1}^{2n}h_{jk}\gamma_j\gamma_k,\quad
    H_{0,\mathrm{hyb}}
    \coloneqq
    \frac{i}{2}\sum_{1\leq j\leq 2m<k\leq 2n}
    h_{jk}\gamma_j\gamma_k.
\end{align*}
For convenience we denote the principal submatrix $(h_{j,k})_{j,k=2m+1}^{2n}$ as $h_{\mathrm{bath}}$.
Define
    $K_0\coloneqq H_{0,\mathrm{imp}}+H_{\mathrm{imp}}$ as the part of the Hamiltonian that acts only on the impurity. We can write 
    \begin{align}
    	H = H_{0,\mathrm{bath}} + K_0 + H_{0,\mathrm{hyb}}. 
    \end{align} 
The first $m$ annihilation operators are defined as
$a_j
=
\frac{\gamma_{2j-1}+i\gamma_{2j}}{2},
1\leq j\leq m.$

To define $a_j$ for $j>m$, apply Lemma~\ref{lem:canonical_quadratic} to the quadratic part of $H_{0,\mathrm{bath}}$. Denote the resulting  single-particle energies by $\lambda_1,\ldots,\lambda_{n-m}$, and the annihilation operators by $a_{m+1},\ldots,a_n$, where 
\begin{align}
    H_{0,\mathrm{bath}}
    =
    \left(
    \alpha-\frac12\sum_{q=1}^{n-m}\lambda_q
    \right)\I
    +
    \sum_{q=1}^{n-m}
    \lambda_q a_{m+q}^\dagger a_{m+q}, \qquad  0\leq \lambda_q \leq \|h_{\mathrm{bath}}\|\leq \es.
\end{align}
We set
\begin{align}
    \alpha\coloneqq\frac12\sum_{q=1}^{n-m}\lambda_q,\qquad 
    E_0
    \coloneqq
    \operatorname{diag}(\lambda_1,\ldots,\lambda_{n-m}).
  \end{align}
Then
\begin{align}
    H_{0,\mathrm{bath}}=\dE(E_0),
    \qquad
    0\leq E_0\leq\es\I.
\end{align}

Finally, after the bath canonical transformation, $H_{0,\mathrm{hyb}}$ remains bilinear in the impurity Majorana operators and the bath Majorana operators. Since each transformed bath Majorana operator is a linear combination of the bath creation and annihilation operators, there exist bath vectors $\mu_1,\ldots,\mu_{2m}$ such that
\begin{align}
    H_{0,\mathrm{hyb}}
    =
    \sum_{j=1}^{2m}
    \left(
        a^\dagger(\mu_j)\gamma_j
        +
        \gamma_j a(\mu_j)
    \right).
\end{align}
Thus Eq.~(\ref{eq:impurity_model_hyb}) follows by setting $O_j=\gamma_j$.

\subsection{Handling zero eigenvalues of $E_0$}
Here we explain why without loss of generality we can assume that there exists a strictly positive number $\omega$ such that
    $$0<\omega\I\leq E_0\leq\es\I.$$

Recall that $E_0$ acts on $W_{\mathrm{bath}}$. Let $P$ be the orthogonal projector onto $\ker(E_0)$. The dimension of
    $\operatorname{span}\{P\mu_1,\ldots,P\mu_{2m}\}$
is no more than $2m$. For simplicity, we include this subspace in the impurity single-particle space and redefine
\begin{align}
	W_{\mathrm{imp}}''\coloneqq W_{\mathrm{imp}}\oplus \operatorname{span}\{P\mu_1,\ldots,P\mu_{2m}\}, \qquad  W_{\mathrm{bath}}''\coloneqq W_{\mathrm{imp}}''^\perp, 
\end{align}
Accordingly we 
 absorb the corresponding hybridization terms into $K_0$, that is, we write $a(\mu_j)=a(P\mu_j) + a((\I-P)\mu_j)$, and merge the hybridization terms that involve $a(P\mu_j)$ into $K_0$.

Note that for any vector    $v\in
    \ker(E_0)\ominus
    \operatorname{span}\{P\mu_1,\ldots,P\mu_{2m}\},$
the annihilation operator $a(v)$ appears neither in $\dE(E_0)$ nor in the hybridization term, nor in the impurity term. Since the Hamiltonian has even parity, under the Fock-space decomposition into the remaining modes and these zero modes, the Hamiltonian takes the form
$H=H_{\mathrm{active}}\otimes \I.$ Hence those modes form a decoupled free-fermion sector that can be
treated separately. After removing these modes, replacing each $\mu_j$ by $(\I-P)\mu_j$, and relabeling the remaining modes, $E_0$ is strictly positive on the bath single-particle space as needed.

\section{Detailed calculation of the occupation recursion}\label{sec:cal_recursion}

\subsection{Counting and enumerating the set $\{x:D(x)\leq T\}$}
\label{app:count_bound}

\begin{lemma}[Counting bound]
\label{lem:count_set}
Assume that
$h_d\geq\min\{d-1,\cut\}$ for every $d\geq1$, where 
$\cut\geq\max\{1,\frac18\log n\}$.
Then
\begin{align}
\left|
\left\{
x:D(x)\leq T
\right\}
\right|
\leq
\exp\!\left(
8T+3\chi
\right).
\end{align}
Here $\chi$, defined in Eq.~(\ref{eq:chi}), is a bound on the maximum number of bath modes at a given depth, summed over all bands, plus  the enlarged-impurity size.
\end{lemma}

\begin{proof}
Write $\widetilde{h}_d=\min\{d-1,\cut\}$.
For every $x$ and every $\theta>0$, one can check that
$\mathbf 1_{\{D(x)\leq T\}}
\leq
\exp\!\left(
\theta T-\theta\sum_i\widetilde{h}_dx_i
\right).$
We replace the hard cutoff by an exponential weight:
 \begin{equation}
\left|
\left\{
x:D(x)\leq T
\right\}
\right|
=
\sum_x\mathbf 1_{\{D(x)\leq T\}}
\leq
e^{\theta T}
\sum_x
e^{-\theta\sum_i\widetilde{h}_dx_i} \leq  e^{\theta T} 2^{\chi}\prod_{i\in \bS_{\mathrm{bath}}, d(i)\geq 2}  (1+ e^{-\theta \widetilde{h}_{d(i)}}).
\label{eq:222}
\end{equation}
Here the last inequality follows because the weighted sum factorizes over the modes, and the enlarged-impurity modes and depth-1 modes ($\tilde{h}_{d(i)}=0$) each contribute a factor $1+1=2$, giving a total factor of at most $2^\chi$.

We divide the modes with $i\in \bS_{\mathrm{bath}}, d(i)\geq 2$ into two groups:  for a mode $i$ where $k\coloneqq d(i)-1< \cut$, it contributes a factor of $(1+e^{-\theta k})$, where for each $k$ there are at most $\chi$ such modes; for a mode $i$ where $d(i)-1\geq\cut$, it contributes a factor of $(1+e^{-\theta \cut})$, and there
are at most $n$ such modes.  
Hence we have
\begin{align}
\prod_{i\in \bS_{\mathrm{bath}}, d(i)\geq 2}  (1+ e^{-\theta \widetilde{h}_{d(i)}})&\leq
\prod_{\substack{1\leq k< \cut}}
\left(
1+e^{-\theta k}
\right)^\chi
\left(
1+e^{-\theta\cut}
\right)^n
\leq
\exp\!\left(
\chi\sum_{k\geq1}e^{-\theta k}
+
ne^{-\theta \cut}
\right).
\end{align}
where in the last inequality we use $\log(1+e^x)\leq e^x$ and relax $1\leq k< \cut$ to $1\leq k$. 
Taking $\theta=8$, since $\cut\geq\frac18\log n$, we have 
$ne^{-8\cut}\leq 1$. Therefore,
\begin{align}
\left|
\left\{
x:D(x)\leq T
\right\}
\right|
\leq
\exp\!\left(
8T+3\chi
\right).
\end{align}
\end{proof}

\paragraph{Algorithm to enumerate $\{x: D(x)\leq T\}$.}
We enumerate $x$ one bit at a time while keeping track of the current value of $D(x)$. For each bit, we consider $x_i=0$ and $x_i=1$, and discard a partial string if its weighted occupation number exceeds $T$. Since all the weights are nonnegative, this enumerates exactly the set $\{x:D(x)\leq T\}$. Every partial string considered by the algorithm is a prefix of some $x\in\bS$, so there are at most $\cO(n|\bS|)$ such partial strings.  The runtime is $\cO(n|\bS|)$.

\subsection{Proof of the commutator identity}\label{app:commutator-identity}

\begin{proof}[Proof of Lemma~\ref{lem:commutator-identity}]
Recall that $A_I=a_{i_q}\cdots a_{i_1}$. The product rule gives
\begin{align}
 [A_I,H]
 =\sum_{r=1}^q
 a_{i_q}\cdots a_{i_{r+1}}
 [a_{i_r},H]
 a_{i_{r-1}}\cdots a_{i_1}.
 \label{eq:commutator-product}
\end{align}
Recall that Lemma~\ref{lem:elementary-commutators} gives
\begin{align}
 [a_{i_r},H]
 =&\sum_{j:\ell(j)=\ell(i_r)}
 \langle e_j,E_{\ell(i_r)}e_{i_r}\rangle a_j
 +\sum_s
 \langle u_{\ell(i_r)s},e_{i_r}\rangle G_{\ell(i_r)s}.
 \label{eq:ai-H}
\end{align}
We treat the two sums on the right-hand side of the equation separately.

\paragraph{Bath-replacement term.}
Fix $r$ and a mode $a_j$ in the same band as $a_{i_r}$. The corresponding term
obtained by inserting the first sum of \eqref{eq:ai-H} into
\eqref{eq:commutator-product} is
\begin{align}
 \langle e_j,E_{\ell(i_r)}e_{i_r}\rangle
 a_{i_q}\cdots a_{i_{r+1}}a_j
 a_{i_{r-1}}\cdots a_{i_1}.
 \label{eq:one-bath-term}
\end{align}
The standard convention is that the annihilation operators attached to an
increasing list are written in decreasing-index order. We therefore move
$a_j$ through the other annihilation operators until this order is restored.
Since $a_ja_k=-a_ka_j$ for distinct modes, this produces a sign
$\varepsilon_I(i_r,j)\in\{\pm1\}$.

If $j\in I\setminus\{i_r\}$, the reordered product contains $a_j^2=0$, so
this term vanishes.  Otherwise \eqref{eq:one-bath-term} becomes
\begin{align}
 \varepsilon_I(i_r,j)
 \langle e_j,E_{\ell(i_r)}e_{i_r}\rangle
 A_{I^{(i_r\to j)}}.
 \label{eq:reordered-bath-term}
\end{align}
Summing \eqref{eq:reordered-bath-term} over $r$ and $j$ gives precisely
$B_I$.

\paragraph{Hybridization contribution term.}
Fix $r$ and $s$. The corresponding term obtained by inserting the second
sum of \eqref{eq:ai-H} into \eqref{eq:commutator-product} is
\begin{align}
 \langle u_{\ell(i_r)s},e_{i_r}\rangle
 a_{i_q}\cdots a_{i_{r+1}}G_{\ell(i_r)s}
 a_{i_{r-1}}\cdots a_{i_1}.
 \label{eq:one-source-term}
\end{align}
The operator $G_{\ell(i_r)s}$ is a linear combination of annihilation
operators on the enlarged-impurity modes. Since these modes are disjoint from
the residual-bath modes, the canonical anticommutation relation implies that $G_{\ell(i_r)s}$ anticommutes
with every residual-bath annihilation operator. Moving
$G_{\ell(i_r)s}$ through the annihilation operators to its left produces a
sign $\eta_I(i_r)\in\{\pm1\}$. Thus \eqref{eq:one-source-term} becomes
\begin{align}
 \eta_I(i_r)\langle u_{\ell(i_r)s},e_{i_r}\rangle
 G_{\ell(i_r)s}A_{I\setminus\{i_r\}}.
 \label{eq:reordered-source-term}
\end{align}

Summing
\eqref{eq:reordered-source-term} over $r$ and $s$ gives precisely $C_I$.
\end{proof}

\subsection{Recursion solution for the ground state}

\begin{proof}[of Theorem \ref{thm:main_ground_state}] 
We prove by induction that the solution to the recursion in Eq.~(\ref{eq:ground_state_recur}) is 
\begin{equation}
 F^{(\infty)}_\bq
 \leq
 \prod_{\ell=0}^L
 \frac{1}{q_\ell!}
 \left(\frac{4\xi_\ell}{\omega_\ell^2}\right)^{q_\ell}.
 \label{eq:F-product-bound}
\end{equation}
The induction is with respect to the total occupation number  $q\coloneqq\sum_{\ell=0}^L q_\ell.$ The claim is immediate for $q=0$ where $\bq=(0,\ldots,0)$ and $F^{(\infty)}_{\bq}=1$.
Now fix $\bq\neq\mathbf 0$ and assume that
\eqref{eq:F-product-bound} holds for all cases of total degree
$|\bq|-1$.

Using the recursion formula in Eq.~(\ref{eq:ground_state_recur}) and the induction hypothesis, we have
\begin{align}
 \left(\sum_{\ell}q_\ell\omega_\ell\right)F^{(\infty)}_\bq
 &\leq
4\sum_{\ell:q_\ell>0}\frac{\xi_\ell}{\omega_\ell}
 F^{(\infty)}_{\bq-\bel} \notag\\
 &\leq
4\sum_{\ell:q_\ell>0}\frac{\xi_\ell}{\omega_\ell}
 \frac{1}{(q_\ell-1)!}
 \left(\frac{4\xi_\ell}{\omega_\ell^2}\right)^{q_\ell-1}
 \prod_{r\neq\ell}
 \frac{1}{q_r!}
 \left(\frac{4\xi_r}{\omega_r^2}\right)^{q_r}.
\end{align}
Factoring out the common term gives
\begin{align}
 \left(\sum_{\ell}q_\ell\omega_\ell\right)F^{(\infty)}_\bq
 \leq
 \left[
 \prod_{r=0}^L
 \frac{1}{q_r!}
 \left(\frac{4\xi_r}{\omega_r^2}\right)^{q_r}
 \right]
 \sum_{\ell:q_\ell>0}
 q_\ell\omega_\ell =
 \left(\sum_{\ell}q_\ell\omega_\ell\right)
 \prod_{r=0}^L
 \frac{1}{q_r!}
 \left(\frac{4\xi_r}{\omega_r^2}\right)^{q_r}.
\end{align}
Since $\sum_\ell q_\ell\omega_\ell>0$, canceling this common factor proves
\eqref{eq:F-product-bound}.  
Then recall the definition of $\kappa$ in Eq.~(\ref{eq:kappa}).

For any residual bath mode $i$, we set $\bq = \mathbf{e}_{\ell(i)}$. Recall that from Eq.~(\ref{eq:xi}) we have $\xi_\ell \omega_\ell^{-2} \leq 8m$. Taking $\lambda=1/4$, and using the definition of $F_\bq^{(\infty)}$ and Eq.~(\ref{eq:F-product-bound}), we get 
\begin{align}
    e^{\frac{1}{2}(d(i)-1)}\tr\left(\rho n_i\right) \leq 32m.
\end{align}

  Setting $\lambda=1/4$ and using Eq.~(\ref{eq:exp_Fq}) and Eq.~(\ref{eq:F-product-bound}), we get 

\begin{align}
 \tr\left(\rho^{(\infty)}\, e^{ \frac{1}{2}\mathsf D^{(\infty)}}\right) \leq 
 \sum_{\bq}F_{\bq} 
\leq
 \prod_{\ell=0}^L
 \left(
 \sum_{q_\ell\geq0}
 \frac{1}{q_\ell!}
 \left(\frac{4\xi_\ell}{\omega_\ell^2}\right)^{q_\ell}
 \right) 
 =
 \exp\left(
 4\sum_{\ell=0}^L\frac{\xi_\ell}{\omega_\ell^2}
 \right).
 \label{eq:weighted-all-moments}
\end{align}
\end{proof}

\subsection{Recursion solution for the thermal state}\label{sec:thermal_recursion_solution}

We first give the solution of the thermal recursion in Lemma \ref{lem:thermal-recursion}.

\begin{lemma}[Thermal recursion solution]
\label{lem:thermal-recursion-solution} Assume $\omega>0$, so $\omega_\ell=2^\ell \omega>0$ for all $\ell$.
For every $\bq=(q_0,\ldots,q_L)$,
\begin{align}
F_\bq^{(\beta)}
\leq
\prod_{\ell=0}^L
\frac{1}{q_\ell!}
\left(
\frac{8\xi_\ell}{\omega_\ell^2}
\right)^{q_\ell}+
\sum_{\mathbf 0<\br\leq\bq}
\frac{4\operatorname{Err}_{\br}^{(\beta)}}
{\alpha_{\br}}
\prod_{\ell=0}^L
\frac{1}{(q_\ell-r_\ell)!}
\left(
\frac{8\xi_\ell}{\omega_\ell^2}
\right)^{q_\ell-r_\ell}.
\label{eq:thermal-recursion-solution}
\end{align}
where $\br\leq\bq$ means that $r_\ell\leq q_\ell$ for every $\ell$.
\end{lemma}

\begin{proof}
We argue by induction on $|\bq|\coloneqq\sum_{\ell=0}^Lq_\ell$. For
$\bq=\mathbf 0$, the first product on the right-hand side equals one and
the second sum is empty, so the claim follows from
$F_{\mathbf 0}^{(\beta)}=1$. Now fix $\bq\neq\mathbf 0$ and assume that
Eq.~\eqref{eq:thermal-recursion-solution} holds for all cases of
total degree at most $|\bq|-1$.

First, a direct calculation by factoring out the common terms gives
\begin{align}
&8\sum_{\ell:q_\ell>0}
\frac{\xi_\ell}{\omega_\ell}
\frac{1}{(q_\ell-1)!}
\left(
\frac{8\xi_\ell}{\omega_\ell^2}
\right)^{q_\ell-1}
\prod_{k\neq\ell}
\frac{1}{q_k!}
\left(
\frac{8\xi_k}{\omega_k^2}
\right)^{q_k}=
\alpha_\bq
\prod_{k=0}^L
\frac{1}{q_k!}
\left(
\frac{8\xi_k}{\omega_k^2}
\right)^{q_k}.
\label{eq:thermal-homogeneous-recursion}
\end{align}
Similarly, for every
$\mathbf 0<\br\leq\bq$ with $\br\neq\bq$, the above equation holds when replacing $\bq,q_\ell, q_k$ with $\bq-\br$ and $q_\ell-r_\ell, q_k-r_k$.

Applying Lemma~\ref{lem:thermal-recursion}, substituting the induction
hypothesis, and using Eqs.~\eqref{eq:thermal-homogeneous-recursion} w.r.t $\bq$ and
its variants w.r.t. $\bq-\bel$,  $\bq-\br$, we obtain
\begin{align}
\alpha_\bq F_\bq^{(\beta)}
\leq
\alpha_\bq
\prod_{\ell=0}^L
\frac{1}{q_\ell!}
\left(
\frac{8\xi_\ell}{\omega_\ell^2}
\right)^{q_\ell}+
\sum_{\substack{\mathbf 0<\br\leq\bq\\
                  \br\neq\bq}}
\frac{4\operatorname{Err}_{\br}^{(\beta)}}
{\alpha_{\br}}
\alpha_{\bq-\br}
\prod_{\ell=0}^L
\frac{1}{(q_\ell-r_\ell)!}
\left(
\frac{8\xi_\ell}{\omega_\ell^2}
\right)^{q_\ell-r_\ell}
+4\operatorname{Err}_\bq^{(\beta)}.
\label{eq:thermal-recursion-induction}
\end{align}
Since
$\alpha_{\bq-\br}=\alpha_\bq-\alpha_{\br}\leq\alpha_\bq$,
the right-hand side of Eq.~\eqref{eq:thermal-recursion-induction} is at
most $\alpha_\bq$ times the right-hand side of
Eq.~\eqref{eq:thermal-recursion-solution}. Dividing by
$\alpha_\bq>0$ proves the claim.
\end{proof}

\begin{proof}[of Theorem \ref{thm:thermal-exponential-bound}]
Note that 
for any functions $F_\br,G_{\mathbf{s}}$, we have $\sum_{\bq}\sum_{0<\br\leq \bq} F_{\br}G_{\bq-\br}=\left(\sum_{\br>0} F_{\br}\right)\left(\sum_{\mathbf{s}\geq 0} G_{\mathbf{s}}\right)$. Thus we have
\begin{align}
	\sum_\bq
\sum_{\mathbf 0<\br\leq\bq}
\frac{4\operatorname{Err}_{\br}^{(\beta)}}
{\alpha_{\br}}
\prod_{\ell=0}^L
\frac{1}{(q_\ell-r_\ell)!}
\left(
\frac{8\xi_\ell}{\omega_\ell^2}
\right)^{q_\ell-r_\ell} = 
\left(
\sum_{\br> 0}
\frac{4\operatorname{Err}_{\br}^{(\beta)}}
{\alpha_{\br}}
\right)
\prod_{\ell=0}^L\left[
\sum_{q_\ell\geq0}
\frac{1}{q_\ell!} 
\left(
\frac{8\xi_\ell}{\omega_\ell^2}
\right)^{q_\ell}
\right]
\end{align}

Thus by Lemma~\ref{lem:thermal-recursion-solution}, summing over $\bq$ gives
\begin{align}
\sum_\bq F_\bq^{(\beta)}
\leq
\prod_{\ell=0}^L
\left[
\sum_{q_\ell\geq0}
\frac{1}{q_\ell!}
\left(
\frac{8\xi_\ell}{\omega_\ell^2}
\right)^{q_\ell}
\right]
\left(
1+
\sum_{\br>\mathbf 0}
\frac{4\operatorname{Err}_{\br}^{(\beta)}}
{\alpha_{\br}}
\right)=
e^{8\kappa}
\left(
1+
\sum_{\br> 0}
\frac{4\operatorname{Err}_{\br}^{(\beta)}}
{\alpha_{\br}}
\right).
\label{eq:thermal-sum-before-error}
\end{align}

Note that $\alpha_{\br}\geq\omega|\br|$. Regroup $\sum_{\br>0, I\sim \br}$ as $\sum_{q\geq 1, I:|I|=q}$. By  the definition of 
$\operatorname{Err}_{\br}^{(\beta)}$ with $\lambda=1/4$ we have
\begin{align}
\sum_{\br> 0}
\frac{4\operatorname{Err}_{\br}^{(\beta)}}
{\alpha_{\br}}
&\leq
\frac{4}{e\beta\omega}
\sum_{q\geq1}\frac{1}{q}
\sum_{|I|=q}
\prod_{i\in I}
\exp\!\left(
-\frac{\beta\omega_{\ell(i)}}{8}
+\frac{1}{2}\min\{d(i)-1,\cut\}
\right)\notag\\
&\leq
\frac{4}{e\beta\omega}
\sum_{q\geq1}
\frac{\eta_\cut ^q}{q\,q!}
\leq
\frac{4\eta_\cut e^{\eta_\cut }}{e\beta\omega}.
\label{eq:sum-thermal-error}
\end{align}
For the last inequality in Eq.~(\ref{eq:sum-thermal-error}), we use $\sum_{q\geq1}
\frac{\eta_\cut ^q}{q\,q!}
\leq e^{\eta_\cut}-1\leq \eta_\cut e^{\eta_\cut}$.
The second inequality in Eq.~(\ref{eq:sum-thermal-error}) follows from
\begin{align}
\sum_{|I|=q}
\prod_{i\in I}
\exp\!\left(
-\frac{\beta\omega_{\ell(i)}}{8}
+\frac{1}{2}\min\{d(i)-1,\cut\}
\right)\leq
\frac{1}{q!}
\left[
\sum_i
\exp\!\left(
-\frac{\beta\omega_{\ell(i)}}{8}
+\frac{1}{2}\min\{d(i)-1,\cut\}
\right)
\right]^q
=
\frac{\eta_\cut ^q}{q!}.
\end{align}

Combining Eqs.~\eqref{eq:thermal-sum-before-error} and
\eqref{eq:sum-thermal-error} proves the upper bound on
$\sum_\bq F_\bq^{(\beta)}$. Then we complete the proof using Eq.~(\ref{eq:exp_Fq}).
\end{proof}

\subsection{Thermal state approximation by a projected Hamiltonian}

\begin{lemma}[Thermal-state approximation by a projected Hamiltonian]\label{lem:thermal_approx_proj}
Let $H$ be a Hermitian operator on a finite-dimensional Hilbert space. Let 
    $\rho=\frac{e^{-\beta H}}{\tr(e^{-\beta H})}$ denote its thermal state at inverse temperature $\beta$ and $Z=\tr(e^{-\beta H})$ denote its partition function. Let $P$ be an orthogonal projector satisfying
\begin{align}
    \tr((\I-P)\rho)\leq\delta,
    \qquad 0<\delta<1.
\end{align}
Define the thermal state of $PHP$ restricted to the range of $P$ by 
    $\sigma_P
\coloneqq
\frac{P e^{-\beta PHP}P}
{\tr(P e^{-\beta PHP}P)}$ and set $Z_P=\tr(P e^{-\beta PHP}P)$.
Then  
\begin{align}
 \theta_P\coloneqq\lVert PH(\I-P)\rVert,\quad  |Z-Z_P| \leq 4(1+\beta\theta_P)\delta^{1/4}Z,\quad  \lVert\rho-\sigma_P\rVert_1
    \leq
    2\delta+
4\sqrt{\beta\theta_P}\,\delta^{1/4}. \label{eq:thermal_approx_bound}
\end{align}
\end{lemma}

\begin{proof}
Define
\begin{align}
    H_{\mathrm{bd}}
    \coloneqq
    PHP+(\I-P)H(\I-P),\qquad
    \rho_{\mathrm{bd}}
    \coloneqq
    \frac{e^{-\beta H_{\mathrm{bd}}}}
    {\tr(e^{-\beta H_{\mathrm{bd}}})}.
\end{align}

\paragraph{Estimate $\|\rho-\rho_{\mathrm{bd}}\|_1$.} First
note that
\begin{align}
    \tr(e^{-\beta H_{\mathrm{bd}}})
    \leq
    \tr(e^{-\beta H}). \label{eq:convex}
\end{align}
Indeed, let $\{\psi_j\}_j$ be an eigenbasis of $H_{\mathrm{bd}}$ such that
each $\psi_j$ lies in either $\operatorname{ran}P$ or
$\operatorname{ran}(\I-P)$. Since $H-H_{\mathrm{bd}}$ contains only the
off-diagonal blocks, we have
$\langle\psi_j,H\psi_j\rangle
=\langle\psi_j,H_{\mathrm{bd}}\psi_j\rangle$.
Writing $H=\sum_kE_k\op{k}{k}$ and
$p_{jk}\coloneqq|\langle k|\psi_j\rangle|^2$, Jensen's inequality gives
\begin{align}
    e^{-\beta\langle\psi_j,H\psi_j\rangle}
    =
    e^{-\beta\sum_kp_{jk}E_k}
    \leq
    \sum_kp_{jk}e^{-\beta E_k}
    =
    \langle\psi_j,e^{-\beta H}\psi_j\rangle.
\end{align}
Summing this inequality over $j$ proves Eq.~(\ref{eq:convex}).
Using Eq.~(\ref{eq:convex}), the relative entropy between $\rho$ and $\rho_{\mathrm{bd}}$ satisfies
\begin{align}
D(\rho\Vert\rho_{\mathrm{bd}})
=
\beta\tr\!\left(\rho(H_{\mathrm{bd}}-H)\right)
+
\log\frac{\tr(e^{-\beta H_{\mathrm{bd}}})}
{\tr(e^{-\beta H})}\leq
\beta\left|
\tr\!\left(\rho(H_{\mathrm{bd}}-H)\right)
\right|.\label{eq:relative_entropy}
\end{align}
Besides, since
$H-H_{\mathrm{bd}}
=
PH(\I-P)+(\I-P)HP,$
the trace-norm inequality and the Schatten Cauchy--Schwarz inequality give
\begin{align}
D(\rho\Vert\rho_{\mathrm{bd}})
&\leq
2\beta\theta_P\,
\lVert P\rho(\I-P)\rVert_1\notag\\
&\leq
2\beta\theta_P\,
\sqrt{\tr(P\rho)
\tr((\I-P)\rho)}\notag\\
&\leq
2\beta\theta_P\,\sqrt{\delta}.
\label{eq:relative-entropy-projection}
\end{align}
Then 
Pinsker's inequality therefore gives

\begin{align}
	\lVert\rho-\rho_{\mathrm{bd}}\rVert_1\leq
\sqrt{2D(\rho\Vert\rho_{\mathrm{bd}})}
\leq
2\sqrt{\beta\theta_P}\,\delta^{1/4}.\label{eq:pinsker}
\end{align}

\paragraph{Estimate $\|\rho_{\mathrm{bd}}-\sigma_P\|_1$.}
Since $H_{\mathrm{bd}}$ is block diagonal, we have
$P\rho_{\mathrm{bd}}P=\tr(P\rho_{\mathrm{bd}})\sigma_P$.
Therefore, compared with $\sigma_P$, the $P$ block of
$\rho_{\mathrm{bd}}$ has weight
$\tr((\I-P)\rho_{\mathrm{bd}})$ missing, while its $(\I-P)$
block has the same weight. Hence
\begin{align}
	\lVert\rho_{\mathrm{bd}}-\sigma_P\rVert_1
=2\tr((\I-P)\rho_{\mathrm{bd}}).\label{eq:rho_bd_sigma_p}
\end{align}

Moreover,
\begin{align}
\tr((\I-P)\rho_{\mathrm{bd}}) = \tr((\I-P)\rho)+ \tr((\I-P)(\rho_{\mathrm{bd}}-\rho))
\leq
\delta+\frac12\lVert\rho-\rho_{\mathrm{bd}}\rVert_1.\label{eq:IP_rho_bad}
\end{align}

Combining Eqs.~(\ref{eq:pinsker}), (\ref{eq:rho_bd_sigma_p}), and (\ref{eq:IP_rho_bad}), we prove the trace-distance bound.

\paragraph{Compare the partition function $Z -\widehat{Z}_P$.}
Denote 
 $Z_{\mathrm{bd}} \coloneqq \tr(e^{-\beta H_{\mathrm{bd}}}).$ Since $H_{\mathrm{bd}}$ is block diagonal,  $Z_P=Z_{\mathrm{bd}} \tr(P\rho_{\mathrm{bd}})$.   Hence, by Eqs.~(\ref{eq:pinsker}) and~(\ref{eq:IP_rho_bad}), we have \begin{align} 1-\frac{Z_P}{Z_{\mathrm{bd}}} \leq \delta+\sqrt{\beta\theta_P}\,\delta^{1/4}.\label{eq:one_Zp_Zbd} 
 \end{align}
Moreover, using the first equality in Eq.~(\ref{eq:relative_entropy}), the fact that $D(\rho\Vert\rho_{\mathrm{bd}})\geq0$, and $\left| \tr\!\left(\rho(H_{\mathrm{bd}}-H)\right) \right| \leq 2\theta_P\sqrt{\delta}$, we obtain \begin{align} \log\frac{Z_{\mathrm{bd}}}{Z} \geq -2\beta\theta_P\sqrt{\delta}, \end{align} and therefore $$\frac{Z_{\mathrm{bd}}}{Z} -1\geq -2\beta\theta_P\sqrt{\delta}.$$ Besides, note that $Z_P\leq Z_{\mathrm{bd}}\leq Z$, where the first inequality is direct and the second inequality comes from Eq.~(\ref{eq:convex}).
Combining Eqs.~(\ref{eq:one_Zp_Zbd}) and (\ref{eq:convex}) and using $Z_P\leq Z_{\mathrm{bd}}\leq Z$, we have
\begin{align}
0\leq 1-\frac{Z_P}{Z}
=
\left(1-\frac{Z_{\mathrm{bd}}}{Z}\right)
+
\frac{Z_{\mathrm{bd}}}{Z}
\left(1-\frac{Z_P}{Z_{\mathrm{bd}}}\right)\leq
\delta
+\sqrt{\beta\theta_P}\,\delta^{1/4}
+2\beta\theta_P\sqrt{\delta}.
\end{align}

Since $0<\delta<1$, we have
$\delta\leq\delta^{1/4}$ and $\sqrt{\delta}\leq\delta^{1/4}$, while
$\sqrt{\beta\theta_P}\leq 1+\beta\theta_P$. Therefore,
\begin{align}
\delta+\sqrt{\beta\theta_P}\,\delta^{1/4}
+2\beta\theta_P\sqrt{\delta}
\leq
4(1+\beta\theta_P)\delta^{1/4}.
\end{align}

Thus we prove the bound $ |Z-Z_P| \leq 4(1+\beta\theta_P)\delta^{1/4}Z$. 
\end{proof}

\section{Continuous-time quantum Monte Carlo}\label{sec:CT-QMC}

In this appendix, we give the detailed algorithm and proof of
Theorem~\ref{thm:ctqmc}. Recall that we consider 
a decomposition
\begin{align}
    H=H_{\mathrm{ref}}+V,
    \qquad
    V=\sum_{a=1}^R v_a U_a,
    \qquad
    \|U_a\|\leq 1,\qquad  W\coloneqq\sum_{a=1}^R |v_a|. \label{eq:H_CT_QMC}
\end{align}
where the coefficients $v_a\in\mathbb{R}$.

\subsection{Imaginary-time interaction-picture expansion}\label{sec:interaction_picture_expansion}

For $r\geq 1$, define the ordered simplex
\[
 \Omega_r(\beta)
 \coloneqq\{(\tau_1,\ldots,\tau_r): \,\beta\geq\tau_1\geq\cdots\geq\tau_r\geq0\}.
\]

For a real value $\tau$, define
$V_I(\tau)\coloneqq e^{\tau H_{\mathrm{ref}}}Ve^{-\tau H_{\mathrm{ref}}}$.
The imaginary-time interaction-picture expansion expands the operator
$e^{-\beta H}$ around $H_{\mathrm{ref}}$:
\begin{align}
	e^{-\beta H}
	&=
	\sum_{r=0}^{\infty}(-1)^r e^{-\beta H_{\mathrm{ref}}}
	\int_{\Omega_r(\beta)}
	V_I(\tau_1)\cdots V_I(\tau_r)
	\,d\tau_1\cdots d\tau_r
	\label{eq:interaction-form}\\
	&=
	\sum_{r=0}^{\infty}(-1)^r
	\int_{\Omega_r(\beta)}
	e^{-(\beta-\tau_1)H_{\mathrm{ref}}}V
	e^{-(\tau_1-\tau_2)H_{\mathrm{ref}}}\cdots
	V e^{-\tau_rH_{\mathrm{ref}}}
	\,d\tau_1\cdots d\tau_r.
	\label{eq:duhamel-form}
\end{align}

For convenience, set $D_0\coloneqq e^{-\beta H_{\mathrm{ref}}}$ and, for $r\geq1$, define
\[
D_r
\coloneqq(-1)^re^{-\beta H_{\mathrm{ref}}}
\int_{\Omega_r(\beta)}
V_I(\tau_1)\cdots V_I(\tau_r)\,d\tau_1\cdots d\tau_r,
\qquad
G_K\coloneqq\sum_{r=0}^K D_r.
\]

\begin{theorem}[Truncating the interaction-picture expansion]
\label{thm:truncation}
The expansion
$e^{-\beta H}=\sum_{r=0}^{\infty}D_r$
converges absolutely in trace norm. For any Hermitian observable $O$ with $\|O\|\leq 1$,
abbreviate
\[
 Z_K\coloneqq \operatorname{Re}\tr(G_K),
 \qquad
 N_K\coloneqq \operatorname{Re}\tr(OG_K),
 \qquad
 \mu_K\coloneqq N_K/Z_K, \qquad Z\coloneqq\tr(e^{-\beta H}).
\]
Then for any precision parameter $\epsilon\in(0,1)$, set
$K=\cO(\beta W+\log(1/\epsilon))$. Then we have
\begin{equation}\label{eq:general-truncated-consequences}
 Z_K\geq\frac{15}{16}Z>0,
 \qquad
 |Z_K-Z|\leq\frac{\epsilon}{16}Z,
 \qquad
 \left|\mu_K-\tr\left(O\rho^{(\beta)}\right)\right|
 \leq\frac{2\epsilon}{15},
 \qquad
 |\mu_K|<\frac65,
\end{equation}
\end{theorem}

\begin{proof}
Let
\[
 U(t)\coloneqq e^{tH_{\mathrm{ref}}}e^{-tH}.
\]
Then $U'(t)=-V_I(t)U(t)$ and $U(0)=\I$. Iterating the corresponding
integral equation and multiplying by $e^{-\beta H_{\mathrm{ref}}}$ gives
Eq.~\eqref{eq:interaction-form}, hence also Eq.~\eqref{eq:duhamel-form}.

Note that $\|V\|\leq W$.
Fix ordered times and let $s_0,\ldots,s_r\geq0$ be the successive time gaps,
so that $\sum_{j=0}^rs_j=\beta$. Generalized Schatten H\"older gives
\begin{align*}
 \norm{e^{-s_0H_{\mathrm{ref}}}V
 e^{-s_1H_{\mathrm{ref}}}\cdots
 V e^{-s_rH_{\mathrm{ref}}}}_1
 &\leq
 W^r\prod_{j:s_j>0}
 \norm{e^{-s_jH_{\mathrm{ref}}}}_{\beta/s_j}=W^rZ_{\mathrm{ref}},
\end{align*}
because
$\norm{e^{-s_jH_{\mathrm{ref}}}}_{\beta/s_j}
=Z_{\mathrm{ref}}^{s_j/\beta}$.
Since the volume of the ordered simplex is
$\operatorname{Vol}(\Omega_r(\beta))=\beta^r/r!$, we have
 $\norm{D_r}_1\leq
 Z_{\mathrm{ref}}\frac{(\beta W)^r}{r!}.$
Thus the expansion converges absolutely in trace norm, and
\begin{equation}\label{eq:general-tail}
 \norm{e^{-\beta H}-G_K}_1
 \leq Z_{\mathrm{ref}}
 \sum_{r=K+1}^{\infty}\frac{(\beta W)^r}{r!}.
\end{equation}

Note that
$H_{\mathrm{ref}}-W\I\leq H\leq H_{\mathrm{ref}}+W\I$.
Therefore
 $e^{-\beta W}Z_{\mathrm{ref}}
 \leq Z\leq e^{\beta W}Z_{\mathrm{ref}}.$
Choose
 $K\geq4e\beta W+\log_2\frac{16}{\epsilon}.$
Then
\[
 \eta\coloneqq
 \frac{\norm{e^{-\beta H}-G_K}_1}{Z}
 \leq e^{\beta W}2^{-K}
 \leq\frac{\epsilon}{16}.
\]
Abbreviate
$N\coloneqq\tr(Oe^{-\beta H})$ and $\mu\coloneqq N/Z$. Then $N,\mu\in\mathbb R$ and
$|\mu|\leq1$ since $O$ is Hermitian and $\|O\|\leq1$. Using
$|\operatorname{Re}\tr(A)|\leq\norm{A}_1$, the above equation implies
\[
 |Z_K-Z|\leq\eta Z,
 \qquad
 |N_K-N|\leq\eta Z.
\]
Hence
$Z_K\geq(1-\eta)Z\geq15Z/16$,
$|Z_K-Z|\leq\epsilon Z/16$, and
 $|\mu_K-\mu|
 \leq\frac{2\eta}{1-\eta}
 \leq\frac{2\epsilon}{15}.$
Similarly,
$|\mu_K|
 \leq\frac{1+\eta}{1-\eta}
 \leq\frac{17}{15}<\frac65.$
\end{proof}

\subsection{Continuous-time quantum Monte Carlo}

Recall that the perturbation term $V$ has the decomposition
\begin{equation}\label{eq:abstract-unitary-decomposition}
 V=\sum_{a=1}^{R}v_aU_a,
 \qquad
 \|U_a\|\leq 1,
 \qquad
 W\coloneqq\sum_{a=1}^{R}|v_a|.
\end{equation}
Then the truncated expansion is
\begin{align}
 G_K
 &=
 e^{-\beta H_{\mathrm{ref}}}
 +\sum_{r=1}^K\sum_{a_1,\ldots,a_r=1}^R
 (-1)^r\left(\prod_{j=1}^rv_{a_j}\right)
 \int_{\Omega_r(\beta)}
 e^{-(\beta-\tau_1)H_{\mathrm{ref}}}U_{a_1}
 e^{-(\tau_1-\tau_2)H_{\mathrm{ref}}}\cdots
 U_{a_r}e^{-\tau_rH_{\mathrm{ref}}}
 \,d\tau_1\cdots d\tau_r.
 \label{eq:Gk}
\end{align}
Theorem~\ref{thm:truncation} reduces the estimation of
$\tr(O\rho^{(\beta)})$
to estimating
\[
\mu_K=\frac{N_K}{Z_K}
=\frac{\operatorname{Re}\tr(OG_K)}
{\operatorname{Re}\tr(G_K)}.
\]
The continuous-time quantum Monte Carlo method (CT-QMC) estimates its
numerator and denominator by sampling configurations in
Eq.~\eqref{eq:Gk}, specified by
\[
 r,\ a_1,\ldots,a_r,\ \tau_1,\ldots,\tau_r.
\]

The decomposition $H=H_{\mathrm{ref}}+V$ and the decomposition
$V=\sum_{a=1}^R v_aU_a$ are chosen such that it is easy to calculate the
configuration values
\begin{align}
\mathcal C_O(r,\boldsymbol a,\boldsymbol\tau)
&=
\tr\left(
Oe^{-(\beta-\tau_1)H_{\mathrm{ref}}}U_{a_1}
e^{-(\tau_1-\tau_2)H_{\mathrm{ref}}}\cdots
U_{a_r}e^{-\tau_rH_{\mathrm{ref}}}
\right),\\
\mathcal C_\I(r,\boldsymbol a,\boldsymbol\tau)
&=
\tr\left(
e^{-(\beta-\tau_1)H_{\mathrm{ref}}}U_{a_1}
e^{-(\tau_1-\tau_2)H_{\mathrm{ref}}}\cdots
U_{a_r}e^{-\tau_rH_{\mathrm{ref}}}
\right).
\end{align}
We use $T_{\mathrm{cfg}}(K)$ to denote an upper bound on the time needed
to evaluate these configuration values for $r\leq K$.
Set
\begin{equation}\label{eq:abstract-parameters}
 K=\cO(\beta W+\log(1/\epsilon)),
 \qquad
 C_K\coloneqq\sum_{r=0}^{K}\frac{(\beta W)^r}{r!},
 \qquad
 d\coloneqq\frac{15}{16}e^{-2\beta W},
 \qquad
 M\coloneqq\left\lceil
 \frac{128}{\epsilon^2d^2}\log\frac4\delta
 \right\rceil.
\end{equation}

The CT-QMC algorithm is as follows:
If $W=0$, then $H=H_{\mathrm{ref}}$ and the reference expectation and
partition function are computed directly. Otherwise, generate one sample as
follows.

\begin{enumerate}
\item Sample $r$ from a Poisson distribution with mean $\beta W$, conditioned
on $r\leq K$. Equivalently,
\[
 \Pr[r]=\frac{(\beta W)^r/r!}{C_K},
 \qquad r=0,\ldots,K.
\]

\item Sample $r$ independent values in $[0,\beta]$ and sort them as
$\beta\geq\tau_1\geq\cdots\geq\tau_r\geq0$.

\item Sample $a_1,\ldots,a_r$ independently with
$\Pr[a_j=a]=|v_a|/W$.

\item For $r=0$, set $w=1$ and
$\cG=e^{-\beta H_{\mathrm{ref}}}$. For $r\geq1$, use
$\sgn(x)\in\{1,-1\}$ as the sign function for $x$, and define
\begin{align}
     w&\coloneqq(-1)^r\prod_{j=1}^r\sgn(v_{a_j}),\\
     \cG&\coloneqq
 e^{-(\beta-\tau_1)H_{\mathrm{ref}}}U_{a_1}
 e^{-(\tau_1-\tau_2)H_{\mathrm{ref}}}\cdots
 U_{a_r}e^{-\tau_rH_{\mathrm{ref}}}.
\end{align}

\item Compute
\begin{equation}\label{eq:abstract-samples}
 X\coloneqq w\frac{\operatorname{Re}\tr(O\cG)}{Z_{\mathrm{ref}}},
 \qquad
 Y\coloneqq w\frac{\operatorname{Re}\tr(\cG)}{Z_{\mathrm{ref}}},
\end{equation}
where
$Z_{\mathrm{ref}}=\tr(e^{-\beta H_{\mathrm{ref}}})$.
\end{enumerate}

For $M$ independent samples, let $\overline X$ and $\overline Y$ be the
sample means. Output
\begin{equation}\label{eq:abstract-output}
 \widehat\mu\coloneqq
 \begin{cases}
 0, & \overline Y\leq d/2,\\[1mm]
 \overline X/\overline Y, & \overline Y>d/2,
 \end{cases}
 \qquad
 \widehat Z\coloneqq C_KZ_{\mathrm{ref}}\overline Y.
\end{equation}
Here $d=\frac{15}{16}e^{-2\beta W}$ is a threshold ensuring that the
denominator is nonzero, so that the ratio $\overline X/\overline Y$ is
well-defined.

\begin{proof}[Proof of Theorem~\ref{thm:ctqmc}]
Note that the density of sampling a particular
$\beta\geq\tau_1\geq\cdots\geq\tau_r\geq0$
is $\frac{r!}{\beta^r}$.
One can check that the sampling procedure gives
\begin{equation}\label{eq:abstract-sampling-identity}
 \E[w\cG]=\frac{G_K}{C_K}.
\end{equation}
Because $\|U_a\|\leq1$, similarly to the H\"older estimate in the proof of
Theorem~\ref{thm:truncation}, we have
$\norm{\cG}_1\leq Z_{\mathrm{ref}}$.
Therefore $X$ and $Y$ are real random variables satisfying
$|X|\leq1$ and $|Y|\leq1$.

Let $x\coloneqq\E[X]$ and $y\coloneqq\E[Y]$. Recall the notation
$N_K,Z_K,\mu_K$ from Theorem~\ref{thm:truncation}.
By Eq.~\eqref{eq:abstract-sampling-identity},
\[
 x=\frac{N_K}{C_KZ_{\mathrm{ref}}},
 \qquad
 y=\frac{Z_K}{C_KZ_{\mathrm{ref}}},
 \qquad
 \frac{x}{y}=\mu_K.
\]
Eq.~\eqref{eq:general-truncated-consequences} gives
$Z_K\geq15Z/16$. Since
$C_K\leq e^{\beta W}$ and
$Z\geq e^{-\beta W}Z_{\mathrm{ref}}$, we have
\[
 y\geq\frac{15}{16}e^{-2\beta W}=d.
\]

Recall that $M$ is the number of independent samples. Set 
$$t\coloneqq\epsilon d/8.$$ Hoeffding's inequality and
the fact that $|X|\leq1$, $|Y|\leq1$ imply
\[
 \Pr\!\left[
 |\overline X-x|>t\ \text{or}\ |\overline Y-y|>t
 \right]
 \leq4e^{-Mt^2/2}\leq\delta.
\]
On the complementary event,
$\overline Y\geq y-t\geq7y/8>d/2$, so the algorithm outputs
$\overline X/\overline Y$ according to Eq.~(\ref{eq:abstract-output}).
Moreover,
Eq.~\eqref{eq:general-truncated-consequences} also gives
$|\mu_K|<6/5$, and hence
\[
 \left|
 \frac{\overline X}{\overline Y}-\mu_K
 \right|
 =
 \left|
 \frac{\overline X}{\overline Y}-\frac{x}{y}
 \right|
 =
 \left|
 \frac{\overline X-x+\mu_K(y-\overline Y)}
 {\overline Y}
 \right|
 \leq
 \frac{t+|\mu_K|t}{7y/8}
 \leq\frac{11}{35}\epsilon.
\]
Combining this with the bias bound in
Eq.~\eqref{eq:general-truncated-consequences} gives total error less
than $\epsilon$.
It remains to estimate the partition function. By
Eq.~\eqref{eq:abstract-output},
\begin{align}
 |\widehat Z-Z|
 &\leq
 C_KZ_{\mathrm{ref}}|\overline Y-y|
 +|Z_K-Z|.
\end{align}
On the same event, using
$C_K\leq e^{\beta W}$,
$Z_{\mathrm{ref}}\leq e^{\beta W}Z$,
$t=\epsilon d/8$, and
$d=\frac{15}{16}e^{-2\beta W}$, we obtain
\begin{align}
 C_KZ_{\mathrm{ref}}|\overline Y-y|
 &\leq
 e^{2\beta W}Z
 \frac{\epsilon d}{8}
 =
 \frac{15\epsilon}{128}Z.
\end{align}
Together with
$|Z_K-Z|\leq\epsilon Z/16$ from
Theorem~\ref{thm:truncation}, this gives
\begin{align}
 |\widehat Z-Z|
 \leq
 \left(
 \frac{15}{128}+\frac{1}{16}
 \right)\epsilon Z
 <\epsilon Z.
\end{align}

Finally, from Eq.~\eqref{eq:abstract-parameters},
\begin{align}
 M
 =
 e^{\cO(\beta W)}
 \operatorname{poly}\!\left(
 \epsilon^{-1},\log\delta^{-1}
 \right).
\end{align}
Since every sampled configuration has order $r\leq K$, one sample can be
evaluated in time $T_{\mathrm{cfg}}(K)$, up to polynomial overhead in $K$.
Therefore, the total runtime is
\begin{align}
 e^{\cO(\beta W)}
 \operatorname{poly}\!\left(
 K,\epsilon^{-1},\log\delta^{-1}
 \right)
 T_{\mathrm{cfg}}(K),
\end{align}
which proves Theorem~\ref{thm:ctqmc}.
\end{proof}

\section{Gibbs-state preparation via quantum belief propagation}\label{sec:QBP}

In this section we show how quantum belief propagation (QBP) can be used to design a quantum algorithm for preparing Gibbs states of Hamiltonians of the form $H = H_0 + V$. We use the form of the statement recorded in~\cite{kato2019quantum}.

\begin{proposition}
    Let $H_s = H_0 + sV$ for $s \in [0, 1]$. For $\beta > 0$, define
    \begin{equation}\label{eq:Phi_s_defn}
        \Phi_s \coloneqq \int_\R dt \, f_\beta(t) e^{-itH_s} V e^{itH_s}, \quad \text{where } f_\beta(t) \coloneqq \frac{2}{\pi\beta} \log\coth\l(\frac{\pi|t|}{2\beta}\r).
    \end{equation}
    The quantum belief propagation (QBP) operator
    \begin{equation}\label{eq:QBP_op}
        \eta_s \coloneqq \TO\exp\l( -\frac{\beta}{2} \int_0^s dr \, \Phi_r \r).
    \end{equation}
    satisfies the identity
    \begin{equation}
        \eta_s e^{-\beta H_0} \eta_s^\dagger = e^{-\beta H_s}.
    \end{equation}
\end{proposition}

\begin{corollary}\label{cor:QBP_density_matrix}
    Fix $\beta \geq 0$. Let $\rho_s = e^{-\beta H_s} / Z_s$ where $Z_s = \tr(e^{-\beta H_s})$. Then
    \begin{equation}
        \eta_1 \rho_0 \eta_1^\dagger = \frac{Z_1}{Z_0} \rho_1.
    \end{equation}
\end{corollary}

When implementing imaginary-time exponentials, it will be useful to assume that the term $V$ is PSD. If this is not the case, we can always achieve it by an identity shift of $+\|V\|$.

Our algorithm will use block encodings extensively.

\begin{definition}
    Let $A$ be a $q$-qubit operator and $\varepsilon \geq 0$. A $(q+a)$-qubit unitary $\BE[A; \alpha, a, \varepsilon]$ is called an $(\alpha, a, \varepsilon)$-block-encoding of $A$ if it holds that
    \begin{equation}
        \|A - \alpha (\bra{0^a} \otimes \I) \BE[A; \alpha, a, \varepsilon] (\ket{0^a} \otimes \I)\| \leq \varepsilon.
    \end{equation}
\end{definition}

The following lemma shows how to discretize the integral defining $\Phi_s$ with exponentially convergent accuracy.

\begin{lemma}\label{lem:discrete_Phi_s}
    Let $0 < \varepsilon \leq 2\|V\|$ and $\Phi_s$ be as defined in \cref{eq:Phi_s_defn}. Further assume that $V \succeq 0$. There exists a discretized approximation
    \begin{equation}\label{eq:approx_Phi_s}
        \widetilde{\Phi}_s \coloneqq \sum_{k=1}^N w_k e^{-it_k H_s} V e^{it_k H_s}, \quad \text{where } \sum_{k=1}^N w_k \leq 1 \text{ and each } w_k \geq 0,
    \end{equation}
    such that $\|\widetilde{\Phi}_s - \Phi_s\| \leq \varepsilon$ with the choice
    \begin{equation}\label{eq:discrete_Phi_s_num_nodes}
        N = \l\lceil \max\l\{ \frac{2e\beta\|H_s\|}{\pi} \log\l(\frac{4\|V\|}{\varepsilon}\r), \frac{1}{2} \log_2\l(\frac{5\|V\|}{\varepsilon}\r) \r\} \r\rceil.
    \end{equation}
\end{lemma}

\begin{proof}
    First we truncate the tails of the integral. Let $q = e^{-\pi t/\beta}$ for some $t > 0$. From the Taylor series expansion
    \begin{equation}
        \log\coth x = 2 \sum_{m=0}^\infty \frac{e^{-2(2m+1)x}}{2m+1}
    \end{equation}
    we have
    \begin{equation}
        f_\beta(t) = \frac{4}{\pi\beta} \sum_{m=0}^\infty \frac{q^{2m+1}}{2m+1}.
    \end{equation}
    Observe that for all $q \leq \frac{1}{2}$,
    \begin{equation}
        f_\beta(t) \leq \frac{4}{\pi\beta} \sum_{m=0}^\infty q^{2m+1} = \frac{4}{\pi\beta} \frac{q}{1 - q^2} \leq \frac{16}{3\pi\beta} e^{-\pi t/\beta}.
    \end{equation}
    Introduce some cutoff $T$ satisfying $T \geq \frac{\beta}{\pi} \log 2$. Then
    \begin{equation}
        \int_{|t| > T} dt \, f_\beta(t) \leq \frac{32}{3\pi^2} e^{-\pi T/\beta} < 2 e^{-\pi T/\beta}.
    \end{equation}
    For notation, define $F_s(t) \coloneqq e^{-itH_s} V e^{itH_s}$ obeying $\|F_s(t)\| = \|V\|$. Hence
    \begin{equation}
        \l\| \int_{|t| > T} dt \, f_\beta(t) F_s(t) \r\| \leq 2 \|V\| e^{-\pi T/\beta}.
    \end{equation}
    Choosing
    \begin{equation}
        T = \frac{\beta}{\pi} \log\l( \frac{4\|V\|}{\varepsilon} \r)
    \end{equation}
    bounds the truncation error by at most $\frac{\varepsilon}{2}$. Note that $T \geq \frac{\beta}{\pi} \log 2$ requires the very mild condition $\varepsilon \leq 2\|V\|$.

    Next, we show how to approximate
    \begin{equation}
        \Phi_s^{(T)} \coloneqq \int_{-T}^T dt \, f_\beta(t) F_s(t)
    \end{equation}
    to error at most $\frac{\varepsilon}{2}$. We accomplish this via a Gaussian quadrature with respect to $f_\beta(t)$. Define the measure
    \begin{equation}
        d\mu_T(t) \coloneqq \chi_{[-T,T]}(t) f_\beta(t) \, dt
    \end{equation}
    where $\chi_A : X \to \{0, 1\}$ is the indicator function on $A \subseteq X$. Note that the singularity at $t = 0$ is integrable, so this measure is legitimate. Because $f_\beta(t)$ integrates over $\R$ to unity, we have
    \begin{equation}
        \int_{-T}^T dt \, f_\beta(t) \leq \int_{-\infty}^\infty dt \, f_\beta(t) = 1.
    \end{equation}
    Now take the $N$-point Gaussian quadrature rule for $\mu_T$~\cite{gautschi2004orthogonal}:~there exist nodes $t_1, \ldots, t_N \in (-T, T)$ and weights $w_1, \ldots, w_N \in (0, 1)$ such that for every polynomial $p \in \C[t]$ of degree at most $2N-1$,
    \begin{equation}
        \int_{-T}^T d\mu_T(t) \, p(t) = \sum_{k=1}^N w_k p(t_k).
    \end{equation}
    An immediate consequence is that for $p = 1$, we have $\sum_{k=1}^N w_k \leq 1$ as desired. These nodes and weights also define $\widetilde{\Phi}_s$ as in \cref{eq:approx_Phi_s}.

    To analyze the quadrature error, we use the adjoint representation for $F_s(t)$. Define
    \begin{equation}
        \ad_X(Y) \coloneqq [X, Y]
    \end{equation}
    which obeys the key identity
    \begin{equation}
        e^{\ad_X}(Y) = e^X Y e^{-X}.
    \end{equation}
    In our context, we take $X = -itH_s$ so that $F_s(t) = e^{-it \ad_{H_s}}(V)$. We therefore approximate $F_s(t)$ by the degree-$(2N-1)$ polynomial
    \begin{equation}
        \widetilde{F}_s(t) \coloneqq \sum_{k=0}^{2N-1} \frac{(-it)^k}{k!} \ad_{H_s}^k(V).
    \end{equation}
    For $|t| \leq T$, the approximation error is
    \begin{equation}
        \|F_s(t) - \widetilde{F}_s(t)\| \leq \|V\| \sum_{k=2N}^\infty \frac{(2\|H_s\|T)^k}{k!},
    \end{equation}
    where we used the estimate
    \begin{equation}
        \|{\ad_{H_s}^k(V)}\| \leq (2\|H_s\|)^k \|V\|.
    \end{equation}
    Condense notation by putting $a \coloneqq 2\|H_s\|T$. Suppose that $N \geq ea$ and use the Stirling bound $m! \geq (m/e)^m$ to estimate the first term of the series:
    \begin{equation}\label{eq:first_term_series_bound}
        \frac{a^{2N}}{(2N)!} \leq \l(\frac{ea}{2N}\r{)^{2N}} \leq 2^{-2N}.
    \end{equation}
    Furthermore, for every $k \geq 2N$, we have
    \begin{equation}\label{eq:ratio_terms_series_bound}
        \frac{a^{k+1}/(k+1)!}{a^k/k!} = \frac{a}{k+1} \leq \frac{a}{2N} \leq \frac{1}{2e}.
    \end{equation}
    Thus
    \begin{equation}
        \sum_{k=2N}^\infty \frac{a^k}{k!} \leq \frac{a^{2N}}{(2N)!} \sum_{j=0}^\infty \l(\frac{1}{2e} \r{)^j} \leq \frac{2e}{2e - 1} 2^{-2N},
    \end{equation}
    implying that
    \begin{equation}
        \|F_s(t) - \widetilde{F}_s(t)\| \leq \l( \frac{2e}{2e - 1} \r) \|V\| 2^{-2N}
    \end{equation}
    for any $t \in [-T, T]$.

    We obtain the final error bound by applying the triangle inequality. Since the Gaussian quadrature guarantees that the discrete sum and continuous integral coincide for degree-$(2N-1)$ polynomials, we have
    \begin{equation}
        \sum_{k=1}^N w_k \widetilde{F}_s(t_k) = \int_{-T}^T d\mu_T(t) \, \widetilde{F}_s(t).
    \end{equation}
    Hence
    \begin{equation}
    \begin{split}
        \l\| \widetilde{\Phi}_s - \Phi_s^{(T)} \r\| &= \l\| \sum_{k=1}^N w_k F_s(t_k) - \int_{-T}^T d\mu_T(t) \, F_s(t) \r\|\\
        &\leq \l\| \sum_{k=1}^N w_k \l[ F_s(t_k) - \widetilde{F}_s(t_k) \r] \r\| + \l\| \int_{-T}^T d\mu_T(t) \l[ \widetilde{F}_s(t) - F_s(t) \r] \r\|\\
        &\leq 2 \sup_{|t| \leq T} \l\| \widetilde{F}_s(t) - F_s(t) \r\|\\
        &\leq \frac{5}{2} \|V\| 2^{-2N}.
    \end{split}
    \end{equation}
    Above, we used $\frac{4e}{2e-1} < \frac{5}{2}$ to simplify the constant. To bound this by at most $\frac{\varepsilon}{2}$, it suffices to take
    \begin{equation}
        N \geq \frac{1}{2} \log_2\l( \frac{5\|V\|}{\varepsilon} \r).
    \end{equation}
    Note however that we also have to guarantee $N \geq ea$ for \cref{eq:first_term_series_bound,eq:ratio_terms_series_bound} to hold. Thus with $a = 2\|H_s\|T$ where $T = \frac{\beta}{\pi} \log\frac{4\|V\|}{\varepsilon}$, the choice of $N$ as in \cref{eq:discrete_Phi_s_num_nodes} suffices to get
    \begin{equation*}
        \l\|\Phi_s - \widetilde{\Phi}_s\r\| \leq \l\| \widetilde{\Phi}_s - \Phi_s^{(T)} \r\| + \l\| \Phi_s^{(T)} - \Phi_s \r\| \leq \varepsilon. \qedhere
    \end{equation*}
\end{proof}

\begin{remark}
    The quadrature nodes and weights $\{(t_k, w_k)\}_{k=1}^N$ for the measure $d\mu_T(t)$ can be computed in $\poly(N)$ arithmetic operations, using the closed form for $f_\beta(t)$. For example, see the classic construction due to Golub and Welsch~\cite{golub1969calculation}.
\end{remark}

We then construct an approximate block encoding of $\Phi_s$ as a linear combination of unitaries (LCU).

\begin{corollary}\label{lem:LCU_Phi_s}
    Assume access to block encodings $B_s(t) = \BE[e^{-itH_s} V e^{itH_s}; \lambda, a, \gamma]$ for some normalization $\lambda \geq \|V\|$ and error $\gamma \geq 0$. Let $0 < \varepsilon \leq 2\|V\|$ and $\beta > 0$. There exists a quantum circuit for
    \[
    \BE[\Phi_s; \lambda, a + \lceil \log_2 N \rceil, \varepsilon + \gamma], \qquad \Phi_s = \int_\R dt \, f_\beta(t) e^{-itH_s} V e^{itH_s},
    \]
    built from $N$ queries to controlled $B_s(t)$ for various $t \in [-T, T]$, where $T = \frac{\beta}{\pi} \log\frac{4\|V\|}{\varepsilon}$ and
    \begin{equation}
        N = \O\l( (1 + \beta H_*) \log\frac{\|V\|}{\varepsilon} \r), 
        \quad H_* \coloneqq \max_{r \in [0, s]} \|H_r\|.
    \end{equation}
\end{corollary}

\begin{proof}
    The construction is a linear combination of $N$ block encodings $B_s(t_k)$ with weights $w_k \geq 0$ that sum to at most $1$. This is a $(\lambda, a + \lceil \log_2 N \rceil, \gamma)$-block-encoding of $\widetilde{\Phi}_s$~\cite{gilyen2019quantum}. The final claimed error follows from a triangle inequality, using the fact that $\|\widetilde{\Phi}_s - \Phi_s\| \leq \varepsilon$ with the appropriate choice of $N$. Note that we have chosen the maximum norm $H_*$ over the interval so that all $B_s(t_k)$ can share the same Gaussian quadrature nodes and weights.
\end{proof}

\begin{remark}\label{rem:B(t)-gates}
    The block encoding of $e^{-itH_s} V e^{itH_s}$ is constructed by standard Hamiltonian simulation techniques. For example, given access to $\BE[H; \alpha_H, a_H, 0]$ we can construct $\BE[e^{itH}; 1, a_H + 2, \gamma']$ using $\O(\alpha_H |t| + \log\frac{1}{\gamma'})$ queries~\cite{gilyen2019quantum}. Thus if $V$ is accessed by the block encoding $\BE[V; \lambda, a_V, 0]$, we have $a = a_V + a_H + 2$ and $\gamma = 2 \lambda \gamma'$ in \cref{lem:LCU_Phi_s}.
\end{remark}

Equipped with efficient block encodings of $\Phi_s$, we can use the linear combination of Hamiltonian simulation (LCHS) method~\cite{an2023linear} to block encode \cref{eq:QBP_op}. We specialize to the case when the generator is purely Hermitian and cite the optimal query complexity due to Low and Somma~\cite{low2025optimal}.

\begin{proposition}[{\cite[Theorem 4, Hermitian generator]{low2025optimal}}]\label{prop:optimal_lchs}
    Fix $s > 0$ and let $A(r)$ be a PSD matrix for all $r \in [0, s]$. Define
    \begin{equation}
        U_s \coloneqq \TO\exp\l( -\int_0^s A(r) \, dr \r).
    \end{equation}
    Suppose we have access to quantum circuits for $\BE[A(r); \lambda, q, 0]$ for any $r$. For any $\varepsilon \in (0, 4/5]$, there exists a block encoding $\BE[U_s; \alpha, a, \varepsilon]$ with
    \begin{equation}
        \alpha = \Theta(1) \quad \text{and} \quad a = q + \O( \log(\|A\|_{L^1}) + \log \log \varepsilon^{-1} ),
    \end{equation}
    where $\|A\|_{L^1} \coloneqq \int_0^s dr \, \|A(r)\|$. This block encoding is constructed from $Q$ queries to $\BE[A(r); \lambda, q, 0]$ and $\O(G + Q \log(\|A\|_{L^1}) \polylog \varepsilon^{-1})$ quantum gates, where $Q$ and $G$ are the query and gate complexities, respectively, of simulating $\TO\exp\l(-i R \int_0^s A(r) \, dr\r)$ (to error $\Theta(\varepsilon)$) for scalars $R \in \R$ bounded as $|R| = \O(\log \varepsilon^{-1})$.
\end{proposition}

In our case, $A(r) = \frac{\beta}{2} \Phi_r$ whose $L^1$-norm on $[0, s]$ can be easily bounded as $\|A\|_{L^1} \leq \frac{\beta s}{2} \|V\|$. Thus the LCHS framework is highly efficient, with only the cost of time-dependent Hamiltonian simulation~\cite{low2018hamiltonian} dominating the entire algorithm. The complexity of this subroutine was recently made optimal by Chen, Gao, Wang, and Zhou using a transduced LCU approach~\cite{chen2026time}. First, we need to define the oracle $\HAMT$, which is at the heart of time-dependent Hamiltonian simulation algorithms.

\begin{definition}
    Let $T > 0$ and $J$ be a positive integer. For a time-dependent Hamiltonian $A(t)$ defined over $t \in [0, T]$, the $\HAMT$ oracle is defined as
    \begin{equation}\label{eq:HAMT_defn}
        \HAMT \coloneqq \sum_{j=0}^{J-1} \op{j}{j} \otimes \BE[A(t_j); \alpha, a, 0]
    \end{equation}
    where $t_j \coloneqq jT/J$, and each block encoding of $A(t_j)$ is promised to be Hermitian and unitary and to have parameters $\alpha, a$ uniformly over all times.
\end{definition}

\begin{proposition}[{\cite[Theorem 9]{chen2026time}}]\label{prop:dyson_series_complexity}
    Let $A(t)$ be a time-dependent Hamiltonian and $T > 0$ an evolution time. Suppose there exist parameters $\alpha, L > 0$ such that
    \begin{equation}
        \|A(t)\| \leq \alpha \quad \text{and} \quad \|A(t) - A(s)\| \leq L |t - s| \quad \text{for all } s, t \in [0, T].
    \end{equation}
    For any $\varepsilon \in (0, \frac{1}{2}]$, there exists a circuit implementing
    \[
        \BE\l[ \TO\exp\l( -i \int_0^T A(t) \, dt \r); 1, \O(a + \log J), \varepsilon \r]
    \]
    which uses
    \begin{equation}
        Q = \O\l( \alpha T + \frac{\log(1/\varepsilon)}{\log\l(e + \frac{1}{\alpha T} \log(1/\varepsilon) \r)} \r)
    \end{equation}
    queries to $\HAMT$. The choice of $J$ is the smallest power of two satisfying
    \begin{equation}\label{eq:TD_J_requirement}
        J \geq \max\l\{ Q, \frac{L T^2}{\varepsilon}, \sqrt{\frac{(\alpha T)^3}{3\varepsilon}} \r\}.
    \end{equation}
\end{proposition}

\begin{remark}\label{rem:TDHS_gates}
    The number of additional one- and two-qubit gates (alongside the queries to $\HAMT$) can be made to be $\Ot(Qa)$ using a refined version of the algorithm~\cite{chen2026gate}.
\end{remark}

We therefore need to show how to construct $\HAMT$ for our QBP generator $A(r) \to \frac{R\beta}{2} \Phi_r$. Recall that technically we only have an approximation of $\Phi_r$;~therefore we will first show how to implement the exponential of the approximation with controlled error. The error to the exponential of the exact $\Phi_r$ can then be handled straightforwardly. Below, we use $s$ instead of $T$ for the evolution time, because our ultimate object to implement is $\eta_s = \TO\exp\l(-\frac{\beta}{2} \int_0^s dr \, \Phi_r \r)$.

\begin{lemma}\label{lem:TD_Phi}
    Let $\theta > 0$, $s \in (0, 1]$, and $J$ be an integer power of two. Let $\varepsilon, \gamma, \delta > 0$ be error parameters such that
    \begin{equation}
        \varepsilon + \gamma \leq \frac{\delta}{4\theta s}.
    \end{equation}
    Let $W(r) = \BE[\Phi_r; \lambda, a + \lceil \log_2 N \rceil, \varepsilon + \gamma]$ be the block encoding from \cref{lem:LCU_Phi_s} with inverse temperature $\beta > 0$. There exists a circuit for
    \[
    \BE\l[ \TO\exp\l( -i \theta \int_0^s dr \, \Phi_r \r); 1, \O(a + \log(NJ)), 2\delta \r],
    \]
    constructed using $\O(QJ)$ queries to controlled $W(r)$ and $W(r)^\dagger$ each, provided that
    \begin{equation}
        Q = \O\l(|\theta| \lambda s + \log\frac{1}{\delta}\r) \quad \text{and} \quad J = \O\l( \frac{|\theta|\beta(\|V\| s)^2}{\delta} + \frac{(|\theta|\lambda s)^{3/2}}{\sqrt{\delta}} \r).
    \end{equation}
\end{lemma}

\begin{proof}
    For the moment, let us regard $W(r)$ as an exact block encoding of some $B_r$ and define $\widetilde{\Phi}_r' \coloneqq \frac{B_r + B_r^\dagger}{2}$, which is $(\varepsilon + \gamma)$-close to $\Phi_r$ in operator norm. The $\HAMT$ oracle requires Hermitian unitaries, so dilate
    \begin{equation}
        W_h(r) \coloneqq \begin{pmatrix}
            0 & W(r)\\
            W(r)^\dagger & 0
        \end{pmatrix},
    \end{equation}
    which is an exact block encoding of $\widetilde{\Phi}'_r$ in the $\ket{+}$ basis of an ancilla qubit and costs one query to controlled $W(r)$ and $W(r)^\dagger$ each. Then, the circuit for $\HAMT$ is simply the product of $J$ different $W_h(r_j)$, controlled on the state $\ket{j}$ in a register of $\log_2 J$ qubits, where $r_j = js/J$ for $j = 0, 1, \ldots, J-1$. We will apply \cref{prop:dyson_series_complexity} using $\HAMT$ and bound the total error to the ideal time-ordered exponential.

    In the notation of \cref{eq:HAMT_defn}, we have $A(r_j) = \theta\widetilde{\Phi}_{r_j}'$ so that each $W_h(r_j) = \BE[\theta\widetilde{\Phi}_{r_j}'; \theta\lambda, a + \lceil \log_2 N \rceil + 1, 0]$. Invoke \cref{prop:dyson_series_complexity} with $\HAMT$ to get a block encoding
    \begin{equation}
    S \coloneqq \BE\l[ \TO\exp\l( -i \int_0^s A(r) \, dr \r); 1, \O(a + \log(NJ) ), \delta \r]
    \end{equation}
    for some $\delta > 0$. The block encodings $W_h(r)$ have $\alpha \equiv \theta\lambda$ for all $r$, implying that $S$ costs
    \begin{equation}
        Q = \O\l(\theta\lambda s + \log\frac{1}{\delta}\r)
    \end{equation}
    queries to $\HAMT$. To estimate the Lipschitz constant in $\|A(r) - A(r')\| \leq L|r - r'|$, we first use Duhamel's formula for the ideal $\Phi_r$ (see~\cite{haah2021quantum} for example):
    \begin{equation}
        \frac{d}{dr} \l( e^{-itH_r} V e^{itH_r} \r) = -i \int_0^t dt' \, e^{-i(t-t')H_r} [V, e^{-it'H_r} V e^{it'H_r}] e^{i(t-t')H_r},
    \end{equation}
    whose norm is at most $2|t| \|V\|^2$. Thus if we integrate in $r$,
    \begin{equation}
        \|e^{-itH_r} V e^{itH_r} - e^{-itH_{r'}} V e^{itH_{r'}}\| \leq 2|t| \|V\|^2 |r - r'|.
    \end{equation}
    Hence
    \begin{equation}
        \|\Phi_r - \Phi_{r'}\| \leq 2\|V\|^2 |r-r'| \int_\R dt \, |t| f_\beta(t).
    \end{equation}
    This remaining integral in $t$ can be elegantly evaluated using the series representation for $\log\coth x$;~however, it suffices to use a computer algebra system to get the closed-form expression
    \begin{equation}
        \int_\R dt \, |t| f_\beta(t) = \frac{4}{\pi\beta} \int_0^\infty dt \, t \log\coth\l(\frac{\pi t}{2\beta}\r) = \frac{7\zeta(3)}{\pi^3} \beta,
    \end{equation}
    where $\zeta(3) \approx 1.202$ is Ap\'{e}ry's constant. This quantity $L_{\mathrm{ideal}} = \frac{14 \, \zeta(3)}{\pi^3} \beta \|V\|^2$ is the ideal Lipschitz constant;~we will need a robust version for the approximation $\widetilde{\Phi}_r'$.

    To do so, we define a piecewise-linear family $\overline{\Phi}_r$ such that $\overline{\Phi}_{r_j} = \widetilde{\Phi}_{r_j}'$ for all $j$, and otherwise linearly interpolates in between the quadrature nodes. (On the last interval $[r_{J-1}, s]$ we assume $\overline{\Phi}_r$ remains constant.) Since $\|\Phi_r - \Phi_{r'}\| \leq L_{\mathrm{ideal}} |r - r'|$, on the nodes we have
    \begin{equation}
        \l\| \widetilde{\Phi}_{r_{j+1}}' - \widetilde{\Phi}_{r_j}' \r\| \leq L_{\mathrm{ideal}} \frac{s}{J} + 2(\varepsilon + \gamma).
    \end{equation}
    Hence by interpolation,
    \begin{equation}
        \l\| \overline{\Phi}_r - \overline{\Phi}_{r'} \r\| \leq \l( L_{\mathrm{ideal}} + \frac{2J(\varepsilon + \gamma)}{s} \r) |r - r'|
    \end{equation}
    and so we can choose the Lipschitz constant for $A(r) = \theta \overline{\Phi}_r$ to be $L = \theta \l( L_{\mathrm{ideal}} + \frac{2J(\varepsilon + \gamma)}{s} \r)$. Note that we need to assume that the error parameters of $W(r)$ obey
    \begin{equation}
        \varepsilon + \gamma \leq \frac{\delta}{4\theta s},
    \end{equation}
    so that the condition $ \frac{Ls^2}{\delta} \leq \frac{\theta L_{\mathrm{ideal}} s^2}{\delta} + \frac{J}{2}$ holds, as required by \cref{eq:TD_J_requirement}. Thus $L = \O(\theta(\beta\|V\|^2 + J\delta/s^2))$, whence we may choose some $J$ that satisfies
    \begin{equation}
        J = \O\l( Q + \frac{\theta\beta(\|V\| s)^2}{\delta} + \frac{(\theta\lambda s)^{3/2}}{\sqrt{\delta}} \r).
    \end{equation}
    
    Now let us bound the error from the ideal unitary $S_{\mathrm{ideal}} \coloneqq \TO\exp\l( -i \theta \int_0^s dr \, \Phi_r \r)$. Let $\Pi$ be the projector onto the block in which $S$ encodes. Then
    \begin{equation}
    \begin{split}
        \|\Pi S \Pi - S_{\mathrm{ideal}}\| &\leq \l\| \TO\exp\l( -i \int_0^s A(r) \, dr \r) - \TO\exp\l( -i \theta \int_0^s dr \, \Phi_r \r) \r\| + \delta\\
        &\leq \theta s \max_{r \in [0, s]} \l\| \overline{\Phi}_r - \Phi_r \r\| + \delta\\
        &\leq \theta s (\varepsilon + \gamma) + \frac{\theta L_{\mathrm{ideal}} s^2}{J} + \delta\\
        &\leq \frac{\delta}{4} + \frac{\delta}{2} + \delta < 2\delta,
    \end{split}
    \end{equation}
    where the second inequality is a standard stability bound for time-dependent Hamiltonians (see~\cite{haah2021quantum} again). This implies that $S$ is in fact a block encoding of $S_{\mathrm{ideal}}$ with error at most $2\delta$.
\end{proof}

Finally, we can plug this block encoding of time-dependent Hamiltonian simulation into LCHS to get the complexity of block encoding the QBP operator $\eta_s$.

\begin{theorem}\label{thm:eta_s_cost}
    Fix $s \in (0, 1]$ and assume access to $B_r(t) = \BE[e^{-itH_r} V e^{itH_r}; \lambda, a, \gamma]$ for every $r \in [0, s]$ and $t$ as in \cref{lem:LCU_Phi_s} (e.g., $V \succeq 0$). Let $H_* = \max_{r \in [0, s]} \|H_r\|$, $\beta > 0$, and $\varepsilon \in (0, \O(1)]$. There exists a block encoding
    \[
        \BE\l[ \eta_s \equiv \TO\exp\l( -\frac{\beta}{2} \int_0^s dr \, \Phi_r \r); \Theta(1), \O\l( a + \log((1 + \beta H_*) \beta\lambda) + \log\log\frac{\lambda}{\varepsilon} \r), \varepsilon \r],
    \]
    constructed from
    \[
    \begin{aligned}
        \Ot\l( (1 + \beta H_*) \frac{\beta^3 \lambda^3}{\varepsilon} \r) \quad &\text{queries to } B_r(t), \text{ and}\\
        \Ot\l( \beta\lambda a + (1 + \beta H_*) \frac{\beta^3 \lambda^3}{\varepsilon} \r) \quad &\text{additional quantum gates},
    \end{aligned}
    \]
    provided that $\gamma \leq \frac{c \, \varepsilon}{\beta s \log \varepsilon^{-1}}$ for some sufficiently small constant $c > 0$.
\end{theorem}

\begin{proof}
    The setup is given by \cref{prop:optimal_lchs} with the ideal PSD generator $A(r) = \frac{\beta}{2} \Phi_r$. The only subtlety we have to keep track of is that we only have access to block encodings of approximations $A(r) = \frac{\beta}{2} \widetilde{\Phi}_r' \approx \frac{\beta}{2} \Phi_r$. Thus our strategy will be to determine the complexity of constructing $U_s = \TO\exp\l( -\int_0^s A(r) \, dr \r)$ and then bound the error $\|U_s - \eta_s\|$.

    Let $W_h^{(R)}(r) = \BE[R \cdot A(r); \frac{|R|\beta}{2} \lambda, a + \lceil \log_2 N \rceil + 1, 0]$ for $|R| \leq \O(\log\varepsilon_U^{-1})$, where $\varepsilon_U$ is the ultimate desired error on $U_s$. From \cref{lem:LCU_Phi_s}, every query to $W_h^{(R)}(r)$ costs $N$ queries to $B_r(t)$. The relevant costs of the real-time evolutions used in the LCHS block encoding come from the costs of time-dependent Hamiltonian simulation:
    \begin{itemize}
        \item \textbf{Number of queries:}~for $\delta$ error,
        \begin{equation}
        \begin{split}
            \mathcal{Q} &= \O\l( |R|\beta\lambda s + \log\frac{1}{\delta} \r) \times \O\l( \frac{|R|(\beta\|V\| s)^2}{\delta} + \frac{(|R|\beta\lambda s)^{3/2}}{\sqrt{\delta}} \r)\\
            &= \Ot\l( \frac{R^2 (\beta \lambda s)^3}{\delta} + \frac{(|R|\beta\lambda s)^{5/2}}{\sqrt{\delta}} \r) = \Ot\l( \frac{(|R|\beta\lambda s)^3}{\delta} \r)
        \end{split}
        \end{equation}
        from \cref{lem:TD_Phi}.

        \item \textbf{Number of gates:}
        \begin{equation}
            \mathcal{G} = \Ot\l(\l(|R|\beta\lambda s + \log\frac{1}{\delta}\r)(a + \log N)\r)
        \end{equation}
        from \cref{rem:TDHS_gates}.
    \end{itemize}
    Hence the cost to block encode $U_s$ is $\mathcal{Q}$ queries to $W_h^{(R)}(r)$ and $\O(\mathcal{G} + \mathcal{Q} \log(\beta\lambda s) \polylog \varepsilon_U^{-1})$ additional one- and two-qubit gates, by \cref{prop:optimal_lchs}. We get the final asymptotic complexities by recognizing that we can take:
    \begin{align}
        \alpha_U = \Theta(1), &\quad \delta = \Theta(\varepsilon_U),\\
        |R| = \O(\log \varepsilon_U^{-1}), &\quad N = \O\l( (1 + \beta H_*) \log\frac{\lambda}{\varepsilon_\Phi} \r),\\
        \text{and } &\quad s \leq 1.
    \end{align}
    Note that $(\lambda, a, \gamma)$ are the inherent block-encoding parameters for $V$, while $\varepsilon_\Phi + \gamma$ is the error between $\widetilde{\Phi}_r'$ and $\Phi_r$.

    Now we estimate the error between $U_s$ and $\eta_s$. First we bound the error between the constructed operator $\widetilde{U}_s$ and the operator $\widetilde{U}_s^{\mathrm{ideal}}$ that we would have constructed had we had access to $\Phi_r$ instead of $\widetilde{\Phi}_r'$. Recall that the LCHS method implements a linear combination:
    \begin{equation}
        \widetilde{U}_s = \sum_j c_j T_j, \quad \widetilde{U}_s^{\mathrm{ideal}} = \sum_j c_j T_j^{\mathrm{ideal}}
    \end{equation}
    where $T_j, T_j^{\mathrm{ideal}}$ are time-dependent Hamiltonian simulations of $\widetilde{\Phi}', \Phi$, respectively. Thus
    \begin{align}
        \l\| \widetilde{U}_s - \widetilde{U}_s^{\mathrm{ideal}} \r\| &\leq \sum_j |c_j| \cdot \l\| T_j - T_j^{\mathrm{ideal}} \r\| \notag\\
        &\leq \sum_j |c_j| \l( \l\| T_j - \TO\exp\l( -i \frac{R_j \beta}{2} \int_0^s dr \, \widetilde{\Phi}'_r \r) \r\| + \l\| T_j^{\mathrm{ideal}} - \TO\exp\l( -i \frac{R_j \beta}{2} \int_0^s dr \, \Phi_r \r) \r\| \r. \notag\\
        &\phantom{\leq \sum_j |c_j| \bigg(} + \l. \l\| \TO\exp\l( -i \frac{R_j \beta}{2} \int_0^s dr \, \widetilde{\Phi}'_r \r) - \TO\exp\l( -i \frac{R_j \beta}{2} \int_0^s dr \, \Phi_r \r) \r\| \r)\\
        &\leq \sum_j |c_j| \l( 2\delta + \frac{|R_j|\beta}{2} \int_0^s \|\widetilde{\Phi}_r' - \Phi_r \| \, dr \r) \notag\\
        &= \O\l( \delta + \max_j |R_j| \beta s (\varepsilon_\Phi + \gamma) \r) \notag
    \end{align}
    since $\sum_j |c_j| = \alpha_U = \Theta(1)$. We already have $\delta = \Theta(\varepsilon_U)$;~it remains to choose both
    \begin{equation}
        \varepsilon_\Phi, \gamma = \O\l( \frac{\varepsilon_U}{\beta s \log \varepsilon_U^{-1}} \r)
    \end{equation}
    to bound this error by $\O(\varepsilon_U)$.
    
    Finally, the error between $\widetilde{U}_s^{\mathrm{ideal}}$ and $\eta_s$ is $\varepsilon_U$ by construction of the LCHS implementation. Hence we get
    \begin{equation}
        \l\| \widetilde{U}_s - \eta_s \r\| \leq \l\| \widetilde{U}_s - \widetilde{U}_s^{\mathrm{ideal}} \r\| + \l\| \widetilde{U}_s^{\mathrm{ideal}} - \eta_s \r\| = \O(\delta) + \varepsilon_U = \O(\varepsilon_U).
    \end{equation}
    Choosing $\varepsilon = \Theta(\varepsilon_U)$ with a sufficiently small constant factor proves the claim.
\end{proof}

Let us now state our generic protocol for preparing quantum Gibbs states. We assume that $H = H_0 + V$ and that an efficient circuit to prepare a purification of the Gibbs state for $H_0$ at any temperature is available. The output is a pure state approximating a purification of $\rho^{(\beta)}(H_1)$ with controlled error. The algorithm can be easily modified to restrict to only mixed Gibbs states as inputs and outputs;~this can reduce the space overhead by $n$ qubits, but costs quadratically more rounds to successfully prepare the state (due to a lack of amplitude amplification). For this reason, we focus only on the purified framework here. Recall that the canonical purification of an $n$-qubit mixed state $\rho$ is the $2n$-qubit pure state
\begin{equation}
    \ket{\sqrt{\rho}} \coloneqq (\sqrt{\rho} \otimes \I) \ket{\Omega}, \quad \text{where } \ket{\Omega} \coloneqq \sum_{x=0}^{2^n-1} \ket{x}\ket{x}.
\end{equation}

\begin{theorem}\label{thm:QBP_prep_alg}
    Let $H_1 = H_0 + V$ and $\rho_s$ be the Gibbs state of $H_s$ at fixed inverse temperature $\beta \geq 0$. Suppose $H_0$ and $V$ are accessed as exact block encodings $U_0 = \BE[H_0; \lambda_0, a_0, 0]$ and $U_V = \BE[V; \lambda_V, a_V, 0]$, respectively. Also assume we have access to a quantum circuit $P_0$ that prepares the purification of $\rho_0$, i.e., $P_0 : \ket{0^p} \ket{0^{2n}} \mapsto \ket{0^p} \ket{\sqrt{\rho_0}}$ for some integer $p \geq 0$. Then for $\varepsilon \in (0, \O(1)]$ there is a quantum circuit $P_1$ with gate complexity and query complexity (to $P_0, U_0, U_V$, and their inverses) of
    \[
        \poly(e^{\beta\|V\|}, \beta, \lambda_{0}, \lambda_V, 1/\varepsilon),
    \]
    such that
    \begin{equation}
        \| P_1 \ket{0^{b}} \ket{0^{2n}} - \ket{0^{b}} \ket{\sqrt{\rho_1}} \| \leq \varepsilon
    \end{equation}
    where the size $b$ of the ancilla register is linear in $a_0, a_V$, and $p$, and at most logarithmic in all other parameters.
\end{theorem}

\begin{proof}
    The cases $\beta = 0$ or $V = 0$ are trivial, so we assume $\beta > 0$ and $\|V\| > 0$ throughout. Shift $V$ to be PSD by defining $\hat{V} \coloneqq V + \|V\| \I$. Also define $H_s' \coloneqq (H_0 + s\hat{V})/2$. Let $\eta_s$ be the corresponding QBP operator for this path $H_s'$ at inverse temperature $\beta$. Observe that the purified Gibbs state of $H_0$ is
    \begin{equation}
        \ket{\sqrt{\rho_0}} = \l(\sqrt{\rho_0} \otimes \I\r) \ket{\Omega},
    \end{equation}
    and so by the QBP identity (\cref{cor:QBP_density_matrix}) we have
    \begin{equation}
        \l(\eta_1 \sqrt{\rho_0} \eta_1^\dagger \otimes \I\r) \ket{\Omega} = c\ket{\sqrt{\rho_1}}, \quad \text{where } c \coloneqq e^{-\beta\|V\|/2}\sqrt{\frac{Z_1}{Z_0}}.
    \end{equation}
     Here, $Z_0 = \tr(e^{-\beta H_0})$ and $Z_1 = \tr(e^{-\beta (H_0 + V)})$ are the usual partition functions of the desired Gibbs states. At the same time, we can rewrite the left-hand side above as
    \begin{equation}
         \l(\eta_1 \sqrt{\rho_0} \eta_1^\dagger \otimes \I\r) \ket{\Omega} = \l(\eta_1 \sqrt{\rho_0} \otimes \eta_1^*\r) \ket{\Omega} = (\eta_1 \otimes \eta_1^*) \ket{\sqrt{\rho_0}}.
    \end{equation}
    Since $0 \preceq \hat{V} \preceq 2\|V\|\I$, the min--max principle (e.g., see~\cite{carlen2014remainder}) gives
    \begin{equation}
        e^{-2\beta\|V\|} Z_0 \leq \tr(e^{-\beta(H_0 + \hat{V})}) \leq Z_0
    \end{equation}
    and therefore $e^{-\beta \|V\|} \leq c \leq 1$.
    
    Now let $U$ be the $(\alpha, q, \delta)$-block encoding of $\eta_1$ from \cref{thm:eta_s_cost}. Write $\widetilde{\eta}_1$ for the $\delta$-approximation and apply:
    \begin{equation}\label{eq:UotU_on_rho0}
    \begin{split}
        (U \otimes U^*) \ket{0^{2q}} \ket{0^{p}} \ket{\sqrt{\rho_0}} &= \ket{0^{2q}} \ket{0^{p}} \frac{\widetilde{\eta}_1 \otimes \widetilde{\eta}_1^*}{\alpha^2} \ket{\sqrt{\rho_0}} + \ket{{\perp}}\\
        &= \ket{0^{2q}} \ket{0^{p}} \frac{1}{\alpha^2} \l( c\ket{\sqrt{\rho_1}} + \ket{\mathrm{err}} \r) + \ket{{\perp}}
    \end{split}
    \end{equation}
    where $\ket{{\perp}}$ is orthogonal to $\ket{0^{2q}}$. The state $\ket{\mathrm{err}}$ has norm at most
    \begin{equation}
    \begin{split}
        \| \ket{\mathrm{err}} \| &= \| (\eta_1 \otimes \eta_1^*) \ket{\sqrt{\rho_0}} - (\widetilde{\eta}_1 \otimes \widetilde{\eta}_1^*) \ket{\sqrt{\rho_0}} \|\\
        &\leq \| \eta_1 \otimes \eta_1^* - \widetilde{\eta}_1 \otimes \widetilde{\eta}_1^* \|\\
        &\leq (\|\eta_1\| + \|\widetilde{\eta}_1\|) \| \eta_1 - \widetilde{\eta}_1 \|\\
        &\leq (2 + \delta) \delta,
    \end{split}
    \end{equation}
    where the final line comes from the fact that $\hat{V} \succeq 0$, so $\Phi_s \succeq 0$ and hence $\|\eta_1\| \leq 1$. Let $\ket{w} \coloneqq c\ket{\sqrt{\rho_1}} + \ket{\mathrm{err}}$ be the state in the success branch of \cref{eq:UotU_on_rho0}. For $\delta = \O(c)$ sufficiently small, $\|\ket{w}\| = \Omega(c)$ and so
    \begin{equation}
        \l\| \frac{\ket{w}}{\|\ket{w}\|} - \ket{\sqrt{\rho_1}} \r\| = \O\l(\frac{\delta}{c}\r) = \O(e^{\beta\|V\|}\delta).
    \end{equation}

    Now apply fixed-point amplitude amplification~\cite{yoder2014fixed,gilyen2019quantum} flagged on the $\ket{0^{2q}}$ ancilla. The success-branch amplitude is $\omega \coloneqq \frac{1}{\alpha^2} \|\ket{w}\| = \Omega(e^{-\beta\|V\|})$, so we can amplify this branch up to amplitude at least $1 - \chi^2$ by making
    \begin{equation}
        L = \O\l(\frac{\log(1/\chi)}{\omega}\r) = \O\l( e^{\beta\|V\|} \log(1/\chi) \r)
    \end{equation}
    queries to controlled $U \otimes U^*$ and $P_0$ (and their inverses). The resulting amplified state $\ket{\Psi}$ therefore has error
    \begin{equation}
        \| \ket{\Psi} - \ket{0} \ket{0^{2q}} \ket{0^{p}} \ket{\sqrt{\rho_1}} \| \leq \chi + \|\ket{\mathrm{err}}\| \leq \chi + \O(e^{\beta\|V\|} \delta).
    \end{equation}
    To make this at most $\varepsilon$, we can choose $\chi = \Theta(\varepsilon)$ and $\delta = \Theta(e^{-\beta \|V\|} \varepsilon)$.
    
    The final complexity for preparing this state is then $\O(L)$ times the cost of each $\eta_1$ block encoding (\cref{thm:eta_s_cost}):
    \[
    \begin{aligned}
        Q = \Ot\l( e^{2\beta\|V\|} (\|H_0\| + \|V\|) \frac{\beta^4 \lambda_V^3}{\varepsilon} \r) \quad &\text{queries to } B_r(t), \text{ and}\\
        \Ot\l( Q + e^{\beta\|V\|} \beta\lambda_V a_V \r) \quad &\text{additional quantum gates}
    \end{aligned}
    \]
    where $\|H_0\| \leq \lambda_{0}$ and $\|V\| \leq \lambda_V$. By \cref{rem:B(t)-gates}, each $B_r(t)$ can be constructed with sufficiently small error using only $\poly(\lambda_0, \lambda_V, |t|)$ gates where $|t| = \Ot(\beta)$.
\end{proof}

\end{document}